\documentclass[11pt,a4paper]{article}
\usepackage[margin=1in]{geometry}
\usepackage{libertine}
\usepackage[T1]{fontenc}
\usepackage[utf8]{inputenc}
\usepackage{amsfonts,amssymb,amsmath,amsthm}
\usepackage{calc}
\usepackage{tikz}
\usetikzlibrary{calc}

\usepackage{url}
\usepackage{hyperref}
\usepackage{enumitem}
\usepackage[mathlines]{lineno}
\usepackage{xspace}
\usepackage{thmtools}

\DeclareMathOperator{\VC}{2VCdim}
\DeclareMathOperator{\tw}{tw}
\DeclareMathOperator{\td}{td}
\DeclareMathOperator{\cc}{cc}

\newtheorem{problem}{Problem}[section]
\newtheorem{theorem}{Theorem}[section]
\newtheorem{lemma}[theorem]{Lemma}
\newtheorem{claim}[theorem]{Claim}
\newtheorem{corollary}[theorem]{Corollary}

\theoremstyle{definition}
\newtheorem{definition}[theorem]{Definition}
\newcommand{\Oh}{\mathcal{O}}
\renewcommand{\varnothing}{\emptyset}

\def\C{\mathcal{C}}

\def\N{\mathbb{N}}
\def\T{\mathcal{T}}

\def\Sol{\mathsf{Sol}}

\newcommand{\CMSO}{\textsf{CMSO}$_2$\xspace}
\newcommand{\cmsotwo}{\CMSO}
\newcommand{\VV}{$\bigvee\mkern-15mu\bigvee$}

\newcommand{\ceil}[1]{\left\lceil #1 \right\rceil }

\usepackage{framed}
\usepackage{tabularx}

\newlength{\RoundedBoxWidth}
\newsavebox{\GrayRoundedBox}
\newenvironment{GrayBox}[1]%
   {\setlength{\RoundedBoxWidth}{.93\columnwidth}
    \def\boxheading{#1}
    \begin{lrbox}{\GrayRoundedBox}
       \begin{minipage}{\RoundedBoxWidth}}%
   {   \end{minipage}
    \end{lrbox}
    \begin{center}
    \begin{tikzpicture}%
       \node(Text)[draw=black!20,fill=white,rounded corners,inner sep=2ex,text width=\RoundedBoxWidth]
             {\usebox{\GrayRoundedBox}};
        \coordinate(x) at (current bounding box.north west);
        \node [draw=white,rectangle,inner sep=3pt,anchor=north west,fill=white]
        at ($(x)+(6pt,.75em)$) {\boxheading};
    \end{tikzpicture}
    \end{center}}

\newenvironment{defproblemx}[1]{\noindent\ignorespaces%
                                \FrameSep=6pt%
                                \parindent=0pt%
                \begin{GrayBox}{#1}%
                \begin{tabular*}{\columnwidth}{!{\extracolsep{\fill}}@{\hspace{.1em}} >{\itshape} p{1.5cm} p{0.86\columnwidth} @{}}%
            }{
                \end{tabular*}%
                \end{GrayBox}%
                \ignorespacesafterend
            }

\newcommand{\problemTask}[3]{%
  \begin{defproblemx}{#1}
    Input: & #2 \\
    Task: & #3
  \end{defproblemx}
}

\title{Sparse induced subgraphs in $P_7$-free graphs of bounded clique number}

\author{Maria Chudnovsky\thanks{Princeton University, United States Of America (\texttt{mchudnov@math.princeton.edu}).
Supported by NSF Grant DMS-2348219 and by AFOSR grant FA9550-22-1-0083.}
\and Jadwiga Czyżewska\thanks{University of Warsaw, Poland (\texttt{j.czyzewska@mimuw.edu.pl}).
Supported by Polish National Science Centre SONATA BIS-12 grant number 2022/46/E/ST6/00143.}
\and Kacper Kluk\thanks{University of Warsaw, Poland (\texttt{k.kluk@mimuw.edu.pl}).
Supported by Polish National Science Centre SONATA BIS-12 grant number 2022/46/E/ST6/00143.}
\and Marcin Pilipczuk\thanks{University of Warsaw, Poland (\texttt{m.pilipczuk@mimuw.edu.pl}).
Supported by Polish National Science Centre SONATA BIS-12 grant number 2022/46/E/ST6/00143.}
\and Paweł Rzążewski\thanks{Warsaw University of Technology \& University of Warsaw, Poland (\texttt{pawel.rzazewski@pw.edu.pl}). Supported by the European Research Council (ERC) under the European Union’s Horizon 2020 research and innovation programme grant agreement number 948057.}\footnotemark[2]}

\begin{document}
%\begin{titlepage}
\date{}
\maketitle

\begin{abstract}
    Many natural computational problems, including e.g. \textsc{Max Weight Independent Set}, \textsc{Feedback Vertex Set}, or \textsc{Vertex Planarization}, can be unified under an umbrella of finding the largest sparse induced subgraph, that satisfies some property definable in \textsf{CMSO}$_2$ logic.
    It is believed that each problem expressible with this formalism can be solved in polynomial time in graphs that exclude a fixed path as an induced subgraph.
    This belief is supported by the existence of a quasipolynomial-time algorithm by Gartland, Lokshtanov, Pilipczuk, Pilipczuk, and Rzążewski~[STOC 2021], and a recent polynomial-time algorithm for $P_6$-free graphs by Chudnovsky, McCarty, Pilipczuk, Pilipczuk, and Rzążewski~[SODA 2024].

    In this work we extend polynomial-time tractability of all such problems to $P_7$-free graphs of bounded clique number.      
\end{abstract}
%\def\thepage{}
%\thispagestyle{empty}
%\end{titlepage}

\section{Introduction}
When studying computationally hard problems, a natural question to investigate is  whether they become tractable when input instances are somehow ``well-structured''.
In particular, in algorithmic graph theory, we often study the complexity of certain (hard in general) problems, when restricted to specific graph classes.
Significant attention is given to classes that are \emph{hereditary}, i.e., closed under vertex deletion.

\paragraph{Potential maximal cliques.}
Arguably, the problem that is best studied in this context is \textsc{Max Weight Independent Set} (MWIS) in which we are given a vertex-weighted graph and we ask for a set of pairwise non-adjacent vertices of maximum possible weight.
Investigating the complexity of MWIS in restricted graphs classes led to discovering numerous new tools and techniques in algorithmic graph theory~\cite{GROTSCHEL1984325,DBLP:conf/stoc/GartlandLMPPR24,Edmonds_1965,DBLP:journals/jctb/DallardMS24}.
One of such general techniques is the framework of \emph{potential maximal cliques} by Bouchitt\'e and Todinca~\cite{BouchitteT02,Todinca}.
Intuitively speaking, a potential maximal clique (or PMC for short) in a graph $G$ is a bag of a ``reasonable'' tree decomposition of $G$ (see Section~\ref{sec:prelim} for a formal definition). The key contribution of Bouchitt\'e and Todinca~\cite{BouchitteT02,Todinca} is showing that:
\begin{enumerate}
    \item the family $\mathcal{F}$ of PMCs in a graph $G$ can be enumerated in time polynomial in $|V(G)|$ and $|\mathcal{F}|$, and
    \item given a family $\mathcal{F}$ containing \emph{all} PMCs of $G$ (and possibly some other sets as well), MWIS can be solved in time polynomial in $|V(G)|$ and $|\mathcal{F}|$ by mimicking natural dynamic programming on an appropriate (but unknown) tree decomposition.
\end{enumerate}

Consequently, MWIS admits a polynomial-time algorithm when restricted to any class  with polynomially many PMCs. (Here we say that a class $\mathcal{X}$ of graphs has polynomially many PMCs if there exists a polynomial $\mathsf{p}$,  such that the number of PMCs in any $n$-vertex graph in $\mathcal{X}$ is at most $\mathsf{p}(n)$.)

Later this framework was extended by Fomin, Todinca, and Villanger~\cite{FominTV15} to the problem of finding a large ``sparse'' (here meaning: of bounded treewidth) induced subgraph satisfying certain \CMSO formula $\psi$. (\CMSO stands for \emph{Counting Monadic Second Order} logic, which is a logic where one can use vertex/edge variables, vertex/edge set variables, quantifications over these variables, and standard propositional operands.) Formally, for a given integer $d$ and a fixed \CMSO formula $\psi$ with one free set variable, the 
$(\mathrm{tw} \leq d,\psi)$-MWIS problem (here `MWIS' stands for `max weight induced subgraph') is defined as follows.

\problemTask{$(\mathrm{tw} \leq d,\psi)$-MWIS}%
{A graph $G$ equipped with a weight function $\mathfrak{w}\colon V(G) \to \mathbb{Q}_+$.}
{Find a pair $(\Sol, X)$ such that
\begin{itemize}
\item $X \subseteq \Sol \subseteq V(G)$,
\item $G[\Sol]$ is of treewidth at most $d$,
\item $G[\Sol] \models \psi(X)$,
\item $X$ is of maximum weight subject to the conditions above,
\end{itemize}
or conclude that no such pair exists.}

As shown by Fomin, Todinca, and Villanger~\cite{FominTV15}, every problem expressible in this formalism can be solved in polynomial time on classes of graphs with polynomially many potential maximal cliques.

Note that if $d=0$ and $\psi$ is a formula satisfied by all sets, then $(\mathrm{tw} \leq d,\psi)$-MWIS is exactly MWIS.
It turns out that $(\mathrm{tw} \leq d,\psi)$-MWIS captures several other well known computational problems (see also the discussion in~\cite{FominTV15,DBLP:journals/corr/abs-2402-15834}), e.g.;
\begin{itemize}
    \item \textsc{Feedback Vertex Set} (one of Karp's original 21 \textsf{NP}-complete problems~\cite{Karp1972}), equivalent by complementation to finding an induced forest of maximum weight,
   \item \textsc{Even Cycle Transversal}~\cite{DBLP:journals/siamdm/PaesaniPR22,DBLP:conf/wg/MisraRRS12,DBLP:conf/iwpec/BergougnouxBBK20}, equivalent by complementation to finding an induced graph whose every block is an odd cycle,
   \item finding a largest weight induced subgraph whose every component is a cycle,
   \item finding a maximum number of pairwise disjoint and non-adjacent induced cycles.
\end{itemize}

\paragraph{MWIS in graphs excluding a long induced path.}
Despite all advantages of the Bouchitt\'e-Todinca framework, its applicability is somehow limited, as there are numerous natural hereditary graphs classes that \emph{do not} have polynomially many PMCs. Can we use at least some parts of the framework  outside its natural habitat?

A notorious open question in algorithmic graph theory is whether MWIS (and, more generally, $(\mathrm{tw} \leq d,\psi)$-MWIS) can be solved in polynomial time in graphs that exclude a fixed path as an induced subgraph.
It is believed that this question has an affirmative answer,
which is supported by the existence of a quasipolynomial-time algorithm of Gartland, Lokshtanov, Pilipczuk, Pilipczuk, and Rzążewski~\cite{DBLP:conf/focs/GartlandL20,DBLP:conf/sosa/PilipczukPR21,DBLP:conf/stoc/GartlandLPPR21}.
However, if it comes to polynomial-time algorithms, we know much less.

A polynomial-time algorithm for MWIS (as well as for many other problems) in $P_4$-free graphs (i.e., graphs that exclude a 4-vertex path as an induced subgraph; analogously we define $P_t$-free graphs for any $t$), was discovered already in 1980s~\cite{DBLP:journals/siamcomp/CorneilPS85}.
In today's terms we would say that these graphs have bounded clique-width and thus polynomial-time algorithms for many natural problems, including $(\mathrm{tw} \leq d,\psi)$-MWIS, follow from a celebrated theorem by Courcelle, Makowsky, and Rotics~\cite{DBLP:conf/wg/CourcelleMR98}.
However, already the $P_5$-free case proved to be quite challenging.
In 2014, Lokshtanov, Vatshelle, and Villanger~\cite{DBLP:conf/soda/LokshantovVV14} showed that \textsc{Max Weight Independent Set} admits a polynomial-time algorithm in $P_5$-free graphs --- by adapting the framework of Bouchitt\'e and Todinca. Let us emphasize that $P_5$-free graphs might have exponentially many PMCs, so the approach discussed above cannot be applied directly. Instead, Lokshtanov, Vatshelle, and Villanger proved that in polynomial time one can enumerate a family of \emph{some} PMCs, and this family is sufficient to solve MWIS via dynamic programming.

By extending this method (and adding a significant layer of technical complicacy), Grzesik, Klimo\v{s}ov\'a, Pilipczuk, and Pilipczuk~\cite{grzesik2020polynomialtime} managed to show that MWIS is polynomially-solvable in $P_6$-free graphs.

Some time later, Abrishami, Chudnovsky, Pilipczuk, Rzążewski, and Seymour~\cite{Tara} revisited the $P_5$-free case and introduced major twist to the method: they proved that in order to solve MWIS and its generalizations, one does not have to have PMCs exactly, but it is sufficient to enumerate their \emph{containers}: supersets that do not introduce any new vertices of the (unknown) optimum solution. They also proved that containers for \emph{all} PMCs in $P_5$-free graphs can be enumerated in polynomial time, circumventing the problem of exponentially many PMCs and significantly simplifying the approach of Lokshtanov, Vatshelle, and Villanger~\cite{DBLP:conf/soda/LokshantovVV14}.
This result yields a polynomial-time algorithm for $(\mathrm{tw} \leq d,\psi)$-MWIS in $P_5$-free graphs.

By further relaxing the notion of a container (to a \emph{carver}), Chudnovsky, McCarty, Pilipczuk, Pilipczuk, and Rzążewski~\cite{DBLP:conf/soda/ChudnovskyMPPR24} managed to extend polynomial-time solvability of $(\mathrm{tw} \leq d,\psi)$-MWIS to the class of $P_6$-free graphs.

\paragraph{Our contribution.}
In this work we push the boundary of tractability of $(\mathrm{tw} \leq d,\psi)$-MWIS in $P_t$-free graphs by showing the following result.

\begin{theorem}\label{thm:mainalgo}
    For every fixed integers $d,k$, and a \CMSO formula $\psi$,
    the $(\mathrm{tw} \leq d, \psi)$-MWIS problem can be solved in polynomial time
    in $P_7$-free graphs with clique number at most $k$.
\end{theorem}

Quite interestingly, in $P_t$-free graphs, the boundedness of treewidth is equivalent to the boundedness of degeneracy and to the boundedness of treedepth~\cite{DBLP:journals/jctb/BonamyBPRTW22} (see also Section~\ref{sec:prelim}).
Thus, Theorem~\ref{thm:mainalgo} yields a polynomial-time algorithm for problems like:
\begin{itemize}
\item \textsc{Vertex Planarization}~\cite{DBLP:conf/soda/JansenLS14,DBLP:journals/dam/Pilipczuk17}, which asks for a largest induced planar subgraph, or
\item for constant $k$, finding a largest set of vertices inducing a subgraph of maximum degree at most $k$ (see, e.g.,~\cite{zbMATH05818768}).
\end{itemize}

Let us remark that the idea of investigating $P_t$-free graphs of bounded clique number was considered before. 
Brettel, Horsfield, Munaro, and Paulusma~\cite{DBLP:journals/ipl/BrettellHMP22} showed that $P_5$-free graphs of bounded clique number (actually, a superclass of these graphs) have bounded \emph{mim-width}, which gives a uniform approach to solve many classic problems in this class of graphs~\cite{DBLP:conf/soda/BergougnouxDJ23}.
Pilipczuk and Rzążewski~\cite{pilipczuk2023polynomial} showed that $P_6$-free graphs of bounded clique number have polynomially many PMCs and thus the framework of Bouchitt\'e and Todinca can be applied here directly.
However, this is no longer true for $P_7$-free graphs (even bipartite).
On the other hand, Brandst\"adt and Mosca~\cite{DBLP:journals/dam/BrandstadtM18} provided a polynomial-time algorithm for MWIS in $P_7$-free triangle-free graphs. However, they used a rather ad-hoc approach that works well for MWIS, but gives little hope to generalize it to more complicated instances of $(\mathrm{tw} \leq d, \psi)$-MWIS.
Our work vastly extends the latter two results.

\subsection{Technical overview}

As mentioned, in $P_t$-free graphs the width parameters treewidth, treedepth, and degeneracy are functionally equivalent~\cite{DBLP:journals/jctb/BonamyBPRTW22}.
Furthermore, as discussed in~\cite{DBLP:conf/soda/ChudnovskyMPPR24}, the property of having treewidth, treedepth, or degeneracy at most $d$ can be expressed as a \CMSO formula of size
depending on $d$ only. Consequently, if we define problems $(\mathrm{td} \leq d,\psi)$-\textsc{MWIS} and $(\mathrm{deg} \leq d,\psi)$-\textsc{MWIS} analogously
as $(\mathrm{tw} \leq d,\psi)$-\textsc{MWIS}, but with respect to treedepth and degeneracy, these three formalisms describe the same family of problems, when restricted to $P_t$-free graphs. 

Following~\cite{DBLP:conf/soda/ChudnovskyMPPR24}, the treedepth formalism is the most handy; we henceforth work with $(\mathrm{td} \leq d, \psi)$-\textsc{MWIS}. 
Furthermore, it is more convenient to work not just with induced subgraphs of bounded treedepth, but with \emph{treedepth-$d$ structures}:
induced subgraphs of treedepth at most $d$ with a fixed elimination forest; formal definition can be found in Section~\ref{sec:prelim}. 

Let $\T$ be a treedepth-$d$ structure in $G$; think of $\T$ as the sought solution to the $(\mathrm{td} \leq d, \psi)$-\textsc{MWIS} problem in question. 
As proven in~\cite{FominTV15} (and cast onto the current setting in~\cite{DBLP:conf/soda/ChudnovskyMPPR24}), there exists a tree decomposition $(T,\beta)$ of $G$
whose every bag is either contained in $N[v]$ for a leaf $v$ of $\T$, or whose intersection with $\T$ is contained in a single root-to-leaf path of $\T$, excluding the leaf.
We call the latter bags \emph{$\T$-avoiding}. 

Furthermore, the framework of~\cite{FominTV15} shows how to solve $(\mathrm{td} \leq d,\psi)$-\textsc{MWIS} given a family $\mathcal{F}$ of subsets of $V(G)$
with the promise that all bags of such tree decomposition $(T,\beta)$ are in $\mathcal{F}$. 
In~\cite{Tara}, it is shown that it suffices for $\mathcal{F}$ to contain only \emph{containers of bounded defect} for bags in $(T,\beta)$: 
we require that there exists a universal constant $\widehat{d}$ such that for every $t \in V(T)$ there exists $A \in \mathcal{F}$ such that
the bag $\beta(t)$ is contained in $A$ and $A$ contains at most $\widehat{d}$ elements of $\T$. 

Note that if $v$ is a leaf of $\T$, then $N[v]$ may contain only ancestors of $v$ among the vertices of $\T$, so $N[v]$ is an excellent container
for any bag contained in $N[v]$. Thus, it suffices to provide containers for $\T$-avoiding bags. 
This is how~\cite{Tara} solved $(\mathrm{td} \leq d,\psi)$-\textsc{MWIS} in $P_5$-free graphs; by generalizing the definition of a container
to a carver, the same approach is applied to $P_6$-free graphs in~\cite{DBLP:conf/soda/ChudnovskyMPPR24}.

In fact, for a treedepth-$d$ structure $\T$, the tree decomposition $(T,\beta)$ of~\cite{FominTV15} is constructed as follows: 
it is shown that there exists a minimal chordal completion $F$ of $G$ (i.e., an inclusion-wise minimal set of edges to add to $G$ to make it chordal)
whose addition keeps $\T$ a treedepth-$d$ structure, and $(T,\beta)$ is any clique tree of $G+F$. 
A \emph{potential maximal clique} (PMC for short) is a maximal clique in $G+F$ for \emph{some} minimal chordal completion $F$ of $G$.
Both~\cite{Tara} and~\cite{DBLP:conf/soda/ChudnovskyMPPR24} actually provide containers/carvers for \emph{every} treedepth-$d$ structure in $G$ and 
an \emph{arbitrary} choice of a $\T$-avoiding PMC in $G$. 
Furthermore, \cite{DBLP:conf/soda/ChudnovskyMPPR24} provides an example of a family of $P_7$-free triangle-free graphs where such a statement is impossible
(i.e., any such family of containers/carvers would need to be of exponential size). 

However, the algorithmic framework of~\cite{FominTV15,Tara,DBLP:conf/soda/ChudnovskyMPPR24} requires only that the provided family $\mathcal{F}$ contains containers/carvers
for all bags of \emph{only one} ``good'' tree decomposition $(T,\beta)$,
and only for the sought solution $\T$.
This should be contrasted with all $\T$-avoiding PMCs, which is in some sense
the set of all reasonable $\T$-avoiding bags, and for all treedepth-$d$ structures $\T$. Thus, if we want to tackle $P_7$-free graphs of bounded clique number, we need
to significantly deviate from the path of~\cite{Tara,DBLP:conf/soda/ChudnovskyMPPR24} and either use the power of the choice of $(T,\beta)$ or the fact that we need to handle only $\T$ being
the optimum solution to the problem we are solving. 

The assumption of bounded clique number allows us to quickly reduce the case of finding a container for a PMC to finding a container for a minimal separator.
For a set $S \subseteq V(G)$, a connected component $A$ of $G-S$ is \emph{full to $S$} if $N(A) = S$; a set $S$ is a \emph{minimal separator} if 
it has at least two full components.
A classic fact about PMCs~\cite{Todinca} is that for every potential maximal clique $\Omega$ and for every connected component $D$
of $G-\Omega$, the set $N(D)$ is a minimal separator. 
By a classic result of Gy\'{a}rf\'{a}s, the class of $P_t$-free graphs is 
$\chi$-bounded, in particular, if $G$ is $P_t$-free and $\omega(G) \leq k$ (where $\omega(G)$ is the number of vertices in a largest clique in $G$),
 then $\chi(G) \leq (t-1)^{k-1}$. 
This, combined with a VC-dimension argument based on the classic result
of Ding, Seymour, and Winkler~\cite{DSW}, shows the following:
\begin{lemma}
    For every $t$ and $k$ there exists $c$ such that for every $P_t$-free graph $G$ with $\omega(G)\leq k$ and any~PMC $\Omega$ of $G$ one of the two following conditions hold:
    \begin{enumerate}
        \item there exists $v\in \Omega$ such that $\Omega = N[v]$, or
        \item there exists a~family $\mathcal{D}\subseteq \cc(G-\Omega)$ of size at most $c$ such that $\Omega=\bigcup_{D\in \mathcal{D}} N(D)$.
    \end{enumerate}
\end{lemma}
That is, except for a simple case $N[v] = \Omega$ for some $v \in \Omega$, 
the PMC $\Omega$ is a union of a constant number of minimal separators. 
Hence, it suffices to generate a family $\mathcal{F}'$ of containers for $\T$-avoiding
minimal separators (which are defined analogously to $\T$-avoiding bags) and take $\mathcal{F}$ to be the family of unions of all
tuples of at most $c$ elements of $\mathcal{F}'$. 

Let $G$ be a $P_7$-free graph with $\omega(G) \leq k$. 
We start as in~\cite{Tara,DBLP:conf/soda/ChudnovskyMPPR24}: let $\T$ be an arbitrary treedepth-$d$ structure in $G$ and let $S$ be a minimal separator in $G$ with full components $A$ and $B$
such that $S$ is $\T$-avoiding.
We want to design a polynomial-time algorithm that outputs a family $\mathcal{F}'$ of subsets of $V(G)$ that contains a container for $\Omega$.
The algorithm naturally does not know $\T$ nor $S$; the convenient and natural way of describing the algorithm as performing some nondeterministic 
guesses about $\T$ and $S$, with the goal of outputting a container for $S$ in the end. We succeed if the number of subcases coming from the guesses is
polynomial in the size of $G$ and the family $\mathcal{F}'$ consists of all containers generated by all possible runs of the algorithm. 

We almost succeed with this quest in Section~\ref{sec:approximate}. 
That is, we are able to guess an ``almost container'' $K$ for $S$: 
$K$ contains a constant number of vertices of $\T$
and we identified a set $\mathcal{D}$ of \emph{tricky} connected components of $G-K$
that are contained in $A \cup B \cup S$; all other connected components
of $G-K$ are contained in $A$, in $B$, or in some other connected component of $G-S$.
Every tricky component $D \in \mathcal{D}$ is somewhat simpler: 
we identify a subset $\emptyset \neq L(D) \subsetneq D$ such that 
if $\mathcal{C}(D)$ is the family of all connected components of
 $G[L(D)]$ and of $G[D \setminus L(D)]$,
 then every $C \in \mathcal{C}(D)$ is a module of $G[D]$. 
 (A set $A \subseteq V(G)$ is a module in $G$ if every vertex $v \in V(G) \setminus A$
 is either complete to $A$ or anti-complete to $A$; more on modules in Section~\ref{sec:prelim}.)

Observe now that every connected component $C \in \mathcal{C}(D)$
satisfies $\omega(G[C]) < \omega(G) \leq k$, as $D \setminus C$ contains a vertex
complete to $C$. Hence, we can recurse on $C$, understanding it fully
(more precisely, finding optimum partial solutions; the definition
 of this ``partial'' requires a lot of \CMSO mumbling). 

Furthermore, we observe that the quotient graph $G[D]/\mathcal{C}(D)$ is bipartite.
Thus, if we can understand how to solve $(\mathrm{td} \leq d, \psi)$-\textsc{MWIS}
on $P_7$-free \emph{bipartite} graphs, then we should be done: the partial
solutions from components $C \in \mathcal{C}(D)$, combined with the understanding
of the $P_7$-free bipartite graph $G[D]/\mathcal{C}(D)$ should give
all partial solutions of $G[D]$ for every $D \in \mathcal{D}$. 
This, in turn, should give an understanding on how $(\mathrm{td} \leq d,\psi)$-\textsc{MWIS}
behaves on $K \cup \bigcup \mathcal{D}$.
As $S \subseteq K \cup \bigcup \mathcal{D}$, this should suffice to solve
$(\mathrm{td} \leq d,\psi)$-\textsc{MWIS} on $G$ using $(K,\mathcal{D},L)$ as 
an object that plays the role of the container for $S$. 

Explaining properly all ``should''s from the previous paragraph is quite involved
and tedious
and done in Section~\ref{sec:wrapup}. In this overview, let us focus
on the more interesting case of solving $(\mathrm{td} \leq d,\psi)$-\textsc{MWIS}
in $P_7$-free bipartite graphs.
Here, we actually prove the following result in Section~\ref{sec:bipartite}.
\begin{theorem}\label{thm:intro:bip}
There exists an algorithm that, given a $P_7$-free bipartite graph $G$ and an integer $d > 0$,
runs in time $n^{2^{2^{\Oh(d^3)}}}$ and computes a family $\mathcal{C}$ of subsets
of $V(G)$ with the following guarantee: 
for every $\Sol \subseteq V(G)$ such that $G[\Sol]$ is of treedepth at most $d$,
there exists
a tree decomposition $(T,\beta)$ of $G$ such that
for every $t \in V(T)$ we have $\beta(t) \in \mathcal{C}$ and 
$|\beta(t) \cap \Sol| = 2^{2^{\Oh(d^3)}}$. 
\end{theorem}
We remark that Theorem~\ref{thm:intro:bip} is not needed if one just
wants to solve the \textsc{MWIS} problem, 
as this problem can be solved in general bipartite graphs using matching or flow techniques.
Thus, if one is interested only in finding a maximum-weight independent set, Section~\ref{sec:bipartite} can be omitted.

On the other hand, our work identified the bipartite case as an interesting
subcase in exploring tractability of $(\mathrm{td} \leq d, \psi)$-\textsc{MWIS}
in $P_t$-free graphs. To state a precise problem for future work, we propose
polynomial-time tractability of \textsc{Feedback Vertex Set} in $P_t$-free
bipartite graphs.
Meanwhile, Theorem~\ref{thm:intro:bip}
and its proof in Section~\ref{sec:bipartite} can be of independent interest. 

Let us also mention that if one is only interested in solving \textsc{Feedback Vertex Set} in $P_7$-free bipartite graphs, instead of using our  Theorem~\ref{thm:intro:bip}, one can follow a different approach.
Lozin and Zamaraev~\cite{DBLP:journals/ejc/LozinZ17} provided a deep analysis of the structure of bipartite $P_7$-free graphs which implies that in particular they have bounded \emph{mim-width} and the corresponding decomposition can be found efficiently. This yields a polynomial-time algorithm for \textsc{Feedback Vertex Set}~\cite{DBLP:conf/soda/BergougnouxDJ23}.
However, it it far for clear what other cases $(\mathrm{td} \leq d, \psi)$-\textsc{MWIS} can be solved efficiently in bounded-mim-width graphs.

The proof of Theorem~\ref{thm:intro:bip}
starts from the work of Kloks, Liu, and Poon~\cite{kloks2012feedback}, who showed that \emph{chordal bipartite} graphs have polynomial number
of PMCs, and thus $(\mathrm{td} \leq d,\psi)$-\textsc{MWIS} problem is solvable
in this graph class by the direct application of the PMC framework of Bouchitt\'e and Todinca. (A graph $G$ is \emph{chordal bipartite} if it is bipartite
and does not contain induced cycles longer than $4$.)
A $P_7$-free bipartite graph is \emph{almost} chordal bipartite: it can contain
six-vertex cycles. 

We would like to add some edges to the input graph $G$ so that it becomes chordal
bipartite. Let $C = c_1-c_2-\ldots-c_6-c_1$ be an induced six-vertex cycle in $G$;
we would like to add the edge $c_1c_4$ to keep $G$ bipartite but break $C$. 
We show that, if one chooses $C$ carefully, one can do it so that $G$
remains $P_7$-free. However, such an addition may break the sought solution
$\T$: if $c_1,c_4 \in V(\T)$ but are incomparable in the elimination forest $\T$,
then $\T$ is no longer a treedepth-$d$ structure in $G+\{c_1c_4\}$. 

We remedy this by a thorough investigation of the structure 
of the neighborhood of a six-vertex cycle in a $P_7$-free graph, loosely
inspired by~\cite{BonomoCMSSZ18}. On a high level, we show that there is a branching process
with polynomial number of outcomes that in some sense ``correctly'' completes
$G$ to a chordal bipartite graph, proving Theorem~\ref{thm:intro:bip}.  

To sum up, the proof of Theorem~\ref{thm:mainalgo} consists of three ingredients.
First, in Section~\ref{sec:approximate} we show the guesswork that leads
to a polynomial number of candidate ``almost containers'' $(K,\mathcal{D},L)$
for a minimal separator $S$ in the input graph $G$. 
Second, in Section~\ref{sec:bipartite} we focus on bipartite graphs
and prove Theorem~\ref{thm:intro:bip}.
Finally, in Section~\ref{sec:wrapup}, we show how these two tools combine with
the dynamic programming framework of~\cite{FominTV15,Tara,DBLP:conf/soda/ChudnovskyMPPR24} 
to prove Theorem~\ref{thm:mainalgo}. 

Let us emphasize that  even if we aim to solve \textsc{Feedback Vertex Set} in $P_7$-free graphs of bounded clique number and use our approach to reduce the problem to the bipartite case, the problem we need to solve 
the latter class is not just \textsc{Feedback Vertex Set}, but a certain extension variant considering \CMSO types.
Thus even in this case the boundedness of mim-width implied by the work of Lozin and Zamaraev~\cite{DBLP:journals/ejc/LozinZ17} is not sufficient as the known algorithm for bounded-mim-width graphs cannot handle such a problem.
Consequently, we need to use Theorem~\ref{thm:intro:bip}.

\section{Preliminaries} \label{sec:prelim}
For an integer $n$, by $[n]$ we denote the~set $\{1,2,\dots, n\}$.

%All graphs in this paper are finite and simple, i.e. they do not contain loops or multiple edges. 

The set of vertices of graph $G$ is denoted as $V(G)$ and the set of edges as $E(G)$. An~edge joining vertices $v$ and $u$ is denoted as $uv$.
By $N(v)$ we denote the~set of neighbors of the vertex $v$ (called its {\it open neighborhood}).
By $N[v]$ we denote $N(v)\cup \{v\}$ (called them{\it closed neighborhood}). 
Given a~set of vertices $S$, we define its open neighborhood $N(S)$ as $\left(\bigcup_{v\in S} N(v)\right)\setminus S$ and its closed neighborhood $N[S]$ as $\left(\bigcup_{v\in S} N[v]\right)$.

%A~graph $H$ is called a~{\it subgraph} of $G$ if it can be obtained from $G$ by vertex or edge deletions.\prz{is it used?}
A~graph $H$ is an {\it induced subgraph} of $G$ if it can be obtained by vertex deletions.
We do not distinguish between sets of vertices and subgraphs induced by them; in cases when it might cause a~confusion we write $G[A]$ for an subgraph induced by $A\subseteq V(G)$.
We use $G-A$ as a~shorthand for $G[V(G) \setminus A]$.

A~{\it class} of graphs $\C$ is a~set of graphs.
We say that a~graph $G$ is $H${\it -free} if $G$ does not contain any induced subgraph isomorphic to $H$. A~class $\C$ is $H${\it -free} if all graphs $G\in \C$ are $H$-free. By $P_t$ we denote the~path on $t$ vertices. A $P_4$-free graph is called a~{\it cograph}.  By $C_t$ we denote the cycle on $t$ vertices. The {\it length} of a path or a cycle is the number of edges in it.

Given a~graph $G$ and its subset of vertices $A$, a~set of connected components of $G-A$ is denoted as $\cc(G-A)$.

We say that a~vertex $v$ is {\it complete} to a~set of vertices $B$ if $v$ is adjacent to all vertices in $B$. Similarly, a~set of vertices $A$ is {\it complete} to a~set of vertices $B$ if every vertex of $A$ is complete to $B$.

We say that a vertex $v$ is {\it anticomplete} to a~set of vertices $B$ if $v$ has no neighbors in $B$. Similarly, a~set of vertices $A$ is {\it anticomplete} to a~set of vertices $B$ if $N(A)\cap B = \varnothing$.

A~{\it module} $M$ is a~non-empty subset of vertices of graph $G$ such that every vertex $v\in V(G)\setminus M$ is adjacent either to all vertices of $M$ or to none of them. A~module is called 
\begin{itemize}
    \item {\it strong} if for any other module $M'$ either $M\subseteq M'$, $M'\subseteq M$ or $M\cap M' = \varnothing$
    \item {\it proper} if $M$ is a~proper subset of $V(G)$
    \item {\it maximal} if it is proper and strong and is not contained in any other proper module 
\end{itemize}
Modules $M$ and $M'$ are called {\it adjacent} if each vertex of $M$ is complete to $M'$ and $M\cap M'=\varnothing$.

Note that for any graph on at least two vertices the family of maximal modules form a partition of $V(G)$.

Given a graph $G$ and the collection of vertex-disjoint modules $M$, the {\it quotient graph} $G/M$ is constructed by replacing each module of $M$ with a single vertex, with the same neighbors as the original vertices of the module. 

%For a~bipartite graph $G$ with parts $A$ and $B$ and set of edges $E$ between these parts we use a~notation $G=(A,B, E)$. For an integer $t$, we denote complete bipartite graph with parts of size $t$ as $K_{t,t}$.
%\Jadwiga{Nie wiem na ile ta notacja jest spójna z tym co póxniej Marcin używa w sekcji}
%\prz[inline]{żadna z nich nie występuje. tam jest: G is a bipartite graph with a fixed bipartition $V_1,V_2$, a biklika po prostu jest nazywana bikliką}

A set of pairwise non-adjacent vertices in a graph $G$ is called \emph{independent}, while a set of pairwise adjacent vertices is a called \emph{clique}.

A~subgraph $H$ of $G$, which is a~connected component in the complement of $G$, is called an {\it anticomponent}.

\subsection{PMCs and chordal completions}
A~graph $G$ is called {\it chordal} if all its induced cycles are of length $3$. Note that all induced subgraphs of a~chordal graph are also chordal. 

A~graph $H$ is called a~{\it supergraph} of $G$ if $V(H)=V(G)$ and $E(G)\subseteq E(H)$. A~graph $H$ is called a~{\it chordal completion} of $G$ if $H$ is a~supergraph of $G$ and $H$ is a~chordal graph. A~chordal completion $H$ of $G$ is called {\it minimal} if it does not contain any proper subgraph which is also a~chordal completion of $G$. Alternatively, we denote a~chordal completion $H$ of $G$ as $G+F$, where $F$ is a~set of non-edges in $G$: then $E(H)=E(G)\cup F$.

A~set of vertices $\Omega\subseteq V(G)$ is called a~{\it potential maximal clique} or a~{\it PMC} if there exists a~minimal chordal completion $H$  of $G$ in which $\Omega$ is a~maximal clique. 

Let us recall a~classic result characterizing potential maximal cliques.
\begin{lemma}[see \cite{Todinca}]\label{lem:PMC_basic}
Given a~graph $G$, a~set of vertices $\Omega$ is a~PMC if and only if the following conditions are fulfilled:
\begin{enumerate}
    \item For each connected component $D$ of $G-\Omega$, the~set $N(D)$ is a~proper subset of $\Omega$.
    \item If $uv$ is a~non-edge of $G$, with $u,v\in\Omega$, then there exists a~component $D$ of $G-\Omega$ such that $u,v\in N(D)$.
\end{enumerate}
\end{lemma}

We say that a~set of vertices $S$ is a~{\it $(u,v)$- separator} if $u$ and $v$ belong to distinct components in $G-S$. A~$(u,v)$-separator $S$ is called {\it minimal} if no proper subset of $S$ is a~$(u,v)$-separator. A~set of vertices $S$ is called a~{\it minimal separator} of graph $G$ if there exists a~pair of vertices $u$ and $v$ such that $S$ is a~minimal $(u,v)$-separator. 

Given a~set of vertices $S$, we say that a~component $D\in \cc(G-S)$ is a~{\it full component} of $S$ if $N(D)=S$. It is easy to prove that $S$ is a~minimal separator of $G$ if and only if it has at least two full components. 
It was proven (see~\cite{Todinca}) that if $\Omega$ is a~PMC in $G$, then for every component $D$ of $G-\Omega$ the set $N(D)$ is a~minimal separator with $D$ as a~full component and another full component containing $\Omega\setminus N(D)$.

Bouchitt\'e and Todinca~\cite{BouchitteT02,Todinca} showed a close relation between potential maximal cliques and minimal separators.

\begin{theorem}[see~\cite{BouchitteT02,Todinca}]\label{thm:minsep-pmc}
If $G$ is an $n$-vertex graph with $a$ minimal separators and $b$ potential maximal
cliques, then $b \leq n(a^2+a+1)$ and $a \leq nb$.
Furthermore, given a graph $G$, one can in time polynomial in the input and output the list of all its minimal separators and potential maximal cliques.
\end{theorem}

\subsection{Treedepth structures and treewidth}
The {\it rooted forest} $\T$ is a~forest in which each component has exactly one distinguished vertex, called a~{\it root}.
A~path in a~rooted forest is called {\it vertical} if it connects a~vertex and any of its ancestors.
Given a~vertex $v$, we define its {\it depth} as the number of vertices on the path connecting $v$ and the root (so the root has depth $1$).
The {\it height} of the rooted forest $\T$ is equal to the~maximum depth of a~vertex of $\T$.
By $\T^\alpha$ we denote a set of vertices of $\T$ of depth exactly $\alpha$.
We also write $\T^{>\alpha}$ for $\bigcup_{\alpha' > \alpha} \T^{\alpha'}$.

We say that vertices $u$ and $v$ are {\it $\T$-comparable} if they can be connected via~a~vertical path. Otherwise, we say that they are {\it $\T$-incomparable}. An {\it elimination forest} of graph $G$ is a~rooted forest $\T$ such that $V(\T)=V(G)$ and for each edge $uv\in E(G)$ vertices $u$ and $v$ are $\T$-comparable. We define the {\it treedepth} of graph $G$ as the minimum height of any possible elimination forest of graph $G$.

Given a~graph $G$ and an integer $d$, a~{\it treedepth-$d$ structure} (the notion first proposed in~\cite{DBLP:conf/soda/ChudnovskyMPPR24}) is a~rooted forest $\T$ of height at most $d$ such that $V(\T)$ is a~subset of $V(G)$ and $\T$ is an elimination forest of the~subgraph of $G$ induced by $V(\T)$. 
A~treedepth-$d$ structure $\T$ is a~\emph{substructure} of a~treedepth-$d$ structure $\T'$
if $\T$ is a~subgraph of $\T'$ (as a rooted forest) and every root of $\T$ is also a~root of $\T'$. 
A~treedepth-$d$ structure $\T$ is called {\it maximal} if there is no treedepth-$d$ structure $\T'$ such that $\T$ is a substructure of $\T'$ and $\T \neq \T'$. Equivalently, $\T$ is maximal if it cannot be extended by adding a~leaf preserving the bound on the height of $\T$.

To avoid notational clutter, we sometimes treat $\T$ as set of vertices, for example in expressions $|A \cap \T|$.

Given a~treedepth-$d$ structure $\T$ of graph $G$, a~chordal completion $G+F$ is {\it $\T$-aligned} if for any edge $uv\in F$ the following conditions hold:
\begin{enumerate}
    \item neither $u$ and $v$ is a vertex of depth $d$ in $\T$,
    \item if both $u$ and $v$ belong to $V(\T)$, then $u$ and $v$ are $\T$-comparable.
\end{enumerate}
The second condition guarantees that $\T$ is a~treedepth-$d$ structure in $G+F$. It was proven (see Lemma~2.11 in~\cite{DBLP:conf/soda/ChudnovskyMPPR24}) that for any treedepth-$d$ structure in graph $G$, there exists a~minimal chordal completion $G+F$ which is $\T$-aligned.  
We say that a~PMC is $\T${\it -avoiding} if it is a~maximal clique in some minimal chordal completion which is $\T$-aligned and does not contain any vertex of $\T$ of depth $d$. 

Given a~treedepth-$d$ structure $\T$, we say that a~set $S'$ is a~$\T${\it -container of defect }$f$ (or a~container of defect $f$ if $\T$ is known from the context) for a~set $S$ if $S\subseteq S'$ and $|S'\cap V(\T)| - |S\cap V(\T)| \leq f$. Note we have $S\cap V(\T)\subseteq S'\cap V(\T)$, so the defect $f$ describes how many additional vertices of $V(\T)$ belong to $S'$. 

Now we can cite the Lemma~2.12 from~\cite{DBLP:conf/soda/ChudnovskyMPPR24}, which shows how to deal with PMCs, which are $\T$-aligned and contain vertices of $\T$ of depth $d$:

\begin{lemma}[see Lemma 2.12 in \cite{DBLP:conf/soda/ChudnovskyMPPR24}]\label{lem:PMCs_with_deep_vertices}
For each positive integer $d$, there is a~polynomial-time algorithm which takes in
a~graph $G$ and returns a~collection $L \subseteq 2^{V(G)}$ such that for any maximal treedepth-$d$ structure $\T$ in $G$, any $\T$-aligned minimal chordal completion $G + F$ of $G$, and any maximal clique $\Omega$ of $G + F$ which contains a~depth-$d$ vertex of $\T$, $L$  contains a~set $\widetilde{\Omega}$ that is a~container for $\Omega$ of defect $0$, i.e., $\Omega\subseteq \widetilde{\Omega}$ and $\widetilde{\Omega} \cap \T = \Omega~\cap \T$.
\end{lemma}

In our paper we will also need Lemma~2.13 from~\cite{DBLP:conf/soda/ChudnovskyMPPR24}, which shows how to use the maximality of a~treedepth-$d$ structure:
\begin{lemma}[see Lemma 2.13 in \cite{DBLP:conf/soda/ChudnovskyMPPR24}]\label{lem:maximality_of_tdstructure}
    Let $G$ be a~graph, $d$ be a~positive integer, and $\T$ be a~maximal treedepth-$d$ structure in $G$. Then for any $\T$-avoiding potential maximal clique $\Omega$ of $G$, each vertex in $\Omega~\setminus  \T$ has a~neighbor in $\T - \Omega$.
\end{lemma}

A {\it tree decomposition} of a graph $G$ is a pair $\T=(T,\beta)$, where $T$ is a tree and $\beta$ is a function assigning each node $v\in V(T)$ a subset of vertices $\beta(v)\subseteq V(G)$, called a {\it bag} of $v$ such that the following conditions are fulfilled:
\begin{itemize}
    \item for each $u\in V(G)$ a set of nodes of $T$ whose bags contain $u$  induces a connected non-empty subtree of $T$;
    \item for each edge $u_1u_2\in E(G)$ there exists a node $v\in V(T)$ such that $\beta(v)$ contains $u_1, u_2$.
\end{itemize}
A {\it width} of a tree decomposition $\T$ is equal to $\max_{x\in V(T)}|\beta(x)|-1$. The minimum width over all tree decompositions of a graph $G$ is called the \emph{treewidth} of $G$ and is denoted as $\tw(G)$.

\subsection{MWIS}

In all problems of interest in this paper, the input graph $G$ comes with vertex
weights that are positive rational numbers. 
%It will be convenient to assume that every subset
%of vertices has different weight; as we do not investigate the exact polynomial running time
%bound of our algorithms, this can be ensured by numbering the vertices with integers
%$0,1,\ldots, |V(G)|-1$ and changing the weight $w$ of a vertex number $i$ to
%$w \cdot 2^{|V(G)|} + 2^i$. \prz{czy z tego korzystamy?}

For fixed integer $d \geq 0$ and \CMSO formula $\phi$ with one free vertex set
variable, the \textsc{$(\mathrm{td} \leq d,\phi)$-MWIS} problem
takes in input a vertex-weighted graph $G$ and asks 
for a pair $(\Sol, X)$ where $X \subseteq \Sol \subseteq V(G)$, $G[\Sol]$ is of treedepth at most $d$, and
$\phi(X)$ is satisfied in $G[\Sol]$, and, subject to the above,
the weight of $X$ is maximized (or a negative answer if such a pair does not exist.)

\problemTask{$(\mathrm{td} \leq d,\psi)$-MWIS}%
{A graph $G$ equipped with a weight function $\mathfrak{w}\colon V(G) \to \mathbb{Z}_+$.}%
{Find a pair $(\Sol, X)$ such that
\begin{itemize}
\item $X \subseteq \Sol \subseteq V(G)$,
\item $G[\Sol]$ is of treedepth at most $d$,
\item $G[\Sol] \models \psi(X)$,
\item $X$ is of maximum weight subject to the conditions above,
\end{itemize}
or conclude that no such pair exists.}

Similarly as in \cite{DBLP:conf/soda/ChudnovskyMPPR24}, for a solution $(\Sol, X)$ we will be mostly looking at some maximal treedepth-$d$ structure $\T$ that contains $\Sol$. 

%\prz[inline]{somehow we need to say how this relates to the treewidth problem}

\subsection{Structural lemmas}
In this section we present some structural lemmas about $P_7$-free graphs. Beforehand, we need the following notation. 

By $X_1-X_2-\dots-X_k$ we denote a path $v_1-v_2-\dots-v_k$, where $v_i$ is some vertex belonging to $X_i$ for each $i\in[k]$. If $X_i$ is singleton $\{x_i\}$ for some $i$, we write $-x_i-$ instead $-\{x_i\}-$.

% Given a~graph $G$, its vertex $v$ and a~subset of vertices $A\subseteq V(G)$, we denote by $vAAA$ a~path on four vertices with one endpoint in $v$ and all other vertices belonging to $A$. 
% \prz{potem jest notacja $v-A-A-A$. Jest też wariant, że np. $P$ jest ścieżką i mamy $v -P -A - w - A -A$ czy coś takiego. poza tym proponuję pisać "path of the form $v - A - A - A$"}
% \Jadwiga{Tę definicję trzeba jakoś zgrabnie przerobić}

\begin{lemma}\label{lem:p_four}
    Let $S$ be a~minimal separator in a~$P_7$-free graph $G$ and let $A$ and $B$ be the full components to $S$. There exists a~subset $Z\subseteq A$ such that $Z$ is a~cograph and $S\setminus N(Z)$ is complete to $B$.
\end{lemma}
\begin{proof}
    Let $F$ be a~set of vertices of $S$ which are complete to $B$. Let $Z$ be a~minimal connected subgraph of $A$ such that $S\setminus F \subseteq N(Z)$. 

    Let $Q$ be a~maximal induced path in $Z$ and let $v$ be its endpoint. Note that $v$ cannot be a~cutvertex of $Z$ -- otherwise we could extend our maximal path $Q$. By the minimality of $Z$ there exists a~vertex $w \in S\setminus F$ such that $N(w)\cap Z = \{v\}$. 
    Let $B_1$ be $N(w)\cap B$ and $B_2=B\setminus B_1$. As $S$ is the minimal separator and $B$ is a~full component to $S$, the set $B_1$ is non-empty. Since $w\notin F$, the set $B_2$ is non-empty. As $B$ is a~connected component, then there exist $x_1\in B_1$, $x_2\in B_2$, which are adjacent.

    Let $P=Q - w- x_1- x_2$. Note that $P$ is an induced path in $G$. We assumed that $G$ is $P_7$-free, thus $|P| < 7$, which implies $|Q|\leq 3$. Thus, $Z$ is a~cograph.
\end{proof}

\begin{lemma}[see Lemma~4.2 in \cite{grzesik2020polynomialtime}]\label{lem:module_lemma}
    Let $S$ be a~minimal separator in graph $G$ and let $A$ be a~full component to $S$ and $|A|\geq 2$. Let $p$ and $q$ be any two distinct vertices of $A$ belonging to different maximal proper adjacent strong modules of $G[A]$. Then, for each vertex $v\in S$ at least one of the following conditions is fulfilled:
    \begin{enumerate}
        \item there exists an induced path $v-A-A-A$
        \item we have $v\in N(p)\cup N(q)$
        \item  $\overline{G[A]}$ is disconnected and $N(v)\cap A$ consists of some connected components of $\overline{G[A]}$
    \end{enumerate}
\end{lemma}

\begin{corollary}[compare Lemma 8 in~\cite{pilipczuk2023polynomial}]\label{cor:X_A}
Let $S$ be a~minimal separator in graph $G$ and let $A$ and $B$ be full components to $S$. There exists a~set $X\subseteq A$ ($X\subseteq B$) of size at most $\omega(G)$ such that for each vertex $v$ belonging to $S\setminus N(X)$  there exists an~induced path $v-A-A-A$ ($v-B-B-B$).
\end{corollary}
\begin{proof}
If $|A|<2$, then we simply take $X=A$. 
Suppose then that $A$ consists of least two vertices. Let $Q$ be a set of such vertices $v\in S$ such that there exists an~induced path $v-A-A-A$. If the complement of $G[A]$ is disconnected, then we choose one vertex from each connected component of $\overline{G[A]}$ and create a set $X$ from the chosen vertices. If $\overline{G[A]}$ is connected, let us pick any two vertices belonging to different maximal proper adjacent modules of $G[A]$ and construct $X$ this way. In both cases we have $|X|\leq \omega(G)$. By Lemma~\ref{lem:module_lemma} for any vertex $v\in S\setminus N(X)$ we have $v\in Q$. We prove the statement for $B$ in symmetrical way.  
\end{proof}

We will also need the following lemma, which was proved in~\cite{grzesik2020polynomialtime}.
\begin{lemma}[see Lemma~4.1 in \cite{grzesik2020polynomialtime}]\label{lem:two_orders}
    Let $(X,\leq_1)$ and $(X,\leq_2)$ be two partial orders on the set $X$ such that any pair of elements of $X$ is comparable in $\leq_1$ or in $\leq_2$. Then, there exists an element $v$ such that for any element $x\in X$ we have $v \leq_1 x$ or $v \leq_2 x$.
\end{lemma}

We say that sets $A$ and $B$ are {\it comparable by inclusion} if $A\subseteq B$ or $B\subseteq A$.

\begin{lemma}\label{lem:ind_set_neighboring_cograph}
    Let $G$ be a~$P_7$-free graph and let $Z$ be an induced subgraph of $G$, which is a~connected cograph.  Let $I\subseteq V(G) \setminus N[Z]$ be an independent set and let $J\subseteq N(Z)\cap N(I)$ an independent set. Then there exist sets of vertices $Q_Z\subseteq Z$ and $Q_I\subseteq I$ such that $J\subseteq N(Q_Z)\cup N(Q_I)$ and $|Q_Z|, |Q_I| \leq (\omega(Z)+1)!$.
\end{lemma}
\begin{proof}
    Let $f\colon \N\to \N$ be defined recursively as $f(1) = 1$  and $f(k) \leq k(1+ f(k-1))$ for $k \geq 2$. A~straightforward induction shows that for every $k\in\N$
    we have $f(k) \leq (k+1)!$.
    
    We will prove the lemma~by the induction on $\omega(Z)$ and with the bound
    $|Q_I|, |Q_Z| \leq f(\omega(Z))$.
    For the base case, if $\omega(Z) = 1$, then $Z$ contains only one vertex $v$ and $J\subseteq N(v)$. Thus, the induction holds for $\omega(Z)=1$.

    Let us assume now that $\omega(Z) > 1$. Let $Z_1,\dots,Z_l$ be anticomponents of $Z$. Note that $l\leq \omega(Z)$. Since $Z$ is a~connected cograph with $\omega(Z) > 1$, we know that $l > 1$. For each $i\in [l]$ we choose an~arbitrary vertex $z_i\in Z_i$. Let $\{Z_i^1, Z_i^2, \dots, Z_i^{m_i}\}$ be a~set of connected components of $Z_i - z_i$.
    
    Let $J' = J\setminus N(\{z_1,\dots, z_l\})$. 
    We split now $J'$ depending on the neighborhood in $Z$. 
    Let $J'_i = (J'\cap N(Z_i))\setminus N(Z_1\cup  \dots \cup Z_{i-1})$. 
    
    Suppose that $m_i\geq 2$. Let us consider two partial orders on $J'_i$: 
    \begin{itemize}
        \item $\mathbf{\leq_1}$: $u \leq_1 v \Leftrightarrow N(u)\cap I \subseteq N(v)\cap I$
        \item $\mathbf{\leq_2}$: $u \leq_2 v \Leftrightarrow \{j\in [m_i]\mid u\in N  (Z_i^j)\} \subseteq \{j\in [m_i]\mid v\in N(Z_i^j)\}$
    \end{itemize}
    Suppose that there exist vertices $x, y\in J'_i$ which are incomparable in both $\leq_1$ and $\leq_2$. Then there exists vertices $x_i, y_i\in I$ such that $x_ix, y_iy \in E(G)$, but $xy_i, x_iy \notin E(G)$. 
    Similarly, there exist connected components $Z_i^{j_x}$ and $Z_i^{j_y}$ such that $x\in N(Z_i^{j_x})$ and $y\in N(Z_i^{j_y})$, but $y\notin N(Z_i^{j_x})$ and $x\notin N(Z_i^{j_y})$. 
    So there exist induced paths $P_x$ and $P_y$ connecting $x$ and respectively $y$ with $z_{i'}$, where $i'\in [l]\setminus\{i\}$ such that $P_x\subseteq V(Z_i^{j_x})$ and $P_y \subseteq V(Z_i^{j_y})$. Since $x,y \notin N(z_{i'})$, each of $P_x$ and $P_y$ must contain at least one vertex. Then a~path $x_i - x - P_x - z_{i'} - P_y - y-y_i$ is an induced path of at least $7$ vertices, a~contradiction.
    
    Thus each pair of elements of $J'_i$ is comparable in $\leq_1$ or $\leq_2$. Then by Lemma~\ref{lem:two_orders} there exists a~vertex $w_i\in J'_i$ such that any other vertex $v\in J'_i$ we have $w_i\leq_1 v$ or $w_i\leq_2 v$. 
    Let $u^i_{w_i}$ be any vertex in $N(w_i)\cap I$ -- such vertex exists because $J\subseteq N(I)$. Let $j_{w_i}$ be any index such that $N({w_i})\cap Z_i^{j_{w_i}} \neq \varnothing$ (such component exists, because of the definition of $J_i$ and the assumption $m_i\geq 2$). Then $J'_i\subseteq N(u_{w_i}^i)\cup N(Z_i^{j_{w_i}})$. 
    Since $Z_i$ is the anticomponent, $\omega(Z_i^{j_{w_i}})\leq \omega(Z_i) \leq \omega(Z) - 1$. 
    By the inductive assumption, applied to a~connected cograph $Z_i^{j_{w_i}}$, $J'_i\setminus N(u_{w_i}^i)$ and $I$, there exist sets $Q_Z^i$ and $Q_I^i$ of size at most $f(\omega(Z)-1)$ such that $J'_i\setminus N(u_{w_i}^i) \subseteq N(Q_Z^i) \cup N(Q_I^i)$. 

    We are left with cases $m_i=0$ and $m_i=1$. In the former, $Z_i=\{z_i\}$ and $J'_i=\varnothing$. We take then $Q_Z^i=Q_I^i=\varnothing$. If $m_i = 1$, then $J'_i\subseteq N(Z_i^1)$. We can use directly inductive assumption to $Z_i^1$, $J'_i$ and $I$, getting sets $Q_Z^i$ and $Q_I^i$ such that $J_i\subseteq N(Q_Z^i)\cup N(Q_I^i)$.
    % \Jadwiga{Czy jest jakiś ładny trik tutaj do $u_w^i$?} Marcin: dekretuje olanie sprawy
    
    Let $Q_Z = \{z_1, \dots, z_l\} \cup \bigcup_{i\in [l]} Q_Z^i$ and $Q_I = \{u_{w_1}^1, \dots, u_{w_l}^l\} \cup \bigcup_{i\in [l]} Q_I^i$. Then 
    \[ |Q_Z|, |Q_I|\leq \omega(Z) + \omega(Z) \cdot f(\omega(Z) - 1) = f(\omega(Z)). \]
    We also have $$J \subseteq N(\{z_1, \dots, z_l\})\cup J' = N(\{z_1, \dots, z_l\}) \cup \bigcup_{i\in [l]} J'_i \subseteq N(Q_Z) \cup N(Q_I).$$
\end{proof}

\subsection{$\chi$-boundness}

By $\omega(G)$ we denote the {\it clique number} of $G$, i.e. a~size of the maximal clique being a~subgraph of $G$.  By $\chi(G)$ we denote the {\it chromatic number} of graph $G$, i.e. the minimal number of colors required to properly color vertices of graph $G$ (so no two adjacent vertices get the same color assigned). For any graph $G$ we have $\omega(G)\leq \chi(G)$. There exist graphs with large girth and large chromatic number, so we cannot hope for bound $\chi(G)\leq f(\omega(G))$ for any graph $G$. This leads us to the notion of {\it $\chi$-boundness}, which was discussed by Gyárfás in~\cite{chiBoundness}. We say that a~class $\C$ is $\chi${\it -bounded} if there exists a~function $f\colon \N\to\N$ such that for any graph $G\in\C$ we have $\chi(G)\leq f(\omega(G))$. 

Gyárfás proved in~\cite{chiBoundness} the following theorem:
\begin{theorem}[Theorem 2.4 in~\cite{chiBoundness}]\label{thm:Gyarfas_paths}
    Class of graphs excluding $P_t$ is $\chi$-bounded and a~function $f_n(x)= (t-1)^{x-1}$ is a~suitable bounding function. 
\end{theorem}

Proof of this theorem, presented in~\cite{chiBoundness}, is constructive and it shows that a~$P_t$-free graph $G$ can be colored with $(t-1)^{\omega(G)-1}$ colors (i.e. divided into color classes) in polynomial time.

\subsection{2VC-dimension}

The notion of  Vapnik-Chervonenkis dimension was introduced by Vapnik and Chervonenkis in~\cite{VC}. In graphs it is usually applied with the set of vertices as a~universe and the family of closed neighborhoods as family of sets. Here, we will use VC-dimension in a~slightly different setting. 

Let $(U, \mathcal{F})$ be a~{\it set system}: let $U$ be a~set and $\mathcal{F}$ be a~family of subsets of $U$. The {\it 2VC-dimension} of a~set system $(U,\mathcal{F})$, denoted as $\VC(U, \mathcal{F})$ is a~maximal size of a~set $X\subseteq U$ such that for every subset $Y\subseteq X$ of size 2 there exists $A\in\mathcal{F}$ such that $A\cap X = Y$. 

For a~set system $(U, \mathcal{F})$, we define a {\it dual set system} as $(\mathcal{F}, \widehat{U})$, where $\widehat{U} = \{x^*\mid x\in U\}$ and $x^* = \{F\in\mathcal{F}\mid x\in F\}$. Then a~{\it dual 2VC-dimension} of set system $(U,\mathcal{F})$, denoted as $\VC^*(U,\mathcal{F})$,  is a~2VC-dimension of its dual system, i.e. the size of a~maximal $\mathcal{F}'\subseteq \mathcal{F}$ such that for every distinct $A,B\in \mathcal{F}'$ there exists $u\in U$ such that $\{C\in \mathcal{F}'\mid u\in C\}=\{A,B\}$. 

Given a~set system $(U, \mathcal{F})$, let $\nu(U, \mathcal{F})$ be a~maximal size of a~subfamily $\mathcal{F}'\subseteq \mathcal{F}$ of pairwise disjoint sets and let $\tau(U, \mathcal{F})$ be a~minimal size of a~set $A\subseteq U$ such that for every $F\in \mathcal{F}$ we have $F\cap A\neq \varnothing$.

Ding, Seymour and Winkler proved the lemma~in~\cite{DSW}, bounding $\tau$ as a~function of $\lambda$ and~$\nu$. They used the language of hypergraphs, which we translated here to set systems. 
\begin{lemma}[compare (1.1) in \cite{DSW}]\label{lem:DSW}
    For any set system $\mathcal{S}=(U, \mathcal{F})$ we have $$\tau(\mathcal{S})\leq 11 \cdot \VC^*(\mathcal{S})^2\left(\VC^*(\mathcal{S}) + \nu(\mathcal{S}) + 3\right){{\VC^*(\mathcal{S}) + \nu(\mathcal{S})}\choose{\nu(\mathcal{S})}}^2.$$
\end{lemma}

\begin{corollary}
    For every $t$ there exists $c$ such that for every $P_t$-free graph $G$, a~PMC $\Omega$ in $G$ and an independent set $I\subseteq \Omega$ such that for every vertex $v\in I$ there exists a~component $D\in \cc(G-\Omega)$ such that $v\in N(D)$, there exists a~family $\mathcal{D}\subseteq \cc(G - \Omega)$ of size at most $c$ such that $I \subseteq \bigcup_{D\in\mathcal{D}} N(D)$. 
\end{corollary}
\begin{proof}
    Let us consider a~set system $\mathcal{S} = (\cc(G-\Omega), \{\mathcal{D}_v\mid v \in I\})$, where $\mathcal{D}_v = \{D\in \cc(G- \Omega)\mid v \in N(D) \}$. 
    By Lemma~\ref{lem:PMC_basic} we know that for every two non-adjacent vertices $u, v\in \Omega$ there exists a~component $D\in \cc(G-\Omega)$ such that $u,v\in N(D)$. Hence, $\nu(\mathcal{S}) = 1$, as $I$ is an independent set. Let us denote $\ceil{\tfrac{t}2}+1$ as $p$. We will show now that $\VC^*(\mathcal{S})< p$. 
    Suppose otherwise -- then there exists a~subfamily $\mathcal{F}'=\{\mathcal{D}_{v_1},\mathcal{D}_{v_2}, \dots, \mathcal{D}_{v_{p}}\}$ such that for every $j,k\in\left[\tfrac{t}2\right]$ there exists a~component $D\in \cc(G-\Omega)$ such that $\{\mathcal{D}_i\in \mathcal{F}'\mid D\in \mathcal{D}_i\}=\{\mathcal{D}_{v_i}, \mathcal{D}_{v_j}\}$. 
    In other words, for every two distinct vertices $u,w\in \{v_1, \dots, v_{p}\}\subseteq I$ there exists a~component $D\in \cc(G-\Omega)$ such that $u, w\in N(D)$, but no other $v_i\in N(D)$ and therefore a~minimal path $P_{uw}\subseteq D$ connecting $u$ and $w$ (note that $P_{uw}$ may contain just one vertex). 
    Let us consider a~path $v_1 - P_{v_1v_2} - v_2 - P_{v_2v_3} - v_3 - \dots - P_{v_{p-1} v_{p}} - v_{p}$. Note that it is an induced path in $G$ of length $p + p - 1 \geq t$, which is a~contradiction. Thus, $\VC^*(\mathcal{S}) < p$.

    Note that $\tau(\mathcal{S})$ is a~minimal size of a~set $\mathcal{D} \subseteq \cc(G-\Omega)$ such that for every $v\in I$ there exists $D\in \mathcal{D}$ such that $v\in N(D)$. By Lemma~\ref{lem:DSW} we get that $\tau(S)$ is bounded by a~function of $p$, so a~function of $t$.     
\end{proof}

We immediately get the following corollary.

\begin{corollary}\label{cor:pmc-cover}
    For every $t$ and $k$ there exists $c$ such that for every $P_t$-free graph $G$ with $\omega(G)\leq k$ and a~PMC $\Omega$ one of the two following conditions hold:
    \begin{enumerate}
        \item there exists $v\in \Omega$ such that for every component $D\in \cc(G-\Omega)$ we have $v\notin N(D)$
        \item there exists a~family $\mathcal{D}\subseteq \cc(G-\Omega)$ of size at most $c$ such that $\Omega=\bigcup_{D\in \mathcal{D}} N(D)$
    \end{enumerate}
\end{corollary}

\section{Approximating minimal separators}\label{sec:approximate}

This section is devoted to the proof of the following theorem.
\begin{theorem}\label{thm:to-bip}
Let $d,k \geq 1$ be fixed integers.
Given a $P_7$-free graph $G$ with $\omega(G) \leq k$, 
one can in polynomial time
compute a family
$\mathcal{F}$ with the following guarantee:
for every maximal treedepth-$d$ structure $\T$
and every $\T$-avoiding minimal separator $S$
in $G$ with full components $A$ and $B$, there
exists $(K,\mathcal{D},L) \in \mathcal{F}$
satisfying the following:
\begin{itemize}
    \item $K \subseteq V(G)$, $S \cap \T \subseteq K$, and $|K \cap \T| = \Oh_{k,d}(1)$;
    \item $\mathcal{D}$ is a subset of the family of connected components of $G-K$;
    \item $L$ is a function that assigns to every $D \in \mathcal{D}$
     a subset $\emptyset \neq L(D) \subsetneq D$ such that
     every connected component of either $G[L(D)]$ or of $G[D \setminus L(D)]$
     is a module of $G[D]$;
    \item for every connected component $D$ of $G-K$, one of the following holds:
    \begin{itemize}
        \item $D$ is a subset of a single connected component of $G-S$; or
        \item $D \in \mathcal{D}$.
    \end{itemize}
\end{itemize}
\end{theorem}

\begin{proof}
Let $G$ be as in the theorem statement. 
Using Theorem~\ref{thm:Gyarfas_paths}, we compute a coloring of $G$ into $c \leq 6^{k-1}$ color
classes and denote them by $V^1, \ldots, V^c$. 
In what follows, we will frequently divide various subsets of $V(G)$ into color classes;
for a set $X$, we usually denote $X \cap V^i$ by $X^i$ for brevity
and call the sets $X^1,\ldots,X^c$ the color classes of $X$. 

Let $\T$ be a~maximal treedepth-$d$ structure in $G$. Let $S$ be a~$\T$-avoiding minimal separator in graph $G$ and let $A$ and $B$ be the full components to $S$. 
We follow here a way from presentation from \cite{Tara,DBLP:conf/soda/ChudnovskyMPPR24}: we want to guess
the tuple $(K,\mathcal{D},L)$ for $(S,A,B)$ by making a guesswork that
ends up in a polynomial number of options. 
That is, we describe a nondeterministic polynomial-time algorithm that
\begin{enumerate}
\item has a polynomial number of possible runs, 
\item in every run, the algorithm either outputs a tuple $(K, \mathcal{D}, L)$ as in the lemma statement
or terminates without outputting any tuple, and
\item for every $\T$, $S$, $A$, and $B$ as above, there exists a run of the 
algorithm that produces a tuple $(K,\mathcal{D},L)$ fitting $\T$, $S$, $A$, and $B$
as in the lemma statement. 
\end{enumerate}
To compute the desired family $\mathcal{F}$, we iterate over all possible runs of the algorithm
and collect all output tuples $(K,\mathcal{D},L)$.
The nondeterministic algorithm does not know $\T$, $S$, $A$, nor $B$, but is able to 
nondeterministically guess some properties of them; however, in such guesses we need to
control the number of possible runs (branches) so that their total number is polynomial. 
For example, while we cannot guess the entire set $S$, we can for example guess 
one leaf of $\T$ that is contained in $A$ (there are $n$ options) or a specific vertex
of $S$. 
%\textcolor{red}{Note that while guessing, we do not know $A$, $B$, $S$ or $\T$. However, we will argue that there exist proper guesses which correspond to certain objects defined by these sets, e.g. a leaf of $\T$ or a subset of $A$ of size $m$ (clearly, if we guess such object, i.e. branch into all possible subsets of vertices of $G$ of size $m$, then one of them is indeed a subset of $A$ fulfilling our conditions). After guessing, we treat an object as known and we can refer to it, e.g. by considering its neighbourhood in the graph.}

Let $X_A$ be a~subset of $A$ of size at most $k$ such that for each vertex $v\in S\setminus N(X_A)$ there exists an~induced path $v-A-A-A$. Analogously, we define $X_B$ as a~subset of $B$ of size at most $k$ such that for every vertex $v\in S\setminus N(X_B)$ there exists an~induced path $v-B-B-B$. Such subsets exists by Corollary~\ref{cor:X_A}. 

Suppose that there exists $v\in S\setminus (N(X_A)\cup N(X_B))$. Then $v$ belongs to $S\setminus N(X_B)$ and $S\setminus N(X_A)$, so there exists induced paths $v-B-B-B$ and $v-A-A-A$. By joining them, we create a $P_7$ which is a~contradiction. Thus, $S\subseteq N(X_A)\cup N(X_B)$. Moreover, note that $N(X_A)\cap N(X_B)\subseteq S$.

We will maintain a set $\widetilde{S}$ which will serve as our current candidate for $K$.
We initiate $\widetilde{S}=\varnothing$. We will refer to $\widetilde{S}$ as ``our container'' for the set $S$. 

We (nondeterministically) guess sets $X_A$, $X_B$ and leafs $p_A$ and $p_B$ of $\T$, belonging to respectively $A$ and $B$. By Lemma~\ref{lem:maximality_of_tdstructure} such vertices exists, as we are guessing vertices of maximal depth in $\T\cap A$ and $\T\cap B$. We also know that there exists suitable guesses for $X_A$ and $X_B$, fulfilling the above definitions.
Since $|X_A|,|X_B| \leq k$, there is a polynomial number of options for this guess.

Then we add to $\widetilde{S}$ vertices of $N(X_A\cup\{p_A\})\cap N(X_B\cup \{p_B\})$ and $N(p_A)\cup N(p_B)$. We also guess $\T\cap S$ and add them to $\widetilde{S}$. Since $S$ is $\T$-avoiding, $|\T \cap S| \leq d-1$.

Note that we added to $\widetilde{S}$ at most $d-1$ ancestors of $p_A$ among vertices in $N(p_A)$ and at most $d-1$ ancestors of $p_B$ among vertices in $N(p_B)$ in the treedepth-$d$ structure.

Let $S_A=S\setminus N(X_A)$ and $S_B=S\setminus N(X_B)$. Note that $S_A$ and $S_B$ are subsets of respectively $N(X_B)$ and $N(X_A)$.

We now focus on vertices of $S$, which have neighbors in connected components others than $A$ and $B$. Let $\widehat{S}=N(V(G) \setminus (A\cup B\cup S))$,
note that $\widehat{S} \subseteq S$, and let $\widehat{S_A}=\widehat{S}\cap S_A$ and $\widehat{S_B}=\widehat{S}\cap S_B$. We divide $\widehat{S_A}$ and $\widehat{S_B}$ into color classes: $\widehat{S_A^1},\widehat{S_A^2}, \dots, \widehat{S_A^c}$ and  $\widehat{S_B^1},\widehat{S_B^2}, \dots, \widehat{S_B^c}$. We need the following lemma~in order to put these vertices into our container. 

\begin{claim}\label{cl:other_components}
    For any $i \in [c]$, if $\widehat{S_A^i} \neq \emptyset$,
    then there exists a~component $D_A^i\in \cc(G-(A\cup S\cup B))$ and vertices $a_A^i,b_A^i,c_A^i\in A$ such that a~set of vertices $\widehat{S_A^i}$ is a~subset of $N(D_A^i)\cup (N(\{a_A^i, b_A^i, c_A^i\})\cap N(X_B))$. A symmetrical statement holds for $B$. 
\end{claim}
\begin{proof}
    Let $v$ be a~vertex in $\widehat{S_A^i}$ which has the neighbors in the least number of connected components of $G-(A\cup S\cup B)$. 
    Then there exists an~induced path $v-A-A-A$ -- let $a_A^i, b_A^i, c_A^i$ be vertices of $A$ such that $v-a_A^i-b_A^i-c_A^i$ is an induced path. 
    Let $D_A^i$ be any connected component of $G-(A\cup S\cup B)$ such that $v\in N(D_A^i)$. 
    
    Suppose that there exists a~vertex $u\in \widehat{S_A^i}$ which does not belong to a~set $N(D_A^i)\cup (N(\{a_A^i, b_A^i, c_A^i\}) \linebreak \cap N(X_B))$. As $\widehat{S_A^i}$ is a~color class, $u$ and $v$ are non-adjacent. 
    However, we can connect them via~$B$, as each vertex in $S_A$ has a~neighbor in $X_B$ and $B$ is a~connected component. Let $Q\subseteq B$ be a~path (maybe on just one vertex) such that $u-Q-v$ is an induced path in $G$.
    By the minimality of $v$, we know that there exists a~connected component $D'\in \cc(G-(A\cup S\cup B))$ and a~vertex $w\in D'$ such that $v\notin N(D')$ and $u$ and $w$ are adjacent. Then $w-u-Q-v-a_A^i-b_A^i-c_A^i$ is an induced path on at least seven vertices; a~contradiction. Thus, $\widehat{S_A^i}\subseteq N(D_A^i)\cup (N(\{a_A^i, b_A^i, c_A^i\})\cap N(X_B))$.

    The proof for $B$ is symmetrical.
\end{proof}

Note that each component $D\in \cc(G-(A\cup S\cup B))$ is a~connected component of $G-S$, different from $A$ and $B$. Since $S\subseteq N(X_A)\cup N(X_B)\subseteq A\cup S\cup B$, then $D\in \cc(G-(N(X_A)\cup N(X_B)))$. Therefore each connected component of $G-(A \cup S \cup B)$ can be guessed.

Now we use Claim~\ref{cl:other_components} to expand $\widetilde{S}$ and therefore to cover $\widehat{S}$. For each $\widehat{S_A^i}$ such that $i\in[c]$ we guess
if $\widehat{S_A^i} \neq \emptyset$ and, if this is the case, we guess a~connected component $D_i$ in $G-(N(X_A)\cup N(X_B))$ and vertices $a^i, b^i, c^i$ and we add $N(D_i)\cup \left(N(\{a^i, b^i, c^i\})\cap N(X_B)\right)$ to $\widetilde{S}$. We do symmetric guesses for each $\widehat{S_B^i}$ for each $i\in [c]$. 

At this point, all vertices of $S \setminus \widetilde{S}$
have neighbors only in $A\cup S \cup B$. We will use the following two claims to cover
some vertices which have two incomparable neighbors in $\T$. 

Fix $\alpha\in[d]$. Let $S_A^\alpha$ be a~subset of $S_A\setminus \widetilde{S}$, consisting of these vertices of $S_A$ which have at least two neighbors in $\T^\alpha\cap B$. Let us divide $S_A^\alpha$ into color classes $S_A^{\alpha,1},S_A^{\alpha, 2},\dots, S_A^{\alpha,c}$. Similarly, let $S_B^\alpha$ be a~subset of $S_B\setminus \widetilde{S}$, consisting of these vertices of $S_B$ which have two neighbors in $\T^\alpha\cap A$. Again, we divide $S_B^\alpha$ into  color classes $S_B^{\alpha,1}, S_B^{\alpha,2}, \dots, S_B^{\alpha,c}$.

\begin{claim}\label{cl:two_on_same_level}
For each $\alpha\in[d]$ and $i\in[c]$ there exist
$Q_Z \subseteq A$ of size at most $2(k+1)!$, and disjoint sets of vertices $Q_I^1,Q_I^2 \subseteq \T^\alpha \cap B$ of size at most $(k+1)!$ each, such that 
$$S_A^{\alpha, i} \subseteq \left(N(Q_Z)\cap N(X_B) \right)\cup \left(\bigcup_{x\in Q_I^1}\bigcup_{y\in (Q_I^1\cup Q_I^2)\setminus \{x\}} N(x)\cap N(y)\right).$$
A symmetrical statement holds for $S_B^{\alpha,i}$ for each $\alpha\in[d]$ and $i\in[c]$.
\end{claim}
\begin{proof}
    Let us fix $\alpha\in[d]$ and $i\in[c]$. We consider $S_A^{\alpha, i}$, for shorthand denoted here as $J$. Since $J$ is a~color class, it is an independent set. By Lemma~\ref{lem:p_four} there exists a~connected cograph $Z$ in $A$ such that $S\setminus N(Z)$ is complete to $B$. Therefore $J\subseteq N(Z)$ (vertices belonging to $S$ which are complete to $B$ belong to $N(p_B)$, so they belong to $\widetilde{S}$). Let $I^1$ be set of vertices of $\T^\alpha$ belonging to $B$. Note that $I^1$ is an independent set. By Lemma~\ref{lem:ind_set_neighboring_cograph} there exist sets $Q_I^1\subseteq I^1$ and $Q_Z^1\subseteq Z$ of size at most $(k+1)!$ each such that $J\subseteq N(Q_I^1)\cup N(Q_Z^1)$. Note that $N(Q_Z^1)\cap S_A^{\alpha, i} \subseteq N(Q_Z^1)\cap N(X_B)$.
    
    Let us denote  $N(Q_Z^1)\cap N(X_B) \cup \bigcup_{x,y\in Q_I^1,\ x\neq y} N(x)\cap N(y)$ as $X$. Let $J^2 = J \setminus X$ and $I^2 = I^1 \setminus Q_I^1$. Note that if $v\in J^2$, then $v$ has exactly one neighbor in $Q_I^1$, so it has a~neighbor in $I^2$. We again use Lemma~\ref{lem:ind_set_neighboring_cograph}, there exists sets  $Q_I^2\subseteq I^2$ and $Q_Z^2\subseteq Z$ of size at most $(k+1)!$ such that $J^2\subseteq N(Q_I^2)\cup N(Q_Z^2)$.
    
    Then we have $J^2\subseteq (N(Q_Z^2)\cap N(X_B)) \cup (\bigcup_{x\in Q_I^1}\bigcup_{y\in Q_I^2} N(x)\cap N(y))$. 
    Thus $J$ is indeed a~subset of 
    $$S_A^{\alpha, i} \subseteq \left(\left(N(Q_Z^1\cup Q_Z^2)\cap N(X_B) \right)\cup \left(\bigcup_{x\in Q_I^1}\bigcup_{y\in (Q_I^1\cup Q_I^2)\setminus \{x\}} N(x)\cap N(y)\right) \right).$$

    Setting $Q_Z = Q_Z^1 \cup Q_Z^2$ concludes the proof; the proof for $B$ is symmetric.
\end{proof} 

For every $\alpha\in[d]$ and $i\in[c]$ we guess if $S_A^{\alpha,i}$ of $S_A^\alpha$ is nonempty, and if it is the case, then we guess sets $Q_Z^{i}$, $Q_I^{1,i}$, $Q_I^{2,i}$. By Claim~\ref{cl:two_on_same_level} we know that there exist suitable guesses such that $S_A^{\alpha, i}$ is contained in
$$\left(\left(N(Q_Z^i)\cap N(X_B) \right)\cup \left(\bigcup_{x\in Q_I^{1,i}}\bigcup_{y\in (Q_I^{1,i}\cup Q_I^{2,i})\setminus \{x\}} N(x)\cap N(y)\right) \right),$$
which we denote as $X^{\alpha, i}$ (note that $X_B$ is already guessed). 
Then for each $\alpha\in[d]$ and $i\in[c]$ we add $X^{\alpha,i}$ to $\widetilde{S}$.

We need to bound how many vertices of $\T$ are added to $\widetilde{S}$. Per each $X^{\alpha, i}$ we may have added to $\widetilde{S}$ vertices from $\T$, which are ancestors of two vertices belonging to $Q_I^{1,i}$ and $Q_I^{2,i}$. So we added at most $d \cdot \left|Q_I^{1,i}\cup Q_I^{2,i}\right|^2$. So the number of vertices of $\T$ in $\widetilde{S}$ increased by at most $d \cdot c \cdot d \cdot 4((k+1)!)^2$.

We follow the same steps for each color class $S_B^{\alpha, i}$ of $S_B^\alpha$, guessing the sets and putting vertices into $\widetilde{S}$.
Symmetrically, this step increases the number of vertices of $\T$ in 
$\widetilde{S}$ by at most $4cd^2((k+1)!)^2$.

We put into $\widetilde{S}$ vertices of $S_A$ (resp. $S_B$), which have at least two neighbors on the same level in $\T\cap B$ (resp. $\T\cap A$). We will use now a~similar idea~to 
put into $\widetilde{S}$ some vertices of $S$ which have two incomparable neighbors in $\T$ of two different levels.

Fix distinct $\alpha, \beta\in[d]$. Let $S_A^{\alpha,\beta}$ be a subset of $S_A\setminus \widetilde{S}$ consisting of these vertices of $S_A$ which have two incomparable neighbors in $\T\cap B$ -- one of depth $\alpha$ in $\T$ and another of depth $\beta$ in $\T$. 
Note that each vertex $v\in S^{\alpha, \beta}_A$ has exactly one vertex in $\T^{\alpha}\cap B$ and exactly one in $\T^\beta\cap B$, as vertices having at least two neighbors in either of these two sets already belong to $\widetilde{S}$. We divide each $S_A^{\alpha,\beta}$ into color classes $S_A^{\alpha,\beta,1}, S_A^{\alpha,\beta, 2}, \dots, S_A^{\alpha,\beta,c}$. Let us define $S_B^{\alpha,\beta, i}$ in a symmetrical way.
\begin{claim}\label{cl:two_incomparable_neighbors_on_different_levels}
    For each distinct $\alpha, \beta\in[d]$ and $i\in[c]$ there exist a set $Q_Z \subseteq A$
    of size at most $2(k+1)!$, and sets $Q_I^1$ and $Q_I^2$, which are subsets of respectively $\T^\alpha\cap B$ and $\T^\beta\cap B$ of size at most $(k+1)!$ such that
    $$S_A^{\alpha,\beta, i}\subseteq \left(N(Q_Z)\cap N(X_B)\right)\cup \underset{x,y \text{ incomparable in }\T}{\bigcup_{x\in Q_I^1}\bigcup_{y\in Q_I^2}} N(x)\cap N(y).$$
    A symmetric statement holds for $S_B^{\alpha,\beta,i}$ for each distinct $\alpha,\beta\in[d]$ and $i\in[c]$.
\end{claim}
\begin{proof}
    Let us fix distinct $\alpha, \beta\in[d]$ and $c\in[c]$. We consider $S_A^{\alpha,\beta, i}$, for short denoted here as $J$. By Lemma~\ref{lem:p_four}, there exists a connected cograph $Z$ in $A$ such that $S\setminus N(Z)$ is complete to $B$. Therefore $J\subseteq N(Z)$. We apply Lemma~\ref{lem:ind_set_neighboring_cograph} to $Z$, $J$ and $\T^\alpha\cap B$, getting sets $Q_Z^{1}$ and $Q_I^{1}$ such that $S_A^{\alpha, \beta, i}\subseteq N(Q_Z^{1})\cup N(Q_I^{1})$.
    We also apply Lemma~\ref{lem:ind_set_neighboring_cograph} to $Z$, $J$ and $\T^\beta\cap B$, getting sets $Q_Z^2$ and $Q_I^2$ such that $S_A^{\alpha, \beta, i}\subseteq N(Q_Z^{2})\cup N(Q_I^{2})$. Let $Q_Z = Q_Z^1 \cup Q_Z^2$.

    Note that 
    \[ S^{\alpha, \beta, i}_A \setminus N(Q_Z) \subseteq \{N(x)\cap N(y) \mid x \in Q_I^1,\ y \in Q_I^2,\ x \text{ and } y \text{ are incomparable in } \T\}, \]
    which holds by the definition of $S^{\alpha, \beta, i}_A$. As $N(Q_Z^1\cup Q_Z^2)\cap S_A^{\alpha, \beta, i} \subseteq N(Q_Z)\cap N(X_B)$, we get 
    $$S_A^{\alpha,\beta, i}\subseteq \left(N(Q_Z)\cap N(X_B)\right)\cup \underset{x,y \text{ incomparable}}{\bigcup_{x\in Q_I^1}\bigcup_{y\in Q_I^2}} N(x)\cap N(y).$$

    The symmetrical proof holds for $B$.
\end{proof}

For each $\alpha, \beta\in[d]$ and $i\in[c]$ we guess id $S_A^{\alpha, \beta, i}$ is nonempty, and if this is the case, we guess sets $Q_Z$, $Q_I^1$ and $Q_I^2$.
By Claim~\ref{cl:two_incomparable_neighbors_on_different_levels} we know that there exists suitable guesses such that $S_A^{\alpha,\beta, i}$ is contained in
$$\left(N(Q_Z)\cap N(X_B)\right)\cup \underset{x,y \text{ incomparable}}{\bigcup_{x\in Q_I^1}\bigcup_{y\in Q_I^2}} N(x)\cap N(y),$$
which we denote as $X^{\alpha, \beta, i}$. For each distinct $\alpha, \beta\in[d]$ and $i\in [c]$ we add $X^{\alpha, \beta, i}$ to $\widetilde{S}$.

With each $X^{\alpha, \beta, i}$ we may have added at most $d\cdot |Q_I^1\cup Q_I^2|^2$ vertices from $\T$ to $\widetilde{S}$, so the size of $\T \cap \widetilde{S}$ increases by at most $d\cdot d \cdot c \cdot d(2(k+1)!)^2$ after adding all $X^{\alpha, \beta, i}$'s.

We follow the same steps for $S_B^{\alpha, \beta, i}$, guessing the sets and adding vertices to $\widetilde{S}$. The increase of $|\T \cap \widetilde{S}|$ is bounded similarly.

At this moment, for every vertex $v\in S_A\setminus \widetilde{S}$ there are no incomparable vertices in $N(v)\cap \T \cap B$. A symmetrical statement holds for vertices belonging to $S_B\setminus \widetilde{S}$.

Fix $\alpha \in [d]$. Let us denote set $\{v\in S_A\setminus \widetilde{S} \mid N(v)\cap A\cap \T^\alpha\neq \varnothing\}$ as  $\overline{S_A^\alpha}$. Let us divide $\overline{S_A^\alpha}$ into color classes: $\overline{S_A^{\alpha, 1}},\overline{S_A^{\alpha,2}}, \dots, \overline{S_A^{\alpha, c}}$. 
Similarly, let $\overline{S_B^\alpha}=\{v\in S_B\setminus \widetilde{S}\mid N(v)\cap B\cap T^\alpha\neq\varnothing\}$ and let $\overline{S_B^{\alpha, 1}}, \overline{S_B^{\alpha, 2}}, \dots, \overline{S_B^{\alpha, c}}$ be color classes of $\overline{S_B^{\alpha}}$.

\begin{claim}\label{cl:neighbors_in_my_part_on_the_same_level1}
For each $i \in [c]$ and $\alpha \in [d]$,
if $\overline{S_A^{\alpha,i}} \neq \emptyset$, there exists 
vertices $p^{\alpha,i},p_1^{\alpha,i},p_2^{\alpha,i} \in \T^\alpha \cap A$
and 
vertices $a^{\alpha,i},b^{\alpha,i},c^{\alpha,i} \in A$ inducing a $P_3$
such that every $v \in \overline{S_A^{\alpha,i}}$ is either
adjacent to at least one of the vertices 
$p^{\alpha,i},p_1^{\alpha,i},p_2^{\alpha,i},a^{\alpha,i},b^{\alpha,i},c^{\alpha,i}$
or satisfies the following: every $u \in N(v) \cap A \cap \T^\alpha$ 
has a neighbor among vertices $a^{\alpha,i},b^{\alpha,i},c^{\alpha,i}$.

A symmetrical statement holds for $\overline{S_B^{\alpha,i}}$ for each $\alpha\in [d]$ and $i\in[c]$.
\end{claim}
\begin{proof}
    Let us fix $i\in[c]$ and $\alpha\in [d]$. 
    We consider $\overline{S_A^{\alpha, i}}$, for short denoted here as $J$. 
    Let $v^{\alpha, i}\in J$ be a vertex, minimizing $|N(v^{\alpha, i})\cap A\cap \T^\alpha|$. 
    By the definition of $S_A$, there exists an induced path $v^{\alpha, i} - a^{\alpha, i} - b^{\alpha, i} -c^{\alpha,i}$, where $a^{\alpha,i}, b^{\alpha, i}, c^{\alpha, i} \in A$. Let $p^{\alpha, i}$ be any vertex belonging to $N(v^{\alpha, i})\cap A\cap \T^\alpha$. 
    Let $$J_1 = \{v \in J\setminus (N(a^{\alpha,i}, b^{\alpha,i}, c^{\alpha, i}, p^{\alpha, i})\cap N(X_B)) \mid N(v)\cap N(v^{\alpha, i})\cap A \cap \T^\alpha \neq \varnothing\}.$$ 
    Let us consider two partial orders on vertices of $J_1$:
    \begin{itemize}
        \item $\leq_1:$ $u\leq_1 w \Leftrightarrow \left((N(u)\cap A\cap \T^\alpha)\cap N(v^{\alpha, i})\right) \subseteq \left((N(w)\cap A\cap \T^\alpha)\cap N(v^{\alpha, i})\right)$ 
        \item $\leq_2:$ $u\leq_2 w \Leftrightarrow \left((N(u)\cap A\cap \T^\alpha)\setminus N(v^{\alpha, i})\right) \subseteq \left((N(w)\cap A\cap \T^\alpha)\setminus N(v^{\alpha, i})\right)$
    \end{itemize}
    Note that for every $w\in J_1$ the set $\left((N(w)\cap A\cap \T^\alpha)\cap N(v^{\alpha, i})\right)$ is non-empty by the definition of $J_1$ and the set $\left((N(w)\cap A\cap \T^\alpha)\setminus N(v^{\alpha, i})\right)$ is non-empty by the choice of  $v^{\alpha, i}$ and the exclusion of neighbors of $p^{\alpha, i}$ in $J_1$.
    
    Suppose that there exists $u$ and $w$, which are incomparable in both orders. Then we can choose  following vertices: \begin{itemize}
        \item $x_1\in \left(N(u)\cap A\cap \T^\alpha\right) \setminus N(w, v^{\alpha, i})$
        \item $x_2\in \left(N(w)\cap A\cap \T^\alpha\right) \setminus N(u, v^{\alpha, i})$
        \item $y_1 \in \left(N(u)\cap N(v^{\alpha,i})\cap A\cap \T^\alpha \right) \setminus N(w)$
        \item $y_2\in \left(N(w)\cap N(v^{\alpha,i})\cap A\cap \T^\alpha \right) \setminus N(u)$
    \end{itemize}
    Then a path $x_1 -u-y_1 -v^{\alpha, i}-y_2-w-x_2$ is an induced $P_7$, which is a contradiction. Thus, any pair of vertices of $J_1$ is comparable in at least one of the orders. 

    By Lemma~\ref{lem:two_orders} there exists a vertex $v_1^{\alpha,i}$ such that for any other vertex $w\in J_1$ we have \linebreak $v_1^{\alpha, i}\leq_1 w$ or $v_1^{\alpha,i}\leq_2 w$. 
    Let us choose vertices $p_1^{\alpha, i}\in \left(N(v_1^{\alpha, i}) \cap A\cap \T^\alpha\cap N(v^{\alpha, i})\right)$ and \linebreak $p_2^{\alpha, i}\in \left(N(v_1^{\alpha, i}) \cap A\cap \T^\alpha \setminus N(v^{\alpha, i})\right)$. Then for any $w\in J_1$ we have $w\in N(p_1^{\alpha, i}, p_2^{\alpha,i})\cap N(X_B)$. 

    Consider now $v \in J \setminus N(p^{\alpha,i},p^{\alpha,i}_1,p^{\alpha,i}_2,a^{\alpha,i},b^{\alpha,i},c^{\alpha,i})$;
    note that $v \notin J_1$.
    Let $u \in N(v) \cap \T^\alpha \cap A$. 
    As $v \notin J_1$, $uv^{\alpha,i} \notin E(G)$.
    Let $Q\subseteq B$ be ap ath (maybe on just one vertex) such that $v-Q-v^{\alpha, i}$ is an induced path in $G$. 
    Then, $c^{\alpha,i}-b^{\alpha,i}-a^{\alpha,i}-v^{\alpha,i}-Q-v-u$ contains an induced $P_7$
    unless $u$ is adjacent to at least one of the vertices $a^{\alpha,i},b^{\alpha,i},c^{\alpha,i}$. This finishes the proof.
\end{proof}

For every $i \in [c]$ and $\alpha \in [d]$, we guess if $\overline{S_A^{\alpha,i}}$ is 
nonempty and, if this is the case, we guess the vertices
$p^{\alpha,i},p_1^{\alpha,i},p_2^{\alpha,i},a^{\alpha,i},b^{\alpha,i},c^{\alpha,i}$
of Claim~\ref{cl:neighbors_in_my_part_on_the_same_level1}
and put 
\[ N(p^{\alpha,i},p_1^{\alpha,i},p_2^{\alpha,i},a^{\alpha,i},b^{\alpha,i},c^{\alpha,i}) \cap N(X_B)
\]
into $\widetilde{S}$. 
Perform a symmetrical operation with the roles of $A$ and $B$ swapped. 

Let $K_A = \{a^{\alpha,i},b^{\alpha,i},c^{\alpha,i}~|~i \in [c],\alpha \in [d]\}$. 
At this point, for every $w \in S_A \setminus \widetilde{S}$, 
all neighbors of $w$ in $\T \cap A$ are in $N(K_A)$, while $N(K_A) \subseteq A \cup S$. 
Symmetrically we define $K_B$. Note that $|K_A|,|K_B| \leq 3cd$. 

Let $A^1, A^2, \ldots, A^c$ and $B^1, B^2, \dots, B^c$ be color classes of $A$ and $B$, respectively.
For every $i,j \in [c]$ and $\alpha \in [d]$,
let $\overline{S_A^{\alpha,i,j}}$ be the set of those vertices of
$\overline{S_A^{\alpha,i}} \setminus \widetilde{S}$ that have a neighbor in $B^j$. 
Symmetrically we define $\overline{S_B^{\alpha,i,j}}$. 

\begin{claim}\label{cl:neighbors_in_my_part_on_the_same_level2}
For each $i,j\in [c]$ and $\alpha\in[d]$, if $\overline{S_A^{\alpha,i,j}} \setminus \widetilde{S} \neq \emptyset$, then there exist
$q^{\alpha,i,j} \in A$ and $v^{\alpha,i,j} \in \overline{S_A^{\alpha,i,j}}$
such that every vertex 
$v \in \overline{S_A^{\alpha, i,j}} \setminus \left(N(q^{\alpha,i,j})\cap N(X_B)\right)$
satisfies $N(v) \cap B^j \subseteq N(v^{\alpha,i,j}) \cap B^j$. 

A symmetrical statement holds for $\overline{S_B^{\alpha,i,j}}$ for each $\alpha\in[d]$ and $i\in[c]$. 
\end{claim}
\begin{proof}
    Fix $i,j \in [c]$ and $\alpha \in [d]$ and for brevity denote
    $J = \overline{S_A^{\alpha,i,j}}$. 

    Let us consider two orders on $J$:
    \begin{itemize}
        \item $\leq_1:$ $u\leq_1 w \Leftrightarrow N(u)\cap A\cap \T^\alpha \subseteq N(w)\cap A\cap \T^\alpha$
        \item $\leq_2:$ $u\leq_2 w \Leftrightarrow (N(u)\cap B^j) \supseteq (N(w)\cap B^j)$
    \end{itemize}
    Suppose that there exist vertices $u$ and $w$ in $J$, which are incomparable in both orders. Then we choose $x_1\in N(u)\cap A\cap T^\alpha\setminus N(w)$ and $x_2\in N(w)\cap A\cap T^\alpha\setminus N(u)$, and $y_1 \in N(u)\cap B^j \setminus N(w)$ and $y_2\in N(w)\cap B^j\setminus N(u)$. 
    Then $x_1-u-y_1$ and $x_2-w-y_2$ are induced paths on 3 vertices and there are no edges between these two paths. As $x_1,x_2\in N(a^{\alpha, i}, b^{\alpha, i}, c^{\alpha,i})$, we can connect these two paths via vertices $a^{\alpha, i}, b^{\alpha, i}, c^{\alpha,i}$, finding then an induced $P_7$, which is a contradiction. Thus any pair of vertices is comparable in at least one of these orders. 

    By Lemma~\ref{lem:two_orders} there exists a vertex $v^{\alpha,i, j} \in J$ such that for any other vertex $w\in J$ we have $v^{\alpha,i, j} \leq_1 w$ or $v^{\alpha, i,j}\leq_2 w$. Let $q^{\alpha, i, j}$ be a vertex belonging to $N(v^{\alpha, i,j})\cap A\cap \T^\alpha$. 
    Then for each $w\in J$ we have $w\in N(q^{\alpha, i,j})\cap N(X_B)$ or $N(w)\cap B^j\subseteq N(v^{\alpha,i,j})\cap B^j$.
    This finishes the proof of the claim. 
\end{proof}
For every $i,j \in [c]$ and $\alpha \in [d]$, guess if $\overline{S_A^{\alpha,i,j}}$
is nonempty and, if this is the case, guess the vertices $q^{\alpha,i,j}$ 
and $v^{\alpha,i,j}$ of Claim~\ref{cl:neighbors_in_my_part_on_the_same_level2}. 
Add $N(q^{\alpha,i,j}) \cap N(X_B)$ 
and $N(v^{\alpha,i,j}) \setminus N(K_A)$ to $\widetilde{S}$.

Note that in the last step we do not add any elements of $N(v^{\alpha,i,j}) \cap \T \cap A$
to $\widetilde{S}$.
As $v^{\alpha,i,j} \notin \widetilde{S}$ prior to this step and
$N(v^{\alpha,i,j}) \cap \T \cap B$ is of size at most $d$, we add at most $d^2 c^2$ vertices of $\T$ to $\widetilde{S}$ in this step.

We perform a symmetric operation with $A$ and $B$ swapped.

At this moment, for every $v \in S_A \setminus \widetilde{S}$, 
we have $N(v) \cap A \cap \T \subseteq N(K_A)$ and $N(v) \cap B \subseteq \widetilde{S}$
and, symmetrically, 
for every $v \in S_B \setminus \widetilde{S}$, 
we have $N(v) \cap B \cap \T \subseteq N(K_B)$ and $N(v) \cap A \subseteq \widetilde{S}$.

Recall that $|X_B| \leq k$ and $|K_B| \leq 3cd$. 
Let $\widetilde{X_B}$ be the smallest connected subgraph of $B$, which contains $X_B\cup K_B$;
as $G$ is $P_7$-free, $|\widetilde{X_B}| \leq 7(k + 3cd)$. 
Symmetrically define $\widetilde{X_A}$. 
Guess $\widetilde{X_A}$, $\widetilde{X_B}$ and add $N(\widetilde{X_A}) \cap N(\widetilde{X_B})$
to $\widetilde{S}$. 
%\Jadwiga{Zastanowić się czy tutaj czegoś nie dopisać.}

We now perform the following cleaning operation.
\begin{enumerate}
    \item Guess all vertices of $\T \cap \widetilde{S}$
and, for every $v \in \T \cap \widetilde{S}$, guess all ancestors of $v$ in $\T$. 
Add those ancestors to $\widetilde{S}$; this increases the size of $\T \cap \widetilde{S}$
by at most a factor of $d$. Now $\T \cap \widetilde{S}$ is a sub-treedepth-$d$ structure
of $\T$. 

    \item Guess the depths and parent/child relation in $\T$
    of vertices of $\T \cap \widetilde{S}$. In particular, this guess determines
    which pairs of vertices are incomparable in $\T$.
    \item For every pair $u,v$ of vertices of $\T \cap \widetilde{S}$ that are incomparable in $\T$, add $N(u) \cap N(v)$ to $\widetilde{S}$. 
    \item For every $v \in \T^d \cap \widetilde{S}$, add $N(v)$ to $\widetilde{S}$.
    
    Note that due to the first step, subsequent steps
    do not add any new vertex of $\T$ to $\widetilde{S}$. 
\end{enumerate}
After the above cleaning, for every $v \in V(G) \setminus \widetilde{S}$, all
vertices of $N(v) \cap \T \cap \widetilde{S}$ are comparable in $\T$
and do not contain any vertex of $\T^d$. In particular, as $\T$ is maximal, $v$ has
a neighbor in $\T \setminus \widetilde{S}$.
Consequently, every vertex of $S_A \setminus \widetilde{S}$ has a neighbor
in $(A \setminus \widetilde{S}) \cap \T$ and every vertex 
of $S_B \setminus \widetilde{S}$ has a~neighbor in 
$(B \setminus \widetilde{S}) \cap \T$. 

We infer that the connected components of $G-\widetilde{S}$ are of the following types:
\begin{description}
    \item[(clean)] Contained in a connected component of $G-S$ (i.e., in $A$, in $B$, or in $G-(A \cup S \cup B)$), or
    \item[(dirty)] Contained in $A \cup S \cup B$, with a nonempty intersection both with $S$
    and $A \cup B$. 
\end{description}
Observe that the clean components are disjoint with $N(\widetilde{X_A})$ or with 
$N(\widetilde{X_B})$ (or both), while the dirty components have nonempty intersection
both with $N(\widetilde{X_A})$ and $N(\widetilde{X_B})$,
as $S_A \subseteq N(\widetilde{X_B})$, 
while every vertex of $S_A \setminus \widetilde{S}$ has a neighbor in $((A \setminus \widetilde{S}) \cap \T)$, which in turn has a neighbor in $K_A \subseteq \widetilde{X_A}$,
and symmetrically for $B$. 
Hence, the algorithm can distinguish clean and dirty components. 

In what follows, in a number of steps we will add some vertices to $\widetilde{S}$,
but we will be only adding vertices of $S$. 
After the updates,
we will continue using the nomenclature of clean/dirty components of $G-\widetilde{S}$, 
and they always refer to the current value of $\widetilde{S}$. 
Note that, as long as we keep the invariant that every vertex of $S_A \setminus \widetilde{S}$
has a neighbor in $(A \setminus \widetilde{S}) \setminus \T$ (which is maintained trivially
if we add only vertices of $S$ to $\widetilde{S}$), then the method of distiguishing
clean and dirty components from the previous paragraph still works. 

We now branch if $S_A \setminus \widetilde{S} = \emptyset$. 
If this is the case, then
every dirty component $C$ satisfies $C \cap N(\widetilde{X_A}) = C \cap S$.
We insert $C \cap N(\widetilde{X_A})$ into $\widetilde{S}$ for every dirty component $C$. 
As a result, $S \subseteq \widetilde{S}$ while $|\widetilde{S} \cap \T| = \Oh_{k,d}(1)$. 
We insert $(K := \widetilde{S}, \mathcal{D} := \emptyset, L := \emptyset)$ into $\mathcal{F}$
and conclude this branch. 

Symmetrically we handle a branch $S_B \setminus \widetilde{S} = \emptyset$. 
In the remaining branch, we assume that both $S_A \setminus \widetilde{S}$ and
$S_B \setminus \widetilde{S}$ are nonempty. 

Let $\widehat{X_B}=\widetilde{X_B}\cap N(S_A\setminus \widetilde{S})$;
note that as $S_A \setminus \widetilde{S} \neq \emptyset$ and $S_A \subseteq N(X_B)$, 
we have $\widehat{X_B} \neq \emptyset$.
We symmetrically define $\widehat{X_A}$ and observe it is nonempty, too. 

We now focus on connected components of $B \setminus \widetilde{S}$. 
Let $C$ be such a component. Note that $C$ has no neighbors in $S_A \setminus \widetilde{S}$;
hence $C$ is also a connected component of $G-\widetilde{S}-N(X_A)$, and thus can be guessed
if needed. 

\begin{claim}\label{cl:Bcomps}
For every connected component $C$ of $B \setminus \widetilde{S}$,
\begin{enumerate}
    \item every $x \in \widehat{X_B}$ is either complete or anticomplete to $C$, and
    \item if the whole $\widehat{X_B}$ is anticomplete to $C$, then 
    the whole $\widetilde{X_B}$ is anticomplete to $C$. 
\end{enumerate}
A symmetrical statement holds with $A$ and $B$ swapped.
\end{claim}
\begin{proof}
For the first point, assume there exists $x \in \widehat{X_B}$ with an induced $P_3$ of the form
$x-C-C$. Let $y$ be any neighbor of $x$ in
$S_A \setminus \widetilde{S}$. Observe that there would exist an induced $P_7$
of the form $C-C-x-y-A-A-A$, a contradiction. 

If $C$ is antiadjacent to $\widehat{X_B}$ but has a neighbor in $\widetilde{X_B}$,
then as $\widehat{X_B} \neq \emptyset$ there exists a path $Q$ from $\widehat{X_B}$ via $\widetilde{X_B}$ to $C$
of length at least $2$; let $x$ be its endpoint in $\widehat{X_B}$ and
let $y$ be a neighbor of $x$ in $S_A \setminus \widetilde{S}$.
Then there exists a $P_7$ of the form $Q-y-A-A-A$, a contradiction.
\end{proof}
Claim~\ref{cl:Bcomps} distinguishes two types of components of $B \setminus \widetilde{S}$:
those antiadjacent to $\widetilde{X_B}$
and those complete to at least one vertex of $\widehat{X_B}$. 
Let $\mathcal{B}_0$ and $\mathcal{B}_1$ be the families of components of the first
and second type, respectively. Symmetrically define $\mathcal{A}_0$ and $\mathcal{A}_1$. 

Recall that every vertex of $S_B \setminus \widetilde{S}$ has a neighbor in
$(B \setminus \widetilde{S}) \cap \T$, which in turn has a neighbor in $K_B \subseteq \widetilde{X_B}$. This implies the following.
\begin{claim}\label{cl:B1nei}
 Every vertex of $S_B \setminus \widetilde{S}$ has a neighbor in a component of $\mathcal{B}_1$. 
 Symmetrically, 
 every vertex of $S_A \setminus \widetilde{S}$ has a neighbor in a component of $\mathcal{A}_1$. 
\end{claim}

For $v \in S_B \setminus \widetilde{S}$, let $N_{\mathcal{B}_0}(v) = \{C \in \mathcal{B}_0~|~C \cap N(v) \neq \emptyset\}$
    and $N_{\mathcal{B}_1}(v) = \{C \in \mathcal{B}_1~|~C \cap N(v) \neq \emptyset\}$.

We now sort out adjacencies between $S_B \setminus \widetilde{S}$ and $\mathcal{B}_0$. 
Let $S_{B,0} $ be the set of vertices belonging to $S_B\setminus \widetilde{S}$ which have a
neighbor in a component of $\mathcal{B}_0$. Let $S_{B,0}^1, S_{B,0}^2, \dots, S_{B,0}^c$ be color classes of $S_{B,0}$. Symmetrically, we define $S_{A,0}$ and color classes
$S_{A,0}^1,\ldots,S_{A,0}^c$. 

\begin{claim}\label{cl:B0_vertices}
For every $i \in [c]$, if $S_{B,0}^i$ is nonempty, then there exist
two components $B_{0,i,0}$ and $B_{0,i,1}$ of $B \setminus \widetilde{S}$
such that $S_{B,0}^i \subseteq N(B_{0,i,0} \cup B_{0,i,1})$. 

A symmetrical statement holds with $A$ and $B$ swapped.
\end{claim}
\begin{proof}
    Fix $i \in[c]$.
    Then we consider two orders on $S_{B,0}^i$:
    \begin{itemize}
        \item $\leq_1$: $x\leq_1 y \Leftrightarrow N_{\mathcal{B}_1}(x) \subseteq N_{\mathcal{B}_1}(y)$;
         
        \item $\leq_2$: $x\leq_2 y \Leftrightarrow N_{\mathcal{B}_0}(x) \subseteq N_{\mathcal{B}_0}(y)$;
    \end{itemize}
    Suppose that there exist $x$ and $y$, which are incomparable in both orders. Then we can choose
    \begin{itemize}
        \item $C_x\in N_{\mathcal{B}_1}(x) \setminus N_{\mathcal{B}_1}(y)$,
        \item $C_y \in N_{\mathcal{B}_1}(y) \setminus N_{\mathcal{B}_1}(x)$,
        \item $D_x\in N_{\mathcal{B}_0}(x) \setminus N_{\mathcal{B}_0}(y)$,
        \item $D_y\in N_{\mathcal{B}_0}(y) \setminus N_{\mathcal{B}_0}(x)$.
    \end{itemize} 
    There are no edges between $C_x \cup D_x \cup \{x\}$ and $C_y \cup D_y \cup \{y\}$. 
    Then, there is an induced $P_7$ of the form $D_x-x-C_x-\widetilde{X_B}-C_y-y-D_y$, a contradiction.
    Thus any pair of vertices is comparable in at least one order.
    
    By Lemma~\ref{lem:two_orders} we can choose a vertex $v^i \in S_{B,0}^i$ such that for any other vertex $y\in S_{B,0}^i$ we have $v^{i}\leq_1 y$ or $v^{i}\leq_2 y$.
    Hence, any $B_{0,i,0} \in N_{\mathcal{B}_0}(v^i)$ (which exists by the definition of $S_{B,0}$) 
    and any $B_{0,i,1} \in N_{\mathcal{B}_1}(v^i)$ (which we exists by Claim~\ref{cl:B1nei})
    would do the job.
\end{proof}
As discussed, components of $\mathcal{B}_0 \cup \mathcal{B}_1$ are components
of $G-\widetilde{S}-N(X_A)$ and can be guessed.
For every $i \in [c]$, we guess $B_{0,i,0}$ and $B_{0,i,1}$ and add
$N(B_{0,i,0} \cup B_{0,i,1}) \cap N(X_A)$ to $\widetilde{S}$.

We perform a symmetrical operation for $A$. 
As a result, no vertex of $S_B \setminus \widetilde{S}$ has a neighbor in a component
of $\mathcal{B}_0$ and no vertex of $S_A \setminus \widetilde{S}$ has a neighbor in 
$\mathcal{A}_0$. Also, we added to $\widetilde{S}$ only vertices of $S$,
so in particular no vertex of $\T$ was added to $\widetilde{S}$
and the families $\mathcal{B}_0$, $\mathcal{B}_1$, $\mathcal{A}_0$, and $\mathcal{A}_1$ stay
intact. 

We say that $x \in S_B \setminus \widetilde{S}$ is \emph{mixed} to 
a component $C \in \mathcal{B}_1$ if there exists an induced path $x-C-C$ or, equivalently,
$x$ is neither complete nor anticomplete to $C$. 
A component $C \in \mathcal{B}_1$ is \emph{problematic} if there is $x \in S_B \setminus \widetilde{S}$ that is mixed to $C$. 
For $x \in S_B \setminus \widetilde{S}$ we denote by $\texttt{Prob}(x) \subseteq \mathcal{B}_1$
the set of components $x$ is mixed to
and by $\texttt{All}(x) \subseteq \mathcal{B}_1$ the set of components $x$ is mixed or complete to.
Let $S_B^P$ be the set of all vertices of $S_B \setminus \widetilde{S}$
that are mixed to at least one component of $\mathcal{B}_1$
(i.e., $S_B^P = \{x \in S_B \setminus \widetilde{S}~|~\texttt{Prob}(x) \neq \emptyset\}$)
and let $S_B^{P,1},S_B^{P,2},\ldots,S_B^{P,c}$ be color classes of $S_B^P$. 
We introduce symmetrical definitions with $A$ and $B$ swapped. 

\begin{claim}\label{cl:problematic_are_a_subset_of_all}
For each $i\in[c]$ and any pair of vertices $v_1, v_2\in S_B^{P,i}$ we have either $\texttt{Prob}(v_1)\subseteq \texttt{All}(v_2)$ or $\texttt{Prob}(v_2)\subseteq \texttt{All}(v_1)$. A symmetrical statement holds for any pair of vertices of $S_A^{P,i}$ for each $i\in[c]$.
\end{claim}
\begin{proof}
   Fix $i\in [c]$. Suppose that there exist vertices $v_1$ and $v_2$ in $S_B^{P,i}$ and components $C_{v_1}\in \texttt{Prob}(v_1)\setminus \texttt{All}(v_2)$ and $C_{v_2}\in \texttt{Prob}(v_2)\setminus \texttt{All}(v_1)$. Then there exist two induced $P_3$s $v_1-C_{v_1}-C_{v_1}$ and $v_2-C_{v_2}-C_{v_2}$, which can be connected via $\widetilde{X_A}$, so we can find a $P_7$, a contradiction.  
\end{proof}

% Let us now consider a subset $\mathcal{C}$ of $cc(G - \widetilde{S})$ containing these components $C$ such that $C\cap A\neq \varnothing$ and $C\cap B\neq \varnothing$. Note that for each $C\in \mathcal{C}$ we have $C\cap N(\widetilde{X_A}) = (C\cap A) \cup (C \cap S_B)$ and $C\cap N(\widehat{X_B}) = (C\cap B)\cup (C\cap S_A)$. As $C\subseteq A\cup S\cup B$, so $C\subseteq N(\widehat{X_A})\cup N(\widehat{X_B})$.  

\begin{claim}\label{cl:problematic_in_S}
    For every $i\in[c]$ there exists components $C_{v^i}, C_{u^i}^1, C_{u^i}^2\in 
    \mathcal{B}_1$ such that \linebreak $S_B^{P,i} \subseteq \left(N(C_{v^i})\cup N(C_{u^i}^1)\cup N(C_{u^i}^2)\right)\cap N(\widehat{X_A})$.
    A symmetrical statement holds for $S_A^{P,i}$ for each $i\in [c]$.
\end{claim}
\begin{proof}
    Fix $i\in [c]$. Let ${v^i}\in S_B^{P,i}$ be a vertex, which has neighbors in the smallest number of problematic components. Choose any $C_{v^i}\in \texttt{Prob}(v^i)$.
    Let $w$ be any vertex in $S_B^{P,i}\setminus (N(C_{v^i}) \cap N(X_A))$. Since $C_{v^i}\in \texttt{Prob}({v^i})\setminus \texttt{All}(w)$, by Claim~\ref{cl:problematic_are_a_subset_of_all} we know that $\texttt{Prob}(w)\subseteq \texttt{All}({v^i})$. By the minimality of ${v^i}$, we know that $\texttt{All}(w)\setminus \texttt{All}({v^i})\neq \varnothing$. Let us consider two orders on $S_B^{P,i}\setminus (N(C_{v^i}) \cap N(X_A))$:
    \begin{itemize}
        \item $\leq_1:$ $u\leq_1 w \Leftrightarrow \{C\in \texttt{All}({v^i})\mid N(u)\cap C\neq \varnothing\} \subseteq \{C\in \texttt{All}({v^i})\mid N(w)\cap C\neq \varnothing\} $
        \item $\leq_2:$ $u\leq_1 w \Leftrightarrow \{C\in \mathcal{B}_1\setminus \texttt{All}({v^i})\mid N(u)\cap C\neq \varnothing\} \subseteq \{C\in \mathcal{B}_1\setminus \texttt{All}({v^i})\mid N(w)\cap C\neq \varnothing\} $
    \end{itemize}
    Note that for every $u\in S_B^{P,i}\setminus (N(C_{v^i}) \cap N(\widehat{X_A}))$ the set $\{C\in \texttt{All}({v^i})\mid N(u)\cap C\neq \varnothing\}$ is non-empty, as by 
    Claim~\ref{cl:problematic_are_a_subset_of_all} we have $\emptyset \neq \texttt{Prob}(u)\subseteq \texttt{All}({v^i})$. Similarly, for every $u\in S_B^{P,i}\setminus (N(C_{v^i}) \cap N(\widehat{X_A}))$ the set $\{C\in \mathcal{B}_1 \setminus \texttt{All}({v^i})\mid N(u)\cap C\neq \varnothing\}$ is non-empty by the choice of ${v^i}$ and the exclusion of $N(C_{v^i})$. 
    
    Suppose that there exist $u$ and $w$ which are incomparable in both these orders. Then we can choose the following components:
    \begin{itemize}
        \item $C_u^1\in \texttt{All}({v^i})$ such that $N(u)\cap C_u^1\neq \varnothing$ and $N(w)\cap C_u^1=\varnothing$
        \item $C_u^2\in \mathcal{B}_1 \setminus \texttt{All}({v^i})$ such that $N(u)\cap C_u^2\neq \varnothing$ and $N(w)\cap C_u^2=\varnothing$
        \item $C_w^1\in \texttt{All}({v^i})$ such that $N(w)\cap C_w^1\neq \varnothing$ and $N(u)\cap C_u^1=\varnothing$
        \item $C_w^2\in  \mathcal{B}_1 \setminus \texttt{All}({v^i})$ such that $N(w)\cap C_w^2\neq \varnothing$ and $N(u)\cap C_u^1=\varnothing$
    \end{itemize}
    Then there exists a path $C_u^2-u-C_u^1-{v^i}-C_w^1-w-C_w^2$ that contains a $P_7$, a contradiction. Hence, any pair of vertices is comparable in at least one order. 
    
    Therefore by Lemma~\ref{lem:two_orders} there exists a vertex ${u^i}$ such that for any other vertex $w\in S_B^{P,i}\setminus (N(C_v)\cap N(\widehat{X_A}))$ we have ${u^i}\leq_1 w$ or ${u^i}\leq_2 w$. Let $C_{u^i}^1\in \texttt{All}({v^i})$ and $C_{u^i}^2\in \mathcal{B}_1\setminus\texttt{All}({v^i})$ such that $N({u^i})\cap C_u^1\neq \varnothing$ and $N({u^i})\cap C_u^2\neq \varnothing$. Then by the choice of ${u^i}$ we have that $S_B^{P,i}\setminus (N(C_{v^i})\cap N(\widehat{X_A}))\subseteq (N(C_{u^i}^1)\cup N(C_{u^i}^2))\cap N(\widehat{X_A})$. Therefore $S_B^{P,i}\subseteq \left(N(C_{v^i})\cup N(C_{u^i}^1)\cup N(C_{u^i}^2)\right)\cap N(\widehat{X_A})$.  

    By symmetry, we show that there exist sought components $C_{v^i}, C_{u^i}^1, C_{u^i}^2\in \mathcal{A}_1$ for each $S_A^{P,i}$.
\end{proof}
For every $i \in [c]$, 
we guess components $C_{v^i}$, $C_{u^i}^1$ and $C_{u^i}^2$ 
in $\mathcal{B}_1$ whose existence is guaranteed by Claim~\ref{cl:problematic_in_S},
and add $\left(N(C_{v^i})\cup N(C_{u^i}^1)\cup N(C_{u^i}^2)\right)\cap N(\widehat{X_A})$
to $\widetilde{S}$. We perform symmetric operation for $S_A^{P,i}$. 

By Claim~\ref{cl:problematic_in_S}, we have $S_A^P \cup S_B^P \subseteq \widetilde{S}$, while
we added only vertices of $S$ to $\widetilde{S}$. 
Now, every vertex of $S_B \setminus \widetilde{S}$ is complete or anticomplete
to every component of $\mathcal{B}_1$ (and anticomplete to every component of $\mathcal{B}_0$). 

Consequently, for every dirty component $C$, every connected component
of $C \cap B$ is a module of $G[C]$ and every connected component of $C \cap A$ is a module of $G[C]$. 
Furthermore, observe that thanks to Claim~\ref{cl:Bcomps} for every dirty component $C$, 
we have $C \cap N(\widetilde{X_A}) = (C \cap A) \cup (C \cap S_B)$
and $C \cap N(\widetilde{X_B}) = (C \cap B) \cup (C \cap S_A)$.
That is, $(C \cap N(\widetilde{X_A}))$ and $(C \cap N(\widetilde{X_B}))$ 
is a partition of $C$ which the algorithm can compute.
Moreover, every connected component of $G[C \cap N(\widetilde{X_A})]$ is either
a component of $A \setminus \widetilde{S}$ or a component of $S_B \setminus \widetilde{S}$
and every connected component of $G[C \cap N(\widetilde{X_B})]$ is either
a component of $B \setminus \widetilde{S}$ or a component of $S_A \setminus \widetilde{S}$.

Since every connected component of $C \cap B$ and of $C \cap A$ is a module of $G[C]$, 
we know that all components of $G[C \cap N(\widetilde{X_A})]$
or of $G[C \cap N(\widetilde{X_B})]$ that are \emph{not} modules of $G[C]$ are part of 
$\widetilde{S}$. That is, we perform the following step exhaustively:
while there exists a dirty component $C$ and a component $C'$ of either
$G[C \cap N(\widetilde{X_A})]$ or $G[C \cap N(\widetilde{X_B})]$ that is not 
a module of $G[C]$, move $C'$ to $\widetilde{S}$. 
This step, again, adds only vertices belonging to $S$ to $\widetilde{S}$.

In the end, for every dirty component $C$,
$(C \cap \widetilde{X_A},C \cap \widetilde{X_B})$ is a partition of $C$ such that
every component of $G[C \cap \widetilde{X_A}]$ and every component of $G[C \cap \widetilde{X_B}]$
is a module of $G[C]$. 
Let $\mathcal{D}$ be the set of dirty components and, for $C \in \mathcal{D}$, let $L(C) = C \cap \widetilde{X_A}$. Then, $(K := \widetilde{S}, \mathcal{D}, L)$ satisfies the
requirements of the lemma for $(A,S,B)$ and we add this tuple to $\mathcal{F}$. 

This finishes the proof of Theorem~\ref{thm:to-bip}.
\end{proof}

\section{$P_7$-free bipartite graphs}\label{sec:bipartite}
In this section we focus on $P_7$-free bipartite graphs. 
We prove the following theorem.
\begin{theorem}\label{thm:bip}
There exists an algorithm that, given a $P_7$-free bipartite graph $G$ and an integer $d > 0$,
runs in time $n^{2^{2^{\Oh(d^3)}}}$ and computes a family $\mathcal{C}$ of subsets
of $V(G)$ with the following guarantee: 
for every treedepth-$d$ structure $\T$ in $G$, there exists
a tree decomposition $(T,\beta)$ of $G$ such that
for every $t \in V(T)$ we have $\beta(t) \in \mathcal{C}$ and 
$|\beta(t) \cap \T| = 2^{2^{\Oh(d^3)}}$. 
\end{theorem}

The general approach is to reduce to the case of \emph{chordal bipartite graphs}.
Recall that a bipartite graph is \emph{chordal bipartite} if its every cycle of length
longer than $4$ has a chord (i.e., the only induced cycles are of length $4$). 

The class of \textsc{MWIS} problems is tractable on chordal bipartite graphs
thanks to the following result of Kloks, Liu, and Poon
(and the general algorithm of Fomin, Todinca, and Villanger~\cite{FominTV15}).
\begin{theorem}[Corollary~2 of~\cite{kloks2012feedback}]\label{thm:kloks}
A chordal bipartite graph on $n$ vertices and $m$ edges has $\Oh(n+m)$ minimal separators.
\end{theorem}

Observe that $P_7$-free bipartite graphs are ``almost'' chordal bipartite: 
they only additionally allow $C_6$ as an induced subgraph.
Our approach is to add edges to the input graph so that it becomes chordal bipartite, 
without destroying the sought solution. To this end, the following folklore characterization
will be handy.
\begin{lemma}[folklore]\label{lem:bip-chordal-char}
Let $G$ be a bipartite graph with bipartition $V_1,V_2$.
Then, $G$ is chordal bipartite if and only if for every minimal separator $S$ of $G$,
$S \cap V_1$ is complete to $S \cap V_2$ (i.e., the separator induces a complete bipartite graph, also called a \emph{biclique}).
\end{lemma}
\begin{proof}
    In one direction, assume that $G$ contains a cycle $C$ of length $\ell \geq 6$ as an induced subgraph and let $x_1,x_2,\ldots,x_\ell$ be consecutive vertices of $C$. Then,
    $\{x_2,x_3\}$ and $\{x_5,\ldots,x_\ell\}$ are anticomplete. Furthermore, any
    minimal separator separating these sets contains
    $x_1$ and $x_4$, which are nonadjacent and on the opposite bipartition sides of $G$

    In the other direction, let $S$ be a  minimal separator of $G$ with two full components $A,B$
    and vertices $x \in S \cap V_1$, $y \in S \cap V_2$, $xy \notin E(G)$.
    Observe that, thanks to the bipartiteness of $G$, a shortest path from $x$ to $y$ via $A$
    is of length at least $3$, and similarly via $B$. These two shortest path together form an induced cycle of length at least $6$. 
\end{proof}

Lemma~\ref{lem:bip-chordal-char} motivates the following process of completing
a $P_7$-free bipartite graph $G$ into a chordal bipartite graph: while $G$ contains
a minimal separator $S$ that violates the statement of Lemma~\ref{lem:bip-chordal-char},
complete $G[S]$ into a biclique. 
In Section~\ref{ss:bip-complete} we analyse this process and show that it is well-behaved
on $P_7$-free graphs, in particular, does not lead outside the class of $P_7$-free bipartite graphs.

\subsection{Completing to a chordal bipartite graph}\label{ss:bip-complete}

Throughout the rest of this section we assume that the input graph $G$ has a fixed bipartition into sets $V_1,V_2$ (this is an ordered partition).
We remark that we will often look at certain induced subgraphs of $G$ that are not necessarily connected,
but a vertex never changes its side of the bipartition.

Lemma~\ref{lem:bip-chordal-char} motivates the following definition.
Let $G$ be a bipartite graph with bipartition $V_1,V_2$, and let $S$ be a minimal separator of $G$.
We say that \emph{$S$ induces a biclique} if $S \cap V_1$ is complete to $S \cap V_2$.
The operation of \emph{completing $S$ into a biclique} turns $G$ into a graph $G+F := (V(G), E(G) \cup F)$,
where $F = \{uv~|~u \in S \cap V_1, v \in S \cap V_2, uv \notin E(G)\}$.

The following lemma is pivotal to our completion process.
\begin{lemma}\label{lem:no_new_path}
Let $G$ be a $P_7$-free bipartite graph. Let $S$ be a minimal separator and let $A$ and $B$ be two full components of $S$. Let $G+F$ be the result of completing $S$ into a biclique.
Then $G+F$ is also $P_7$-free.
\end{lemma}
\begin{proof}
    By contradiction, suppose that $G+F$ contains a $P_7$. Let $Q$ be a $P_7$ in $G + F$, which minimizes $|E(Q)\cap F|$ among all $P_7$s contained in $G+F$. 
    
    As $G$ is $P_7$-free, $E(Q) \cap F \neq \emptyset$. Consequently,
    $|V(Q) \cap S| \geq 2$.
    We consider several cases.

    \medskip

    \noindent\textbf{Case A.} 
    $|V(Q)\cap S|=2$.\\ 
    Let $V(Q) \cap S = \{x,y\}$. Note that $E(Q) \cap F = \{xy\}$, i.e., $xy \in F$
    and $xy \notin E(G)$.
    Let $Q_x$ and $Q_y$ be the components of $Q-\{xy\}$ that contain $x$ and $y$, respectively.
    
    Since $|V(Q) \setminus S| = 5$, without loss of generality assume that $Q_x$ is of length at most $2$ and $Q_y$ is of length at least $3$.
    Note that $Q_y \setminus \{y\}$ is contained in a single component of $G-S$; without loss
    of generality assume that this component is not $A$. 
    Let $R$ be a shortest path from $x$ to $y$ via $A$; note that $R$ 
    has at least $3$ edges.
    Then, the concatenation of $R$ and $Q_y$ is an induced path in $G$ on at least $7$ vertices, a contradiction.

    \medskip

    \noindent \textbf{Case B.} $|V(Q)\cap S| >2$.\\ As $S$ induces a biclique in $G+F$, we have $|V(Q)\cap S|=3$ and $V(Q) \cap S$ consists of three consecutive vertices of $Q$;
    call these vertices $x,y,z$ (in the order of appearance on $Q$, i.e., $xy,yz \in E(G) \cup F$). 
    
    As $E(Q) \cap F \neq \emptyset$, either $xy \in F$ or $yz \in F$ (or both).
    Without loss of generality, we can assume that $xy\in F$. Note that $yz$ may or may not belong to $F$.

    Let $Q_x$ and $Q_z$ be the components of $Q-\{y\}$ that contain $x$ and $z$, respectively.
    Each of $Q_x-\{x\}$ and $Q_z-\{z\}$ is contained in a single component of $G-S$.
    By the symmetry between $A$ and $B$, we assume that $Q_z-\{z\}$ is not contained
    in $A$.
    
    Let $R$ be a shortest path from $x$ to $y$ via $A$; as before, $R$ is of length at least $3$. Let $p$ be the neighbor of $x$ on $R$.
    
    In the following we consider cases depending on the position of $x$ on $Q$.

        \begin{enumerate}
            \item Vertex $x$ is the first or the second vertex of $Q$, i.e., $Q_x$ is of length at most $1$.  \\ 
            If $pz\in E(G)$, then substituting $y$ with $p$ on $Q$ yields a $P_7$ in $G$.
            If $pz\notin E(G)$, then the concatenation of $Q_z$, the edge $yz$,
            and $R-\{x\}$ is an induced path on at least $7$ vertices 
            either in $G$ (if $yz \in E(G)$) or in $G+F$ with strictly fewer edges of $F$
            than $Q$ (if $yz \in F$). In all cases, we get a contradiction.
            
            \item There are at least three vertices before $x$ on the path $Q$, i.e., 
            $Q_x$ is of length at least $3$. 
            \\ 
            If $Q_x-\{x\}$ is contained in $A$, then let $R'$ be a shortest path from $x$ to $y$ via $B$, and otherwise let $R' := R$. Then, the concatenation of $R'$ and $Q_x$
            is an induced path in $G$ on at least $7$ vertices, a contradiction.
            
            \item Vertices $x,y,z$ are the middle vertices of $Q$, i.e., both $Q_x$ and $Q_z$ are of length $2$. \\ If there exists a component $C\in \cc(G-S)$ such that $x,z\in N(C)$ and $V(Q)\cap C = \varnothing$, then we can  connect $x$ and $z$ via a shortest path $R_C$ in $C$. Then the concatenation of $Q_x$, $R_C$, and $Q_z$
            is an induced path on at least $7$ vertices, a contradiction.
            
            If there is no such component $C$, then $Q$ must be of a form $a_1 - a_2 - x-y-z-b_2- b_1$, where $a_1,a_2\in A$ and $b_1, b_2\in B$. Note that if there is a path connecting $x$ and $z$ via $A$ or $B$ with at least $3$ intermediate vertices, then we can extend this path with either $a_2$ and $a_1$ or $b_2$ and $b_1$ to get a $P_7$. Similarly, if there a path connecting $x$ and $y$ (or $y$ and $z$ if $yz\in F$) via $A$ or $B$ with at least four intermediate vertices, then we can also extend it 
            with $a_2$ and $a_1$ or $b_2$ and $b_1$ to get a $P_7$.
            Consequently, $R$ is of the form $x - p - q- y$ with $p,q \in A$.
            
            We make the following observation. If there exists $u\in A$ such that $x,z\in N(u)$ and $a_1\notin N(u)$, then a path $a_1-a_2-x-u-z-b_2-b_1$ is a $P_7$ in $G$. Thus any vertex in $A$ which is a neighbor of $x$ and $z$ is also a neighbor of $a_1$. By symmetry any vertex in $B$ which is a neighbor of $x$ and $z$ is also a neighbor of $b_1$.
            
            Suppose now that $pz\in E(G)$. By the observation above, 
            we have $pa_1\in E(G)$. If $a_2q\notin E(G)$, then a path $a_2-a_1-p-q-y-B-B$ is a $P_7$ in $G$, a contradiction (here, $y-B-B$ stands for any two-edge path from $y$ into $B$, which exists as $B$ is connected and contains at least $2$ vertices). Otherwise, a path $a_1-a_2-q-y-z-b_2-b_1$ is an induced $P_7$ in $G+F$, which has strictly fewer edges in $F$ than $Q$, a contradiction.
            
            Therefore, $pz\notin E(G)$. If $yz\in E(G)$, then a path $x-p-q-y-z-b_2-b_1$ is an induced  $P_7$ in $G$. Thus $yz\in F$. Let $R'$ be the shortest path connecting $y$ and $z$ via $B$. We have that $R' = z-p'-q'-y$ for some $p',q' \in B$.
            If $p'x\in E(G)$, then we can conclude in an analogous way to a~case if $pz\in E(G)$. Therefore we can assume that $p'x\notin E(G)$. The path $x-p-q-y-q'-p'-z$ is an induced $P_7$ in $G$ then, which is the final contradiction. 
        \end{enumerate}
        This completes the proof.
\end{proof}

In the next lemma, we establish that completing a minimal separator into a biclique in a $P_7$-free graph does not create new $C_6$s.

\begin{lemma}\label{lem:no_new_cycle}
    Let $G$ be a $P_7$-free bipartite graph. Let $S$ be a minimal separator in $G$
    and let $A$ and $B$ two full components of $S$.
    Let $G+F$ be the result of completing $S$ into a biclique. 
    Let $C$ be an induced $C_6$ in $G+F$. 
    Then, $C$ is also an induced $C_6$ in $G$.
\end{lemma}
\begin{proof}
By contradiction, suppose that $E(C)\cap F\neq \emptyset$. As $S$ in a biclique in $G+F$, $V(C)\cap S$ contains either two or three consecutive vertices of $C$.
Thus, $P := C \setminus S$ is a path and belongs to exactly one component of $G-S$.
Without loss of generality, we can assume that $P$ does not lie in $B$, i.e., $C \cap B=\varnothing$. 

Let $x,z \in V(C) \cap S$ be the two vertices of $C$ adjacent on $C$ to the vertices of $P$.
Note that $xz \notin E(G)$: either $|V(C) \cap S| = 3$ and $x,z$ are on the same biparteness
side of $G$ or $|V(C) \cap S| = 2$ and then $xz \in F$. 

Let $R$ be the shortest path from $x$ to $z$ via $B$. Then, $R \cup P$ induce a hole $C'$ in $G$.
Since $R$ is of length at least $2$, $C'$ is not shorter than $C$. Since $G$ is $P_7$-free,
$C'$ is a six-vertex hole. This can only happen if $R$ is of length $2$
and $V(C) \cap S = \{x,y,z\}$ for some $y \in S$.
Then, $C=x-y-z-r-q-p-x$, where $p, q,r$ lie in a single component of $G-S$ different than $B$. Without loss of generality, we can assume that $yz\in F$. Then we can find a shortest path $R'$, which connects $y$ and $z$ via $B$; it is of length at least $3$. Then the path $R-r-q-p$ is an induced path in $G$ on at least $7$ vertices, a contradiction.
\end{proof} 

Let $G$ be a $P_7$-free bipartite graph with fixed bipartition $V_1,V_2$, and let $\T$ be a treedepth-$d$ structure in $G$.
A tuple $(C,x,y)$ is a \emph{bad $C_6$} (with respect to $\T$) in $G$ if $C$ is an induced six-vertex cycle in $G$,
and $x \in V_1$, $y \in V_2$ are two vertices that are at the same time (a) opposite vertices of $C$, and (b)
incomparable vertices of $\T$.
That is, if one adds the edge $xy$ to $G$ (which can be done without violating the biparteness of $G$), then one breaks the treedepth-$d$ structure $\T$, i.e., $\T$ is not a treedepth-$d$ structure in $G + xy$.

We have the following simple corollary of Lemma~\ref{lem:no_new_cycle}.
\begin{lemma}\label{lem:bip_T_stays}
    Let $G$ be a $P_7$-free bipartite graph, let $\T$ be a treedepth-$d$ structure in $G$,
    let $S$ be a minimal separator in $G$, and let $G+F$ be a result of completing 
    $S$ into a biclique.
    Assume that there is no bad $C_6$ in $G$ with respect to $\T$.
    Then, $\T$ is a treedepth-$d$ structure in $G+F$ and, furthermore, 
    there is no bad $C_6$ in $G+F$ with respect to $\T$.
\end{lemma}
\begin{proof}
    Let $A$ and $B$ be two full components of $S$.     
    Let $xy \in F$ be arbitrary. Let $R_A$ and $R_B$ be shortest paths from $x$ to $y$ via $A$ and $B$,
    respectively. Then, the concatenation of $R_A$ and $R_B$ is an induced cycle $C$ in $G$;
    since $G$ is $P_7$-free, $C$ is a six-vertex cycle with $x$ and $y$ being two opposite vertices. 
    Since $(C,x,y)$ is not a bad $C_6$ w.r.t. $\T$, the addition of the edge $xy$
    does not break the treedepth-$d$ structure $\T$. Since the choice of $xy \in F$ was
    arbitrary, $\T$ is a treedepth-$d$ structure in $G+F$. 
    Furthermore, Lemma~\ref{lem:no_new_cycle} ensures that every $C_6$ in $G+F$ is also present
    in $G$; thus, $G+F$ does not contain a bad $C_6$ w.r.t. $\T$. 
\end{proof}

We conclude this section with the following enumeration. 
\begin{lemma}\label{lem:bip-alg-no-bad-C6}
There exists an algorithm that, given 
a $P_7$-free bipartite graph $G$ and an integer $d$,
runs in polynomial time and returns a
family $\mathcal{C}$ of subsets of $V(G)$ of size $\Oh(|V(G)|^5)$
with the following property:
for every treedepth-$d$ structure $\T$ in $G$ that admits no bad $C_6$,
there exists a tree decomposition $(T,\beta)$ of $G$
such that for every $t \in V(T)$, the set $\beta(t)$ belongs to $\mathcal{C}$
and $\beta(t)$ contains at most $d$ elements of $\T$.
\end{lemma}
\begin{proof}
Consider the following process. Start with $\widehat{G} := G$.
While $\widehat{G}$ is not chordal bipartite, find an induced six-vertex cycle $C$ in $G$, fix two opposite vertices $x$ and $y$ in $C$, find a minimal separator $S$ in $G$ that contains $x$ and $y$ (note that any minimal separator separating the two components of $C-\{x,y\}$ would do),
complete $S$ into a biclique obtaining $\widehat{G}+F$, and set $\widehat{G} := \widehat{G} + F$. 

Lemma~\ref{lem:no_new_path} ensures that $\widehat{G}$ stays $P_7$-free bipartite.
Lemma~\ref{lem:bip_T_stays} ensures that every treedepth-$d$ structure ${\cal T}$ in $G$ that admits no bad $C_6$ remains so in $\widehat{G}$.

By Theorem~\ref{thm:kloks}, $\widehat{G}$ has $\Oh(|V(G)|^2)$ minimal separators. 
By Theorem~\ref{thm:minsep-pmc}, $\widehat{G}$ has $\Oh(|V(G)|^5)$ potential maximal cliques
and they can be enumerated in polynomial time. We return the list of all potential
maximal cliques of $\widehat{G}$ as $\mathcal{C}$.

As shown in~\cite{DBLP:conf/soda/ChudnovskyMPPR24}, there exists a minimal chordal completion $F$ of $\widehat{G}$
such that for every maximal clique $\Omega$ of $\widehat{G}+F$, the set $\Omega \cap \T$
is contained in a single leaf-to-root path in $\T$ and thus is of size at most $d$. 
Since all these maximal cliques are enumerated in $\mathcal{C}$, any clique
tree of $\widehat{G}+F$ serves as the promised tree decomposition $(T,\beta)$. 
\end{proof}

\subsection{Cleaning}\label{ss:bip-cleaning}

In this section we define a branching step that cleans a specific part of the graph.
The step will be general enough to be applicable in many contexts.

Let $G$ be a $P_7$-free bipartite graph. 
We say that a triple $(A,B,C)$ of pairwise disjoint vertex sets of $G$ is \emph{\VV-free} if there are no two anticomplete $P_3$s of the form $A-B-C$. Note that such a configuration naturally appears in a $P_7$-free bipartite graph if $A \cup C$ is in one side on the bipartition, $B$ is in the other side, and there is an additional vertex $v$ in the same side as $B$ such that $A \subseteq N(v)$ but $C \cap N(v) = \emptyset$.

In our case, \VV-free sets appear naturally in the following context.
\begin{lemma}\label{lem:seagulls-to-P7}
Let $G$ be a $P_7$-free bipartite graph with a fixed bipartition $V_1,V_2$. 
Let $A$ and $C$ be two disjoint subsets contained in $V_i$, for some $i \in \{1,2\}$,
and let $B \subseteq V_{3-i}$.
Furthermore, assume that there exists a vertex $v \in V_{3 - i}\setminus B$ such that 
$A \subseteq N(v)$ but $C \cap N(v) = \emptyset$. Then, $(A,B,C)$ is \VV-free.
\end{lemma}
\begin{proof}
Any two anticomplete $P_3$s of the form $A-B-C$, together with $v$, induce a $P_7$ in $G$.
\end{proof}

Recall that in our setting, we have some unknown treedepth-$d$ structure $\T$ in $G$
and we are building a container for it. That is, we are happy with inserting any number
of vertices of $G$ into the container, as long as we are guaranteed that we insert only
a constant number of vertices of $\T$ along the way. 
In the case of a \VV-free triple $(A,B,C)$, we would like to simplify this
part of $G$ by filtering out vertices of $B$ that have neighbors both in $A$ and in $C$.
To achieve this goal, we will guess a set $X \subseteq A \cup B \cup C$ that
on one hand contains (in one of the branches) at most a constant number of vertices
of the fixed unknown $\T$, and on the other hand satisfies the following:
no $b \in B \setminus X$ has a neighbor both in $A \setminus X$ and in $C \setminus X$. 

To this end, we will rely on the following simple yet powerful observation. 
Observe that if $(A,B,C)$ is \VV-free and both $A \cup C$ and $B$ are independent sets
(which happen, in particular, if $A \cup C$ is on one side of the bipartition and $B$ is on the other side of the bipartition), then, for every distinct $b_1,b_2 \in B$
either $N(b_1) \cap A$ and $N(b_2) \cap A$ are comparable by inclusion
or $N(b_1) \cap C$ and $N(b_2) \cap C$ are comparable by inclusion. 
This allows to use Lemma~\ref{lem:two_orders} on subsets of $B$, with the orders
$\leq_1$ and $\leq_2$ being inclusion of the neighborhoods in $A$ and $C$, respectively. 
This observation is the engine of the following lemma
that formalizes our goal of filtering out vertices of $B$ that have neighbors both
in $A$ and in $C$. 

\begin{lemma}\label{lem:bip-cleaning}
Fix an integer $d$.
Let $G$ be a graph, and let $(A,B,C)$ be a \VV-free
triple in $G$ with both $A \cup C$ and $B$ being independent set. Then one can in polynomial time enumerate a family $\mathcal{F}$ of subsets of $A \cup B \cup C$ of 
size at most $n^{4(d+1)^2}$ with the following guarantee:
\begin{itemize}
    \item For every $X \in \mathcal{F}$, there is no $b \in B \setminus X$ 
    with both $N(b) \cap (A \setminus X)$ and $N(b) \cap (C \setminus X)$ nonempty.
    \item For every treedepth-$d$ structure $\T$ in $G$, there exists $X \in \mathcal{F}$
    such that $|X \cap \T| \leq d^2(d+1)$.
\end{itemize}
\end{lemma}
\begin{proof}
    Fix a treedepth-$d$ structure $\T$ in $G$. We will describe the process of enumerating
    the elements of $\mathcal{F}$ as a branching algorithm, guessing some properties of $\T$.
    The number of leaves of the branching will be polynomial in the size of $G$
    and in each leaf of the branching we will output one set $X$ that satisfies the second
    property for every $\T$ that agrees with the guesses made in this leaf.

    For every $b \in B$, let $\iota_A(b)$ be the maximum integer $1 \leq \alpha \leq d$
    such that $N(b) \cap A \cap \T^\alpha$ contains at least two vertices;
    $\iota_A(b) = 0$ if such a $\alpha$  does not exist.
    Similarly define $\iota_C(b)$ with respect to $N(b) \cap C \cap \T^\alpha$.
    For $0 \leq \alpha,\beta \leq d$, let $B_{\alpha,\beta} = \{b \in B~|~\iota_A(b) = \alpha \wedge \iota_C(b) = \beta\}$. 
    Note that sets $B_{\alpha,\beta}$ form a partition of $B$.

    Initialize $X = \emptyset$.   
    For every $0 \leq \alpha,\beta \leq d$, we we will guess some set of vertices and include them into $X$.
    We will argue that there will be a branch where the vertices guessed for $\alpha,\beta$ ensure that every vertex from $B_{\alpha,\beta} \setminus X$ satisfies the first statement of the lemma. Thus, for (at least) one of sets $X$ generated in the process, the statement will hold for every vertex in $B \setminus X$.
         
    Consider fixed $0 \leq \alpha,\beta \leq d$. 
    If $B_{\alpha,\beta}$ is empty, there is nothing to do, so assume otherwise.
    Define the following two orders $\leq_1$ and $\leq_2$ on $B_{\alpha,\beta}$:
    \begin{align*}
      b \leq_1 b' &\Longleftrightarrow N(b) \cap A \supseteq N(b') \cap A,  & \mathrm{if}\ \alpha=0, \\
      b \leq_1 b' &\Longleftrightarrow N(b) \cap A \cap \T^\alpha \subseteq N(b') \cap A \cap \T^\alpha &\mathrm{if}\ \alpha > 0, \\
      b \leq_2 b' &\Longleftrightarrow N(b) \cap C \supseteq N(b') \cap C,  & \mathrm{if}\ \beta=0, \\
      b \leq_2 b' &\Longleftrightarrow N(b) \cap C \cap \T^\beta \subseteq N(b') \cap C \cap \T^\beta &\mathrm{if}\ \beta > 0.
    \end{align*}
    Apply Lemma~\ref{lem:two_orders} to $B_{\alpha,\beta}$ with $\leq_1$ and $\leq_2$, 
    obtaining a vertex $b_{\alpha,\beta}^\ast$.
    
    If $\alpha=0$, guess $b_{\alpha,\beta}^\ast$ and add $N(b^\ast_{\alpha,\beta}) \cap A$ to $X$.
    This adds at most $d$ vertices of $\T$ to $X$, while adds to $X$
    all vertices of $N(b) \cap A$ such that $b \in B_{\alpha,\beta}$ and $b_{\alpha,\beta}^\ast \leq_1 b$.
        
    If $\alpha > 0$, guess two elements $a_{\alpha,\beta}^1,a_{\alpha,\beta}^2 \in N(b^\ast_{\alpha,\beta}) \cap A \cap \T^\alpha$
    and add $N(a_{\alpha,\beta}^1) \cap N(a_{\alpha,\beta}^2) \cap B$ to $X$.
    This adds at most $(\alpha-1)$ vertices of $\T$ to $X$, while adds to $X$ every
    $b \in B_{\alpha,\beta}$ such that $b^\ast_{\alpha,\beta} \leq_1 b$. 

    Perform symmetrical operation for $\beta$ and $C$:
    If $\beta=0$, guess $b_{\alpha,\beta}^\ast$ and add $N(b^\ast_{\alpha,\beta}) \cap C$ to $X$.
    This adds at most $d$ vertices of $\T$ to $X$, while adds to $X$
    all vertices of $N(b) \cap C$ such that $b \in B_{\alpha,\beta}$ and $b_{\alpha,\beta}^\ast \leq_2 b$.
    If $\beta > 0$, guess two elements $c_{\alpha,\beta}^1,c_{\alpha,\beta}^2 \in N(b^\ast_{\alpha,\beta}) \cap C \cap \T^\beta$
    and add $N(c_{\alpha,\beta}^1) \cap N(c_{\alpha,\beta}^2) \cap B$ to $X$.
    This adds at most $(\beta-1)$ vertices of $\T$ to $X$, while adds to $X$ every
    $b \in B_{\alpha,\beta}$ such that $b^\ast_{\alpha,\beta} \leq_2 b$. 

    Finally, add the resulting set $X$ to $\mathcal{F}$ if it satisfies the first bullet point
    of the statement. 

    We have already argued that in the branch when all guesses concerning $\T$ are correct,
    the first bullet point of the statement is satisfied. Furthermore, the size of $\T \cap X$ is bounded by
    \[ \sum_{\alpha=1}^d \sum_{\beta=1}^d \left( \alpha + \beta \right) = 2d \; \frac{d(d+1)}{2} = d^2(d+1).\]
    Finally, observe that the number of branches is bounded by $n^{4(d+1)^2}$
    as for every $0 \leq \alpha,\beta \leq d$ we guess at most four vertices of $G$. 
\end{proof}

Thanks to Lemma~\ref{lem:seagulls-to-P7}, Lemma~\ref{lem:bip-cleaning} is applicable in the following setting.
\begin{definition}
    Let $G$ be a $P_7$-free bipartite graph with a fixed bipartition $V_1,V_2$ and let $Z \subseteq V(G)$. 
    The \emph{neighborhood partition with respect to $Z$} is a partition $\mathcal{Z}$ of $V(G) \setminus Z$ into sets depending on (1) their side in the bipartition, and (2) their neighborhood in $Z$.
    More precisely, 
    \[ \mathcal{Z} = \{ A_{i,Y}~|~ i \in \{1,2\}, Y \subseteq Z \cap V_{3-i} \}, \]
    where
    \[ A_{i,Y} = \{v \in (V(G) \setminus Z) \cap V_i ~|~ N(v) \cap Z = Y \}. \]
\end{definition}
\begin{lemma}\label{lem:bip-cleaning-apply}
    Let $G$ be a $P_7$-free bipartite graph with a fixed bipartition $V_1,V_2$, let $Z \subseteq V(G)$,
    and let $\mathcal{Z}$ be the neighborhood partition  with respect to $Z$.
    Then, for every distinct $A,B,C \in \mathcal{Z}$, the triple $(A,B,C)$ is \VV-free. 
\end{lemma}
\begin{proof}
    There is no $P_3$ of the form $A-B-C$ unless $A$ and $C$ are on one side of the bipartition, and $B$ is on 
    the other side of the bipartition. Then, by the definition of $\mathcal{Z}$, the vertices of $A$ and the vertices of $C$
    differ in their neighborhood in $Z$: there exists $v \in Z$ such that either $N(v) \cap (A \cup C) = A$ or $N(v) \cap (A \cup C) = C$. Then, the claim follows from Lemma~\ref{lem:seagulls-to-P7} and the symmetry of $A,C$.
\end{proof}

We summarize this section with the following cleaning step. We remark that if $G$ is connected, 
the assumption of a fixed bipartition is not needed, as a connected bipartite graph has only one bipartition.
\begin{lemma}\label{lem:bip-cleaning-full}
    Let $G$ be a $P_7$-free bipartite graph with a fixed bipartition $V_1,V_2$, let $Z \subseteq V(G)$,
    and let $\mathcal{Z}$ be the neighborhood partition  with respect to $Z$.
    Then one can in time $|V(G)|^{2^{3|Z|} \cdot \Oh(d^2)}$ enumerate a family $\mathcal{F}$ of at most $n^{4(d+1)^2 2^{3|Z|}}$ subsets
    of $V(G) \setminus Z$ with the following properties:
    \begin{itemize}
    \item For every $X \in \mathcal{F}$, for every $v \in V(G) \setminus X$, the elements of $N(v) \setminus (X \cup Z)$
    are contained in a single set of $\mathcal{Z}$.
    \item For every treedepth-$d$ structure $\T$ in $G$, there exists $X \in \mathcal{F}$
    such that $|X \cap \T| \leq 2^{3|Z|} \cdot d^2(d+1)$.
    \end{itemize}
\end{lemma}
\begin{proof}
    Observe that $|\mathcal{Z}| \leq 2^{|Z|} + 1$. 
    Let $\mathcal{Z}^3$ be the family of all triples $(A,B,C)$ where $A,B,C$ are distinct elements of $\mathcal{Z}$.
    For every triple $(A,B,C) \in \mathcal{Z}^3$, we apply Lemma~\ref{lem:bip-cleaning} 
    obtaining a family $\mathcal{F}_{(A,B,C)}$;
    Lemma~\ref{lem:bip-cleaning-apply} asserts that the assumptions are satisfied. 
    For every triple $(X_{(A,B,C)})_{(A,B,C) \in \mathcal{Z}^3}$\linebreak$ \in \prod_{(A,B,C) \in \mathcal{Z}^3} \mathcal{F}_{(A,B,C)}$
    insert $\bigcup_{(A,B,C) \in \mathcal{Z}^3} X_{(A,B,C)}$ into $\mathcal{F}$.
    The promised guarantees and size bounds are immediate from Lemma~\ref{lem:bip-cleaning} and the bound
    $|\mathcal{Z}| \leq 2^{|Z|+1}$, which implies $|\mathcal{Z}^3| \leq (2^{|Z|}+1) \cdot 2^{|Z|} \cdot (2^{|Z|}-1) \leq 2^{3|Z|}$. 
\end{proof}

\subsection{Structural properties of a bad $C_6$}\label{ss:bip-bad-C6}
\newcommand{\mainrest}{\mathsf{MR}}
\newcommand{\level}{\mathsf{depth}}
\newcommand{\RichClasses}{\mathsf{Cls}}

% Let $(\T, X)$ be the solution to \textsc{$(\mathrm{td} \leq d, \phi)$-MWIS} on $G$. 
The tools developed in Section~\ref{ss:bip-complete} allow us to solve \textsc{$(\mathrm{td} \leq d, \phi)$-MWIS} on $G$, assuming
that there is no bad $C_6$ w.r.t. the solution treedepth-$d$ structure $\T$. In this section we investigate properties of bad $C_6$s.

We will often denote an induced $C_6$ as $C = c_1-c_2-\ldots-c_6-c_1$. We implicitly assume that
the indices behave cyclically, i.e., $c_7=c_1$ etc.
Then $(c_1, c_4),\ (c_2, c_5),\ (c_3, c_6)$ are pairs of the opposite vertices in $C$. 

We will need the following observation implicit in~\cite{Bonomo-Braberman21}; here we provide a proof for completeness.
\begin{lemma}\label{lem:bip-C6-bigrest}
Let $G$ be a $P_7$-free bipartite graph, let $C=c_1-c_2-\ldots-c_6-c_1$ be an induced $C_6$ in $G$,
and let $D$ be a connected component of $G-N[V(C)]$ that contains at least two vertices.
Then, for every $v \in N(D)$, the set $N(v) \cap V(C)$ equals to either $\{c_1,c_3,c_5\}$ or $\{c_2,c_4,c_6\}$.
\end{lemma}
\begin{proof}
    Let $v \in N(D)$. Since $|D| > 1$, there exists a $P_3$ of the form $v-D-D$. 
    By the definition of $D$, $v \in N(V(C))$.
    Observe that if $c_i \in N(v)$, then $c_{i+2} \in N(v)$, as otherwise $G$ contains the following $P_7$: $c_{i+3}-c_{i+2}-c_{i+1}-c_i-v-D-D$. 
    The lemma follows.
\end{proof}
For an induced six-vertex-cycle $C=c_1-c_2-\ldots-c_6-c_1$, denote 
\begin{align*}
    S_1^C &= \{v \in N(V(C))~|~N(v) \cap V(C) = \{c_1,c_3,c_5\}\}, \\
    S_2^C &= \{v \in N(V(C))~|~N(v) \cap V(C) = \{c_2,c_4,c_6\}\}.
\end{align*}
Furthermore, let $\mainrest_C$ be the union of vertex sets of all connected components of $G-N[V(C)]$ that contain
at least 2 vertices. ($\mainrest$ stands for the ``main remainder''.)
With this notation, Lemma~\ref{lem:bip-C6-bigrest} states that $N(\mainrest_C) \subseteq S_1^C \cup S_2^C$. 

Greatly simplifying, our algorithm will guess a bad $C_6$ $(C,x,y)$, resolve $N[V(C)]$ using cleaning,
and recurse on connected components of $\mainrest_C$. 
To restrict the space of possible recursive calls, we need the following observation.
\begin{lemma}\label{lem:bip-comp-nei}
Let $G$ be a $P_7$-free bipartite graph with a fixed bipartition $V_1,V_2$, and let $\mathcal{B}$ be a family of
pairwise disjoint and anticomplete subsets of $V(G)$ such that for every $B \in \mathcal{B}$,
$G[B]$ is connected and $|B| > 1$. 
Let $D$ be a connected component of $G-\bigcup_{B \in \mathcal{B}} N[B]$, let $C$ be the connected component of $G$ that contains $D$.
Then, for every $i \in \{1,2\}$, either $N(D) \cap V_i = \emptyset$ or there exists a single set $B_i \in \mathcal{B}$ such that
$N(D) \cap V_i \subseteq N(B_i)$. 
Consequently, there exists a subfamily $\mathcal{B}' \subseteq \mathcal{B}$ of size at most $2$ such that $N(D) \subseteq \bigcup_{B \in \mathcal{B}'} N[B]$
(i.e., $D$ is a connected component of $G-\bigcup_{B \in \mathcal{B}'} N[B]$). 
\end{lemma}
\begin{proof}
    Let $\mathcal{B}_i \subseteq \mathcal{B}$ be an inclusion-wise minimal set such that $\bigcup_{B \in \mathcal{B}_i} N(B) \supseteq N(D) \cap V_i$.
    Assume $|\mathcal{B}_i| > 1$; let $B,B' \in \mathcal{B}_i$ be distinct. 
    By minimality, there exists $v,v' \in N(D) \cap V_i$ such that $v \in N(B) \setminus N(B')$ and $v' \in N(B') \setminus N(B)$.
    Since $|B|,|B'|>1$, there exists an induced $P_3$ of the form $v-B-B$ and an induced $P_3$ of the form $v'-B'-B'$. 
    Let $Q$ be a shortest path from $v$ to $v'$ via $D$. Then, there exists an induced path in $G$ of the form $B-B-v-Q-v'-B'-B'$
    and this path has at least $7$ vertices, a contradiction.
\end{proof}

Let $\T$ be a treedepth-$d$ structure in $G$.
A set $A \subseteq V(G)$ is \emph{$\T$-rich} if there exists $\alpha \in [d]$ such that $|A \cap \T^\alpha| > 1$, and \emph{$\T$-poor} otherwise. 
For $\alpha \in [d]$ and $v \in \T^\alpha$, we denote by $\level(v) = \alpha$  the depth of $v$. %\prz{to wygląda, jakby mogło się przydać wcześniej}
The \emph{depth} of a bad $C_6$ $(C,x,y)$ is the pair $\level(C,x,y) := (\level(x),\level(y)) \in [d] \times [d]$.
As the depth of a bad $C_6$ is a pair, in order to compare  depths of bad $C_6$s, we introduce a total order $\prec$ on $[d] \times [d]$
as the lexicographical comparison of $(\level(x) + \level(y), \level(x))$, i.e.,
\[ (1,1) \prec (1,2) \prec (2,1) \prec (1,3) \prec (2,2) \prec (3,1) \prec \ldots \prec (d-1,d) \prec (d,d-1) \prec (d,d). \]
A bad $C_6$ is \emph{$\prec$-maximal} if there is no other bad $C_6$ whose depth is larger in the $\prec$ order.

We have the following observation.
\begin{lemma}\label{lem:bip-better-C6-in-rich-S}
Let $G$ be a $P_7$-free bipartite graph with a fixed bipartition $V_1,V_2$, and let $\T$ be a treedepth-$d$ structure in $G$.
Let $(C,x,y)$ be a bad $C_6$ with $\level(C,x,y) = (\alpha,\beta)$. 
If $S_1^C$ contains at least two vertices of $\T^{\alpha'}$ for some $\alpha' \in [d]$, then $\alpha' > \alpha$.
Similarly, if $S_2^C$ contains at least two vertices of $\T^{\beta'}$ for some $\beta' \in [d]$, then $\beta' > \beta$.

Furthermore, if both $S_1^C$ and $S_2^C$ are $\T$-rich, then there exists a bad $C_6$ in $G$ of depth $(\beta',\alpha')$ for some
$\alpha' > \alpha$ and $\beta' > \beta$.
Consequently, if $(C,x,y)$ is a $\prec$-maximal bad $C_6$, then either $S_1^C$ or $S_2^C$ is $\T$-poor (possibly both).
\end{lemma}
\begin{proof}
    Let $C = c_1-c_2-\ldots-c_6-c_1$ where $c_1=x$ and $c_4=y$. 
    Assume that $S_1^C$ contains two vertices $x',x'' \in \T^{\alpha'}$; as $xx',xx'' \in E(G)$ and $x \in \T^\alpha$, we have $\alpha' > \alpha$ and $x',x''$ are descendants of $x$ in $\T$.
    Similarly, if $S_2^C$ contains two vertices $y',y'' \in \T^{\beta'}$, then $\beta' > \beta$ and $y',y''$ are descendants of $y$
    in $\T$. This proves the first part of the lemma.

    If both of the above happen, then since $x$ and $y$ are incomparable in $\T$, $x'$ and $y'$ are also incomparable in $\T$.
    Then, $C' = x'-c_3-c_2-y'-c_6-c_5-x'$ is an induced $C_6$ and $(C',y',x')$ is a bad $C_6$ of depth $(\beta',\alpha')$.
    As $\alpha' > \alpha$ and $\beta' > \beta$, we have that $\alpha' + \beta' > \alpha + \beta$ and thus $(\alpha,\beta) \prec (\beta',\alpha')$.
\end{proof}

We will need the following nomenclature. The \emph{class} of a bad $C_6$ $(C,x,y)$ (with respect to some fixed treedepth-$d$ structure $\T$) is the pair of bits indicating whether $S_1^C$ is $\T$-poor and whether $S_2^C$ is $\T$-poor.
Let $\RichClasses$ be the set of possible classes (i.e., $|\RichClasses| = 4$). 
We will usually denote them as poor/poor, rich/poor, poor/rich, and rich/rich. 

Informally speaking, if $(C,x,y)$ is a $\prec$-maximal bad $C_6$ that is of poor/poor class with respect to $\T$
(where $(\T,X)$ is the sought solution),
then it is relatively easy to separate $N[V(C)]$ from $\mainrest_C$, as $S_1^C \cup S_2^C$ contains at most $2d$ vertices of the solution. 
However, if $S_1^C$ or $S_2^C$ is $\T$-rich, the situation requires a bit more work.
The next few lemmata investigate such a situation.

\begin{lemma}\label{lem:bip-MR-comparable}
Let $G$ be a $P_7$-free bipartite graph with a fixed bipartition $V_1,V_2$, let $Z \subseteq V(G)$ be connected,
and let $D,D'$ be two connected components of $G-N[Z]$ that are of size at least $2$ each.
Then, for every $i \in \{1,2\}$, the sets $N(D) \cap V_i$ and $N(D') \cap V_i$
are comparable by inclusion.
\end{lemma}
\begin{proof}
    Assume the contrary, let $x \in (N(D) \setminus N(D')) \cap V_i $
    and $x' \in (N(D') \setminus N(D)) \cap V_i$. 
    By the existence of $x$ and $x'$, we observe that $Z, x,x',D,D'$ lie all in the same connected component of $G$ and $x,x' \in N(Z)$. Let $Q$ be the shortest path from $x$ to $x'$ via $Z$. 
    Since $|D|,|D'| > 1$, there exists an induced $P_3$ of the form $x-D-D$ and an induced $P_3$ of the form $x'-D'-D'$.
    But then there exists an induced path of the form $D-D-x-Q-x'-D'-D'$, which has
    at least $7$ vertices, a contradiction.
\end{proof}

The most frequent usage of Lemma~\ref{lem:bip-MR-comparable} will be when $Z = V(C)$
for some bad $C_6$ $(C,x,y)$. Then, $D,D'$ are connected components of $\mainrest_C$
and Lemma~\ref{lem:bip-MR-comparable} asserts that for every $i \in \{1,2\}$, 
the neighborhoods $N(D) \cap S_i^C$ and $N(D') \cap S_i^C$ are comparable by inclusion.

\begin{lemma}\label{lem:bip-MR-nextC}
Let $G$ be a $P_7$-free bipartite graph with a fixed bipartition $V_1,V_2$.
Let $\T$ be a treedepth-$d$ structure in $G$ and let $(C,x_1,x_2)$ be a bad $C_6$ w.r.t. $\T$.
Let $(\alpha_1,\alpha_2) = \level(C,x_1,x_2)$.
Let $i \in \{1,2\}$ and assume that $S^C_{3-i}$ contains a vertex of $\T^{> \alpha_i}$.

A bad $C_6$ $(C',x_1',x_2')$ is called \emph{similar to $(C,x_1,x_2)$}
if (a) $C$ is disjoint and anticomplete to $C'$, (b) $\level(C',x_1',x_2') = \level(C,x_1,x_2)$, 
and (c) the component $D'$ of $G[\mainrest_C]$ that
contains $C'$ is adjacent to at least one vertex of $S^C_{3-i} \cap \T^{> \alpha_i}$,

Then, for every bad $C_6$ $(C',x_1',x_2')$ that is similar to $(C,x_1,x_2)$,
the following holds.
\begin{enumerate}
    \item $C'$ is anticomplete to $S^C_{3-i} \cap \T^{> \alpha_i}$.
    \item For every $u \in S^C_{3-i} \cap \T^{> \alpha_i} \cap N(D')$
    there exists $v \in D'$ that is a neighbor of $u$ and belongs to $S^{C'}_{i}$.
    \item For every other bad $C_6$ $(C'',x_1'',x_2'')$ that is similar to $(C,x_1,x_2)$,
    $C''$ is contained in $D'$.
\end{enumerate}
\end{lemma}
\begin{proof}
    As $C$ and $C'$ are disjoint and anticomplete, $C'$ lies in $\mainrest_C$ and $C$ lies in $\mainrest_{C'}$. 

    Consequently, if $u \in N(D')$ has a neighbor in $C'$, then $u \in S_1^{C'} \cup S_2^{C'}$, 
    that is, $u$ is adjacent to all three vertices of $C'$ from the same bipartition side.
    If additionally $u \in S^C_{3-i}$, then $u$ is adjacent to both $x_i$ and $x_i'$;
    as these vertices are both in $\T^{\alpha_i}$, $u$ cannot belong to $\T^{> \alpha_i}$. 
    This proves the first point. 

    Let $u \in N(D') \cap \T^{> \alpha_i} \cap S^C_{3-i}$ and let $Q$
    be a shortest path from $u$ to $V(C')$ via $D'$. 
    Let $v$ be the penultimate (i.e., just before the endpoint in $V(C')$) vertex
    of $Q$. The first point asserts that $u$ is anticomplete to $C'$, i.e., $v \neq u$.
    Since $u \in N(V(C))$ and $C \subseteq \mainrest_{C'}$,
    by Lemma~\ref{lem:bip-C6-bigrest}, $v \in S_1^{C'} \cup S_2^{C'}$. 
    Note that there exists an induced $P_3$ of the form $v - C' - C'$ and an induced $P_3$ of the form
    $u-C-C$. Then, $G$ contains an induced path $C'-C'-v-Q-u-C-C$. Since $G$ is $P_7$-free,
    $Q$ is of length one, i.e., $v u \in E(G)$.
    As $u$ and $v$ are on different sides of the bipartition, this proves the second point.

    For the third point, let $(C'',x_1'',x_2'')$ be a bad $C_6$ w.r.t. $\T$
    that is similar to $(C,x_1,x_2)$,
    but contained in a component $D'' \neq D'$ of $G[\mainrest_C]$.
    Pick any $u' \in S^C_{3-i} \cap \T^{> \alpha_i} \cap N(D')$
    and let $v' \in D'$ be the vertex whose existence is asserted in the second point
    for $(C',x_1',x_2')$.
    Pick any $u'' \in S^C_{3-i} \cap \T^{> \alpha_i} \cap N(D'')$ (possibly $u' = u''$)
    and let $v'' \in D''$ be the vertex whose existence is asserted in the second point
    for $(C'',x_1'',x_2'')$.
    Note that there exists an induced $P_4$ of the form $u'-v'-C'-C'$ and an induced $P_4$
    of the form $u''-v''-C''-C''$. By possibly connecting these two $P_4$s via $C$
    (if $u' \neq u''$), we obtain an induced path in $G$ on at least $7$ vertices, a contradiction.
    This proves the third point.
\end{proof}

\begin{lemma}\label{lem:bip-doubleC6}
Let $G$ be a $P_7$-free bipartite graph with a fixed bipartition $V_1,V_2$.
Let $\T$ be a treedepth-$d$ structure in $G$, and
let $(C,x_1,x_2)$ and $(C',x_1',x_2')$ be two disjoint and anticomplete bad $C_6$s w.r.t. $\T$
that are in the same connected component of $G$
such that $(\alpha_1,\alpha_2) = \level(C,x_1,x_2) = \level(C',x_1',x_2')$.
Furthermore, let $u \in (S_1^C \cup S_2^C) \setminus N[V(C')]$ and
$v \in (S_1^{C'} \cup S_2^{C'}) \setminus N[V(C)]$ with $uv \in E(G)$. 
Denote $Z = V(C) \cup V(C') \cup \{u,v\}$. 
Then, for every connected component $D$ of $G-N[Z]$, the following holds.
\begin{enumerate}
    \item If $|D| > 1$, then $N(D) \subseteq N(V(C) \cup V(C'))$.
    \item if additionally $D$  contains a bad $C_6$ $(C'',x_1'',x_2'')$
       with $\level(C'',x_1'',x_2'') = (\alpha_1,\alpha_2)$,
       then $N(D) \cap V_1$ and $N(D) \cap V_2$ are $\T$-poor. 
    \end{enumerate}
\end{lemma}
\begin{proof}
    Let $D$ be a connected component of $G-N[Z]$ with $|D| > 1$. 
    
    Let $w \in N(D)$. 
    By the definition of $D$, $w$ is adjacent to at least one vertex
    of $Z$.
    If $w$ is anticomplete to both $V(C)$ and $V(C')$, then 
    there exists an induced $P_5$ of the form $w-v-u-C-C$ or $w-u-v-C'-C'$. 
    If $|D| > 1$, then there exists an induced $P_3$ of the form $w-D-D$,
    resulting in a $P_7$ with the aforementioned $P_5$, a contradiction.
    Thus, $N(D) \subseteq N(V(C)) \cup N(V(C'))$, as promised in the first point.

    Let now $(C'',x_1'',x_2'')$ be as in the second point and assume $w \in N(D) \cap \T$. 
    Let $i \in \{1,2\}$ be such that $w \in V_{3-i}$. 
    By Lemma~\ref{lem:bip-C6-bigrest}, $w$ is either anticomplete to $C$ or
    $w \in S_i^C$ and, similarly, $w$ is either anticomplete to $C'$
    or $w \in S_i^{C'}$. 
    We have already established that $w \in N(V(C)) \cup N(V(C'))$, that is,
    $w \in S_i^C \cup S_i^{C'}$;
    in particular, $w$ is adjacent to either $x_i$ or $x_i'$.

    Assume first that  $w$ is adjacent to at most one vertex of the set  $\{x_i,x_i',x_i''\}$.
    Then, $w$ is adjacent to at least one vertex of $Z$,
    but is not complete to all vertices of $Z$ on the opposite side of the bipartition.
    Consequently, there exists an induced $P_4$ of the form $w-Z-Z-Z$.
    Furthermore, as $w$ is adjacent to $x_i$ or $x_i'$, we know that $w$ is nonadjacent to $x_i''$. 
    By Lemma~\ref{lem:bip-C6-bigrest}, $w$ is anticomplete to $C''$. 
    However, then a shortest path from $w$
    to $C''$ via $D$, with one more step in $C''$, gives an induced $P_4$ of the form $w-D-D-D$.
    Concatenated with the already established $P_4$ of the form $w-Z-Z-Z$ this gives a $P_7$
    in $G$, a contradiction.

    Hence, $w$ is adjacent to at least two of the vertices of
    $\{x_i,x_i',x_i''\}$. Since all these vertices are of depth $\alpha_i$
    in $\T$, $w$ is a common ancestor of all vertices of $N(w) \cap \{x_i,x_i',x_i''\}$.
    As the choice of $w \in \T \cap N(D)$ was arbitrary,
    this proves that $N(D) \cap V_1$ and $N(D) \cap V_2$ are $\T$-poor, as desired.
\end{proof}

\subsection{The algorithm}
\newcommand{\LiveClasses}{\mathsf{ActiveCls}}
\newcommand{\richclass}{\mathsf{c}}
\newcommand{\hbound}{h_\mathrm{max}}
\newcommand{\bigbound}{d_{\mathrm{max}}}
\newcommand{\hugebound}{d_{\mathrm{max},2}}
\newcommand{\poorcomps}{D_{\mathrm{poor}}}

Armed with the structural insights from the previous sections, we are ready to present the algorithm of Theorem~\ref{thm:bip}.
It will be a recursive procedure that takes as input an induced subgraph
$G[A]$ of $G$ and attempts at decomposing $G[A]$. (The recursive call will take also 
a number of additional parameters on input to guide the recursion; these will be described
shortly.) Similarly as it was done in~\cite{DBLP:conf/soda/ChudnovskyMPPR24}, we will describe the algorithm
as guessing properties of a fixed hypothetical treedepth-$d$ structure $\T$ in $G$
that we want to handle. 

We define 
\begin{align*}
\hbound & := 48(d+1)^3, \\
\bigbound &:= 2^{43} (d+1)^3. 
%\hugebound &:= \hbound^2 \cdot \bigbound^2.
\end{align*}

\paragraph{Recursion arguments and guarantee.}
A recursive call, apart from the induced subgraph $G[A]$, takes as input
two integers $\alpha_1,\alpha_2 \in [d]$, 
a subset $\LiveClasses \subseteq \RichClasses$,
and a function $\gamma : \LiveClasses \to \mathbb{N}$.
Let $h$ be the recursion depth of the recursive call in question.

The goal is to compute a family $\mathcal{C}$ of subsets of $A$
with the following guarantee for every treedepth-$d$ structure $\T$.

\medskip 
\fbox{
\parbox{0.95\textwidth}
{
\noindent \underline{\textbf{Recursion guarantee.}} Suppose the following hold:
\begin{itemize}
    \item There is no bad $C_6$ $(C,x,y)$ w.r.t. $\T$ with $V(C) \subseteq A$,
    and $(\alpha_1,\alpha_2) \prec (\level(x), \level(y))$.
    \item For every bad $C_6$ $(C,x,y)$ w.r.t. $\T$ with $V(C) \subseteq A$,
    and $(\alpha_1,\alpha_2) = (\level(x), \level(y))$,
    the class $\richclass$ of $(C,x,y)$ belongs to $\LiveClasses$. Furthermore,
    the number of vertices of $\T \setminus A$ that are adjacent to $V(C)$
    is at least $\gamma(\richclass)$.
    \item $N(A) \cap \T$ is of size at most $\bigbound^{4^h}$.
\end{itemize}
Then there exists a tree decomposition $(T,\beta)$ of $G[A]$ such that
for every $t \in V(T)$ we have $\beta(t) \in \mathcal{C}$
and $|\beta(t) \cap \T| \leq \bigbound^{4^{h+1}}$. 
}}

\medskip 

We will maintain the following progress in the recursion.
For a call on $G[A]$, $(\alpha_1,\alpha_2)$, $\LiveClasses$, and $\gamma$, 
a direct subcall on $G[A']$, $(\alpha_1',\alpha_2')$, $\LiveClasses'$, $\gamma'$ satisfies:
\begin{itemize}
    \item $(\alpha_1',\alpha_2') \prec (\alpha_1,\alpha_2)$, or
    \item $(\alpha_1',\alpha_2') = (\alpha_1,\alpha_2)$ and
    \begin{itemize}
        \item $\LiveClasses' \subsetneq \LiveClasses$, or
        \item $\LiveClasses' = \LiveClasses$, $\gamma(\richclass) \leq 12d$
        for every $\richclass \in \LiveClasses$, $\gamma'(\richclass) \geq \gamma(\richclass)$
        for every $\richclass \in \LiveClasses$, and $\gamma'(\richclass) > \gamma(\richclass)$
        for at least one $\richclass \in \LiveClasses$. 
    \end{itemize}
\end{itemize}
This progress property will always be immediate from the description of the algorithm
and thus usually not checked explicitly. 
Note that this progress property bounds the depth of the recursion by $(d+1)^2 \cdot (48d + 4) \leq \hbound$. 

\paragraph{Leaves of the recursion.}
The recursion reaches its leaf when either $G[A]$ is $C_6$-free (hence chordal bipartite)
or we have $(\alpha_1,\alpha_2) = (1,1)$ and 
$\LiveClasses = \emptyset$. In both cases, for every $\T$ that satisfies the premise of
the recursion guarantee, there is no bad $C_6$ w.r.t. $\T$. 
Hence, it suffices to invoke
the algorithm of Lemma~\ref{lem:bip-alg-no-bad-C6} to $G[A]$ and return the computed
family $\mathcal{C}$.

\paragraph{Recursion.}
Let us focus on a recursive call with input $G[A]$, $(\alpha_1,\alpha_2)$, $\LiveClasses$, and 
$\gamma$. 
If $G[A]$ is disconnected, we recurse on each connected component independently
and return the union of the returned families. Hence, we assume that $G[A]$ is connected.

Let $\T$ be a treedepth-$d$ structure in $G$ that satisfies the premise of the recursion guarantee.
Thanks to Lemma~\ref{lem:bip-better-C6-in-rich-S}, we can delete rich/rich from $\LiveClasses$ and the domain of $\gamma$, if it is present there.
If $\LiveClasses = \emptyset$, then there are no bad $C_6$s $(C,x,y)$ of level $(\alpha_1,\alpha_2)$,
so we can recurse on the immediate predecessor
of $(\alpha_1,\alpha_2)$ in the $\prec$ order, with $\LiveClasses = \RichClasses$
and $\gamma \equiv 0$, and pass on the set $\mathcal{C}$ returned by this call. 

If there exists $\richclass \in \LiveClasses$ with $\gamma(\richclass) > 12d$, proceed
as follows. 
Observe that if $(C,x,y)$ is a bad $C_6$ of level $(\alpha_1,\alpha_2)$ and of class $\richclass$
with $V(C) \subseteq A$, then there is $c \in V(C)$ that is adjacent to at least two
vertices in $N(A) \cap \T$.
Iterate over all possible choices of $N(A) \cap \T$ (there is only a polynomial number
of them thanks to the recursion guarantee premise) together with the levels of vertices
of $N(A) \cap \T$. Let $B$ be the union, over all pairs $(u,u')$ of distinct vertices
of $N(A) \cap \T$ of the same level, of the sets $N(u) \cap N(u') \cap A$. 
Recurse on $G[A \setminus B]$, $(\alpha_1,\alpha_2)$, and $\LiveClasses \setminus \{\richclass\}$,
obtaining a set $\mathcal{C}_B$. 
Finally, return the union of all sets $\{B \cup D~|~D \in \mathcal{C}_B\}$ 
constructed in this process.

Note that there is only a $(dn)^{\bigbound^{4^h}} = n^{\bigbound^{2^{\Oh(d^3)}}}$
choices of $N(A) \cap \T$ together with its levels and, for every pair $(u,u')$, 
the set $N(u) \cap N(u')$ contains at most $(d-1)$ vertices of $\T$. 
Hence, the set $B$ constructed in the branch where all information
on $N(A) \cap \T$ is guessed correctly satisfies $|B \cap \T| \leq (d-1) \binom{\bigbound^{4^h}}{2} \leq \bigbound^{4^{h+1}}$.
Consequently, in this branch the premise of the recursion guarantee for $\T$ is satisfied
for the recursive subcalls; if $(T',\beta')$ is the tree decomposition
promised for $G[A \setminus B]$ and $\T$, then 
we can construct the promised tree decomposition $(T,\beta)$ by taking $T = T'$
and $\beta(t) = \beta'(t) \cup B$ for every $t \in V(T')$. 

We are left with the case $\gamma(\richclass) \leq 12d$ for every $\richclass \in \LiveClasses$.
Pick $\richclass \in \LiveClasses$, preferably rich/poor or poor/rich
if one of those belong to $\LiveClasses$.
The algorithm will perform an internal dynamic programming routine on a carefully
chosen family of subsets of $A$. 
We say that a set $A' \subseteq A$ is \emph{well-shaped} if there is a family 
$\mathcal{B}$ consisting of at most two connected subsets of $A$, each on at least $2$
and at most $14$
vertices, such that $A'$ is a connected component of $G[A]-N[\bigcup_{B \in \mathcal{B}} B]$. 
Let $\mathcal{A}$ be the set of well-shaped sets; note that $|\mathcal{A}| \leq |A|^{29}$
and $\mathcal{A}$ can be computed in polynomial time. 

A well-shaped set $A'$ is \emph{$\T$-regular} if there exists a family $\mathcal{B}$
witnessing that $A'$ is well-shaped such that for every $B \in \mathcal{B}$
and $i \in \{1,2\}$, the set $N(A') \cap A \cap N(B) \cap V_i$ is $\T$-poor.
Note that if $A'$ is $\T$-regular, then $|N(A') \cap A \cap \T| \leq 4d$. 

Note that $A$ is $\T$-regular (in particular, $A \in \mathcal{A}$),
as witnessed by $\mathcal{B} = \emptyset$ (recall that $G[A]$ is connected).

For every $A' \in \mathcal{A}$, in the order of increasing size of $A'$,
we will compute a family $\mathcal{C}(A')$ of subsets of $A \cap N[A']$
with the following guarantee:
if $A'$ is $\T$-regular, then there exists a tree decomposition $(T,\beta)$ of $G[N[A'] \cap A]$
such that for every $t \in V(T)$ we have $\beta(t) \in \mathcal{C}(A')$
and $|\beta(t) \cap \T| \leq \bigbound^{4^{h+1}}$. 
The family $\mathcal{C}(A)$ is the desired output of the recursive call.

Lemma~\ref{lem:bip-comp-nei} ensures that the following process does not lead
outside $\mathcal{A}$.
\begin{claim}\label{cl:stay-in-A}
Let $A' \in \mathcal{A}$, let $B$ be a connected subset of $G[A']$ on at least $2$
and at most $14$ vertices, and let $D$ be a connected component of $G[A' \setminus N[B]]$.
Then, $D \in \mathcal{A}$. 
Furthermore, if $A'$ is $\T$-regular and, for every $\alpha \in \{1,2\}$, the
set $N(D) \cap A \cap N(B) \cap V_\alpha$ is $\T$-poor, then $D$ is $\T$-regular, too.
\end{claim}
\begin{proof}
    Let $\mathcal{B}$ be a family witnessing $A' \in \mathcal{A}$. 
    Let $\mathcal{B}' = \mathcal{B} \cup \{B\}$. 
    Then, $D$ is a connected component of $G[A]-\bigcup_{B' \in \mathcal{B}'} N[B']$
    and the family $\mathcal{B}'$ satisfies the requirements of Lemma~\ref{lem:bip-comp-nei} in the graph $G[A]$.
    Hence, there is a subfamily $\mathcal{B}'' \subseteq \mathcal{B}'$ of size at most $2$
    such that $D$ is a connected component of $G[A]-\bigcup_{B'' \in \mathcal{B}''} N[B'']$.
    Then, $\mathcal{B}''$ witnesses that $D \in \mathcal{A}$.

    For the second statement, observe that if $\mathcal{B}$ witnesses that $A'$ is $\T$-regular,
    then so does $\mathcal{B}''$ for $D$. 
\end{proof}
We emphasize here that the set $N(D) \cap A \cap N(B) \cap V_i$ in Claim~\ref{cl:stay-in-A} may contain
vertices outside $A'$ and our control of the number of vertices of $\T$ in this set
is very important for the correctness of the algorithm.

Fix $A' \in \mathcal{A}$.
A bad $C_6$ $(C,x,y)$ of level $(\alpha_1,\alpha_2)$ and of class $\richclass$
with $V(C) \subseteq A'$ is henceforth called a \emph{relevant $C_6$}.

Guess if $\T$ contains a relevant $C_6$.
For the ``no'' branch, recurse on $G[A']$, $(\alpha_1,\alpha_2)$, and $\LiveClasses \setminus \{\richclass\}$,
and put every returned set into the computed family $\mathcal{C}(A')$. 

For the ``yes'' branch, proceed as follows.
For a connected set $A'' \subseteq A'$, a relevant $C_6$ $(C,x,y)$ in $G[A'']$ is \emph{extremal}
if it maximizes the number of vertices in its main remainder in $G[A'']$ and,
subject to the above, it minimizes the number of vertices of $\T$ in $N[V(C)] \cap A''$. 

\paragraph{Finding a pivot set $Z$.}
We will now identify (by branching) a connected set $Z \subseteq A'$ with the following properties.
\begin{itemize}
\item $Z$ is of one of the following forms:
\begin{description}
 \item[(single $C_6$)] $V(C)$, where $(C,x,y)$ is an extremal relevant $C_6$ in $G[A']$, or
 \item[(double $C_6$)] $V(C) \cup V(C') \cup \{u,v\}$, where $(C,x,y)$ is an extremal relevant $C_6$ in $G[A']$, $(C',x',y')$ is a relevant $C_6$ that is 
 extremal in a connected component $D'$ of $G[A']-N[V(C)]$, $u \in S_i^C \cap A' \cap \T$ for some $i \in \{1,2\}$, $u \notin N(V(C'))$,
 $v \in S_{3-i}^{C'} \cap D'$, and $uv \in E(G)$.
\end{description}
\item For every connected component $F$ of $G[A']-N[Z]$ that contains a relevant $C_6$, for every $i \in \{1,2\}$, the set $N(F) \cap A \cap N(Z) \cap V_\alpha$ is $\T$-poor. 
\end{itemize}

We start by guessing an extremal relevant $C_6$ $(C,x,y)$.
Somewhat abusing the notation, let us denote its main remainder in $G[A']$ by $\mainrest_C$.

We check if $Z = V(C)$ satisfies the above properties. If this is the case, we are done, so assume otherwise. 

We say that a component $D$ of $G[\mainrest_C]$ is \emph{terrifying}
if it contains a relevant $C_6$ and additionally
$N(D) \cap A \cap N(Z) \cap V_1$ is $\T$-rich  
or $N(D) \cap A \cap N(Z) \cap V_2$ is $\T$-rich. 
Observe that $Z = V(C)$ does not satisfy the desired properties if and only there exists a terrifying component.
Note that it can happen only if $\richclass$ is not poor/poor.

Let $D_0$ be a terrifying component.
Guess $D_0$ and a relevant $C_6$ $(C',x',y')$ that is extremal in $G[D_0]$. 
Lemma~\ref{lem:bip-better-C6-in-rich-S} implies that $N(D_0) \cap A \cap S_1^C$ contains two vertices
of $\T^{\alpha_1'}$ for some $\alpha_1' > \alpha_1$ or $N(D_0) \cap A \cap S_2^C$ contains two vertices of $\T^{\alpha_2'}$ for some $\alpha_2' > \alpha_2$;
let $u$ be any of those two vertices. 
Then, Lemma~\ref{lem:bip-MR-nextC} applies to the graph
$G[A' \cup (A \cap (S_1^C \cup S_2^C))]$: 
by the first point, $u \notin N(V(C'))$,
the second point gives $v \in D_0 \cap (S_1^{C'} \cup S_2^{C'})$ with $uv \in E(G)$,
and the third point implies that $D_0$ is the only terrifying component. 

Denote $Z := V(C) \cup V(C') \cup \{u,v\}$. 
Lemma~\ref{lem:bip-doubleC6} applies to $(C,x,y)$, $(C',x',y')$, $u$, and $v$
in the graph $G[A' \cup (A \cap N[Z])]$:
every connected component $D$ of $G[A']-N[Z]$ that contains a relevant $C_6$,
the set $N(D) \cap A \cap N(Z)$ is contained in $N(V(C)) \cup N(V(C'))$ 
and both $N(D) \cap A \cap N(Z) \cap V_1$ and $N(D) \cap A \cap N(Z) \cap V_2$ are $\T$-poor.
Consequently, $Z$ satisfies the desired properties.

\paragraph{Finding a neighborhood $Y$.}
Lemma~\ref{lem:bip-MR-comparable}, applied to $Z$ in the graph $G[A' \cup (A \cap N[Z])]$ asserts
that the components of $G[A']-N[Z]$ that contain a relevant $C_6$ 
have comparable neighborhoods in $A \cap N[Z] \cap V_1$ and comparable neighborhoods
in $A \cap N[Z] \cap V_2$. 
For every $i \in \{1,2\}$, guess a component $D_i^\circ$ of $G[A']-N[Z]$
that contains a relevant $C_6$ and with maximal $N(D_i^\circ) \cap A \cap N[Z] \cap V_i$;
a valid guess is $D_i^\circ = \emptyset$ if no such component exists. 
By the properties of $Z$, we have that $N(D_i^\circ) \cap A \cap N[Z] \cap V_i$ is $\T$-poor.

Let $Y := \bigcup_{i \in \{1,2\}} N(D_i^\circ) \cap A \cap N[Z] \cap V_i$. 
We have that $|Y \cap \T| \leq 2d$ and for every component $D$ of $G[A']-N[Z]$,
if $D$ contains a relevant $C_6$, then $N(D) \cap A \cap N[Z] \subseteq Y$. 

Let $\poorcomps$ be the union of vertex sets of all components $D$ of $G[A']-N[Z]$
such that $|D| > 1$ and $N(D) \cap A \cap N[Z] \subseteq Y$.
Note that every relevant $C_6$ that is contained in $G[A']-N[Z]$ is in fact in $G[\poorcomps]$.

\paragraph{Using cleaning.}
Let $\mathcal{F}$ be the result of applying Lemma~\ref{lem:bip-cleaning-full}
to $G[A' \setminus (Y \cup \poorcomps)]$ and $Z$. 
The following claim establishes the progress of the recursion. 
\begin{claim}\label{cl:cleaning-progress}
For every $X \in \mathcal{F}$, every connected component $D$ of
$G[A' \setminus (X \cup Y \cup Z \cup \poorcomps)]$ either does not contain a relevant $C_6$,
  or every relevant $C_6$ in $D$ has at least one neighbor in $(X \cup Y \cup Z) \cap \T$. 
\end{claim}
\begin{proof}
    Let $\mathcal{Z} = \{A_{i,Q}~|~i \in \{1,2\}, Q \subseteq Z \cap V_i\}$
    be the neighborhood partition  with respect to $Z$ in $G[A' \setminus \poorcomps]$.
    Fix $X \in \mathcal{F}$ and consider a connected component $D$ of $G[A' \setminus (X \cup Y \cup Z \cup \poorcomps)]$.
    
    If $D$ is anticomplete to $Z$, then $D$ does not contain a relevant $C_6$ as $D$ is not part of $\poorcomps$. 
    Lemma~\ref{lem:bip-cleaning} ensures that $D$ is contained in the union
    of two sets of $\mathcal{Z}$, denote them $A_{1,Q_1}$ and $A_{2,Q_2}$. 
    If any of those sets is adjacent to a vertex of $Z \cap \T$, 
    then every relevant $C_6$ in $D$ has a neighbor in $Z \cap \T$, as desired.
    Otherwise, both $A_{1,Q_1}$ and $A_{2,Q_2}$ are anticomplete to $x$ and $y$,
    and, in case $Z$ has the double $C_6$ structure, also to $u$, $x'$, and $y'$. 

    In particular, $D$ is disjoint with $S_1^C \cup S_2^C$. Hence, either 
    $D \subseteq A' \setminus (\mainrest_C \cup S_1^C \cup S_2^C)$ or $D \subseteq \mainrest_C$. Consider now the first case. 

    If $D$ does not contain a relevant $C_6$, then we are done, so assume otherwise. 
    Let $(C'',x'',y'')$ be a relevant $C_6$ in $D$. Observe that $\mainrest_C$
    is contained in the main remainder of $C''$ in $G[A']$. Hence, as $(C,x,y)$ is extremal,
    these two main remainders are equal.
    
    Denote $C = c_1-c_2-\ldots-c_6-c_1$ with $c_1=x$ and $c_4=y$. 
    Observe that, unless $Q_1 = \{c_3,c_5\}$ and $Q_2 = \{c_2,c_6\}$, there exists
    an an edge of $C$ that is anticomplete to $C''$. This edge is contained in the main
    remainder of $C''$ in $G[A']$, a contradiction to the fact that $(C,x,y)$ is extremal. 
    In the other case, the fact that $(C,x,y)$ is extremal and $(C,x,y)$ and $(C'',x'',y'')$
    have equal main remainders in $G[A']$ implies that 
    $|A' \cap N[V(C'')] \cap \T| \geq |A' \cap N[V(C)] \cap \T|$.
    Observe that $x,y \in N[V(C)] \setminus N[V(C'')]$. 
    Hence, there are at least two vertices of $\T$ in $A' \cap \T \cap (N[V(C')] \setminus N[V(C)])$.
    Since $D \subseteq A_{1,Q_1} \cup A_{2,Q_2} \subseteq N(V(C))$, 
    these two vertices of $\T$ necessarily belong to $X \cup Y \cup Z$, as desired. 
    This completes the case $D \subseteq A' \setminus (\mainrest_C \cup S_1^C \cup S_2^C)$. 

    We are left with the remaining case $D \subseteq \mainrest_C$. 
    If $Z$ is of single $C_6$ structure, then we are done, as $D \cap \poorcomps = \emptyset$
    and thus $D$ contains no relevant $C_6$. 
    Assume then $Z$ is of double $C_6$ structure. 
    
    Assume first that $D \cap D_0 = \emptyset$.
    Recall that $u \in Q_1 \cup Q_2$ as otherwise
    any relevant $C_6$ in $D$ has a neighbor in $Z \cap \T$ as desired.
    Hence, $D \cap N[Z] = \emptyset$. Then, as $D$ is disjoint with $\poorcomps$,
    $D$ does not contain any relevant $C_6$. 

    In the second case, $D \subseteq D_0$ as $D$ is disjoint with $S_1^C \cup S_2^C$. 
    Since $x',y',u \notin Q_1 \cup Q_2$, $D \subseteq D_0 \cap N(V(C')) \setminus (S_1^{C'} \cup S_2^{C'})$ or $D$ is part of the main remainder of $C'$ in $D_0$. 
    In the latter case, the properties of $Z$ ensure that $D$ is contained in $G[A']-N[Z]$ and,
    as $D \cap \poorcomps = \emptyset$, $D$ contains no relevant $C_6$. 

    In the former case, let $C' = c_1'-c_2'-\ldots-c_6'-c_1'$
    and assume there exists a relevant $C_6$ $(C'',x'',y'')$ in $D$. 
    Observe that the main remainder of $C'$ in $G[D_0]$ is contained in the main remainder
    of $C''$. Since $(C',x',y')$ is extremal in $G[D_0]$, these main remainders are equal.
    Observe that unless $Q_1 \cap V(C') = \{c_3', c_5'\}$ and $Q_2 \cap V(C') = \{c_2',c_6'\}$,
    there is an edge of $C'$ that is anticomplete to $C''$, a contradiction to the
    fact that the main remainders of $C'$ and $C''$ in $G[D_0]$ are equal.
    Otherwise, since $(C',x',y')$ is extremal, there are at least two elements of $\T$
    in $D_0 \cap (N(V(C'')) \setminus N(V(C')))$. Since $D \subseteq N(V(C'))$
    (if $Q_1 \cap V(C') = \{c_3', c_5'\}$ and $Q_2 \cap V(C') = \{c_2',c_6'\}$), 
    these two elements of $\T$ necessarily are in $X \cup Y \cup Z$, as desired. 
    This completes the case analysis and the proof of the claim.
\end{proof}

\paragraph{Wrapping up.}
We guess $Z$ and $Y$ as described above, compute $\poorcomps$ and 
invoke Lemma~\ref{lem:bip-cleaning-full} to $Z$ in $G[A' \setminus (Y \cup \poorcomps)]$,
obtaining a family $\mathcal{F}$. 
Claim~\ref{cl:stay-in-A} ensures that every connected component of $G[\poorcomps]$ belongs to $\mathcal{A}$ and, furthermore, if $A'$ is $\T$-regular, then so is every connected
component of $G[\poorcomps]$.
If this is the case, let $(T_D,\beta_D)$ be the promised tree decomposition for 
a component $D$ of $G[\poorcomps]$. 
Insert into $\mathcal{C}(A')$ 
all sets of $\mathcal{C}(D)$ for every
connected component $D$ of $G[\poorcomps]$.

For every $X \in \mathcal{F}$, we proceed as follows.
First, insert $X \cup Y \cup Z \cup (N(A') \cap A)$ into $\mathcal{C}(A')$. 
Then, recurse on every connected component
$D$ of $G[A' \setminus (X \cup Y \cup Z \cup \poorcomps)]$
with the same parameters $(\alpha_1,\alpha_2)$, $\LiveClasses$, and $\gamma$
except for the value $\gamma(\richclass)$ increased by one,
obtaining a family $\mathcal{C}'(D)$.
For every set $B \in \mathcal{C}'(D)$ from the returned family,
insert $X \cup Y \cup Z \cup (N(A') \cap A) \cup B$ into the set $\mathcal{C}(A')$. 
This completes the description of the algorithm.

Observe that if $A'$ is $\T$-regular,
then for the set $X \in \mathcal{F}$ that satisfies 
the promise of Lemma~\ref{lem:bip-cleaning-full} for 
$\T$, we have 
\begin{equation}\label{eq:bip-margin-bound}
|(X \cup Y \cup Z \cup (N(A') \cap A)) \cap \T| \leq 2^{3 \cdot 14} d^2(d+1) + 2d + 14 + 4d < 2^{43} (d+1)^3 = \bigbound.
\end{equation}
Hence, for that set the premise of the recursion guarantee is satisfied,
and we have the promised tree decomposition $(T_D, \beta_D)$ of $G[D]$.

Construct the desired tree decomposition $(T,\beta)$ as follows. 
Start with the root bag $(X \cup Y \cup Z \cup (N(A') \cap A))$.
For every connected component $D$ of $G[\poorcomps]$, attach $(T_D,\beta_D)$
as a subtree with its root adjacent to the root bag of $(T,\beta)$.
For every connected component $D$ of $G[A' \setminus (X \cup Y \cup Z \cup \poorcomps)]$,
denote $\beta_D'(t) = \beta_D(t) \cup X \cup Y \cup Z \cup (N(A') \cap A)$ for $t \in V(T_D)$
and attach $(T_D,\beta_D')$ as a subtree with its root adjacent to the root bag of $(T,\beta)$.
It is immediate that $(T,\beta)$ is a tree decomposition of $G[N[A'] \cap A]$;
the bound on $|\T \cap \beta(t)|$ for any $t \in V(T)$ follows from~\eqref{eq:bip-margin-bound}.

%\DeclareMathOperator*{\VV}{\bigvee\mkern-15mu\bigvee}
%$\VV$-free

\section{Proof of Theorem~\ref{thm:mainalgo}}\label{sec:wrapup}
\newcommand{\subcomps}{\mathcal{C}}

\newcommand{\lmset}{\{\!\{}
\newcommand{\rmset}{\}\!\}}
\newcommand{\Aa}{\mathcal{A}}
\newcommand{\lbl}{\mathsf{label}}
\newcommand{\Multi}{\mathsf{Multi}}

In this section we use Theorems~\ref{thm:to-bip} and~\ref{thm:bip}
to prove Theorem~\ref{thm:mainalgo}.
Fix an integer $d > 0$. 

On the very high level, the algorithm works as follows.
We invoke Theorem~\ref{thm:to-bip} on the input graph $G$, obtaining the family $\mathcal{F}$.
Let $\T$ be the sought solution and let $S$ be a $\T$-avoiding minimal separator of $G$.
Let $\Psi = (K,\mathcal{D},L) \in \mathcal{F}$ be the element promised for $S$ by Theorem~\ref{thm:to-bip}.
Note that $S \subseteq W(\Psi) := K \cup \bigcup \mathcal{D}$. 
We would like to first understand the components of $\mathcal{D}$ and then use $W(\Psi)$ as a container for $S$. 
Luckily, every $D \in \mathcal{D}$ is somewhat easier than the whole graph $G$.
Let $\subcomps(\Psi)$ be the family of all connected components of $G[L(D)]$ and all connected components of $G[D \setminus L(D)]$;
recall that they are modules of $G[D]$.
Every connected component of $\subcomps(\Psi)$ has clique number strictly smaller than $G$ (so we can recurse on them),
while the quotient graph after contracting every component $\subcomps(\Psi)$ in $G[D]$ is 
bipartite and thus handled by Theorem~\ref{thm:bip}.

Unfortunately, we do not know how to handle the aforementioned recursion into $\subcomps(\Psi)$ using only the language of families of containers or carvers
and we will need to do it using the language of ``partial solutions to \cmsotwo formulae''. 
To this end, we will use the handy notion of a \emph{threshold automata} as in~\cite{DBLP:conf/soda/ChudnovskyMPPR24}.

\paragraph{Threshold automata.}
Let $\Sigma$ be a finite alphabet. 
A \emph{$\Sigma$-labelled forest} is a rooted forest $F$ where every $v \in V(F)$
as a label $\lbl(v) \in \Sigma$. 

We use the notation $\lmset\cdot\rmset$ for defining multisets. For a multiset $X$ and an integer $\tau\in \N$, let $X\wedge \tau$ be the multiset obtained from $X$ by the following operation: for every element $e$ whose multiplicity $k$ is larger than $2\tau$, we reduce its multiplicity to the unique integer in $\{\tau+1,\ldots,2\tau\}$ with the same residue as $k$ modulo $\tau$. For a set $Q$, let $\Multi(Q,\tau)$ be the family of all multisets
that contain elements of $Q$, and each element is contained at most $2\tau$ times.

\begin{definition}[\cite{DBLP:conf/soda/ChudnovskyMPPR24}]
 A {\em{threshold automaton}} is a tuple $\Aa=(Q,\Sigma,\tau,\delta,C)$, where:
 \begin{itemize}[nosep]
  \item $Q$ is a finite set of states;
  \item $\Sigma$ is a finite alphabet;
  \item $\tau\in \N$ is a nonnegative integer called the {\em{threshold}};
  \item $\delta\colon \Sigma\times \Multi(Q,\tau)\to Q$ is the transition function; and
  \item $C\subseteq \Multi(Q,\tau)$ is the accepting condition.
 \end{itemize}
 For a $\Sigma$-labelled forest $F$, the {\em{run}} of $\Aa$ on $F$ is the unique labelling $\xi\colon V(F)\to Q$ satisfying the following property for each $x\in V(F)$:
 $$\xi(x)=\delta(\lbl(x),\lmset\xi(y)\colon y\textrm{ is a child of }x\rmset\wedge \tau).$$
 We say that $\Aa$ {\em{accepts}} $F$ if 
 $$\lmset\xi(z)\colon z\textrm{ is a root of }F\rmset\wedge \tau \in C,$$
 where $\xi$ is the run of $\Aa$ on $F$.
\end{definition}

Recall that for a \cmsotwo formula with one vertex set variable $\phi$,
the $(\phi, \td \leq d)$-\textsc{MWIS} problem,
given a vertex-weighted graph $G$, asks for a pair $(\Sol,X)$
with $X \subseteq \Sol \subseteq V(G)$ such that $G[\Sol]$ is of treedepth at most $d$,
$\phi(X)$ is satisfied in $G[\Sol]$, and the weight of $X$ is maximized. 

Let $(\Sol, X)$ be such that $X \subseteq \Sol \subseteq V(G)$ and
$\td(G[\Sol]) \leq d$ and let $\T$ be a treedepth-$d$ structure containing $\Sol$. 
Let $\lbl$ be a labeling of $V(\T)$ that labels $v \in V(\T)$ with 
the following information: (1) its depth in $\T$, (2) which of its ancestors
in $\T$ it is adjacent to, and (3) whether $v \in X$ and/or $v \in \Sol$. 
As discussed in~\cite{DBLP:conf/soda/ChudnovskyMPPR24}, the check 
whether $\phi(X)$ is satisfied in $G[\Sol]$ can be done
with a threshold automaton run on $(\T,\lbl)$.
More precisely, there exists a threshold automaton $\Aa$
(depending only on $d$ and $\phi$) such that $\phi(X)$ is satisfied in $G[\Sol]$
if and only if $\Aa$ accepts $\T$ with labeling $\lbl$. 
Thus, instead of focusing on a \cmsotwo formula $\phi$, we
can fix a threshold automaton $\Aa$ and look for a
treedepth-$d$ structure $\T$ and $X \subseteq \Sol \subseteq \T$
such that $\Aa$ accepts $\T$ with the aforementioned labeling
and, subject to the above, the weight of $X$ is maximized. 
In what follows, the labeling $\lbl$ will be implicit and not mentioned. 

\paragraph{Partial solutions.}
A treedepth-$d$ structure $\T$ is \emph{neat} if for every non-root node $v$
with a parent $u$, $u$ is adjacent to at least one descendant of $v$ (possibly $v$ itself).
A simple (cf.~\cite{DBLP:conf/soda/ChudnovskyMPPR24}) argument shows that every treedepth-$d$ structure
can be turned into a neat one without increasing the depth of any vertex.

Let $\preceq_1$ be the following quasi-order on sets
of vertices of a graph $G$: $X \preceq_1 Y$ if $|X| > |Y|$ or
$|X| = |Y|$ and $X$ precedes $Y$ lexicographically. 
In~\cite{DBLP:conf/soda/ChudnovskyMPPR24}, a quasi-order $\preceq$ on treedepth-$d$ structures in
$G$ was introduced: to compare two treedepth-$d$ structures $\T_1$ and $\T_2$,
the following sequence of sets is compared with $\preceq_1$ and the first
differing one determines whether $\T_1 \preceq \T_2$: first the set of all vertices,
then the set of all depth-1 vertices (i.e., roots), then the
set of all depth-2 vertices, \ldots, the set of all depth-$d$ vertices. 
Although two distinct treedepth-$d$ structures $\T_1,\T_2$ can satisfy both
$\T_1 \preceq \T_2$ and $\T_2 \preceq \T_1$, the quasi-order $\preceq$
is a total order on the set of neat treedepth-$d$ structures~\cite{DBLP:conf/soda/ChudnovskyMPPR24}.

A \emph{partial solution} in a graph $G$ is a tuple $(\T,X,\Sol)$
where $\T$ is a treedepth-$d$ structure in $G$ and $X \subseteq \Sol \subseteq V(\T)$.
As in~\cite{DBLP:conf/soda/ChudnovskyMPPR24}, we extend $\preceq$ to comparing partial solutions:
$(\T,X,\Sol) \preceq (\T',X',\Sol')$ if:
\begin{enumerate}[nosep]
    \item the weight of $X$ is larger than the weight of $X'$, or
    \item the weights of $X$ and $X'$ are equal, but $X \preceq_1 X'$;
    \item $X = X'$, but $\Sol \preceq_1 \Sol'$;
    \item $X = X'$ and $\Sol = \Sol'$, but $\T \preceq \T'$.
\end{enumerate}
We will be looking for a $\preceq$-minimal solution $(\T,X,\Sol)$
to our $(\phi,\td \leq d)$-\textsc{MWIS} problem on $G$.
As shown in~\cite{DBLP:conf/soda/ChudnovskyMPPR24}, $\T$ is then maximal and neat.

A partial solution $(\T', X', \Sol')$ is an \emph{extension} of a partial solution $(\T, X, \Sol)$ if $\T$ is an induced subgraph of $\T'$, every root of $\T$ is a root of $\T'$, and $X' \cap \T = X$ and $\Sol' \cap \T = \Sol$. 
An extension $(\T',X',\Sol')$ of $(\T,X,\Sol)$ is \emph{neat} if 
for every non-root node $v \in V(\T')$ if $v \notin V(\T)$, then 
there is an edge between the parent of $v$ and some descendant of $v$ in $\T'$. 

For a fixed threshold automaton $\Aa = (Q,\Sigma,\tau,\delta,C)$,
a \emph{multistate assignment} of a partial solution $(\T,X,\Sol)$
is a function $\xi : \{\emptyset\} \cup V(\T) \to \Multi(Q,\tau)$. 
We say that an extension $(\T',X',\Sol')$ of $(\T,X,\Sol)$ 
\emph{evaluates} to $\xi$ if for the run $\xi'$
of $\Aa$ on $(\T',X',\Sol')$ (with the aforementioned natural labeling of $\T'$)
it holds that
        \[ \lmset \xi'(z)~|~z\textrm{ is a root of }\T'\textrm{ but not of }\T \rmset \wedge \tau_\Aa = \xi(\emptyset), \]
         and, for every $v \in V(\T)$,
        \[ \lmset \xi'(z)~|~z \textrm{ is a child of }v \textrm{ in } \T' \textrm{ but not in }\T\rmset \wedge \tau_\Aa = \xi(v). \]
Intuitively, $\xi$ contains all the information
on how $\T' \setminus \T$ contributes to the run of $\Aa$ on $(\T',X',\Sol')$. 

We will be computing set families $\mathcal{F}$ as follows. 
\begin{definition}
Let $d,\widehat{d},\widetilde{d} > 0$ be integers
and $\Aa$ be a threshold automaton.
For a graph $G$ and a set $D \subseteq V(G)$, we say
that a family $\mathcal{F}$ of subsets of $D$
is a \emph{defect-$\widetilde{d}$ extension-$\widehat{d}$ treedepth-$d$
  container family} 
if for every partial solution $(\T,X,\Sol)$ in $G$ with $|V(\T)| \leq \widehat{d}$
and for every multistate assignment $\xi$ of $(\T,X,\Sol)$,
if $(\T',X',\Sol')$ is the $\preceq$-minimal extension of $(\T,X,\Sol)$
among extensions that satisfy (1) $V(\T') \setminus V(\T) \subseteq D$ and (2)
$(\T',X',\Sol')$ evaluates to $\xi$, then
if $(\T',X',\Sol')$ is actually a neat extension of $(\T,X,\Sol)$, then
there exists a tree decomposition $(T,\beta)$ of $G[D]$
such that for every $t \in V(T)$ there exists $A \in \mathcal{F}$
with $\beta(t) \subseteq A$ and $|A \cap \T'| \leq \widetilde{d}$. 

Furthermore, the family $\mathcal{F}$ is \emph{exact} if
additionally $\beta(t) = A$ in the above. 
\end{definition}

The dynamic programming algorithm of~\cite[Section~3]{DBLP:conf/soda/ChudnovskyMPPR24} can be expressed
as the following algorithmic lemma that we will use as a subroutine.
\begin{lemma}[\cite{DBLP:conf/soda/ChudnovskyMPPR24}]\label{lem:dp-algo}
Let $d,\widehat{d},\widetilde{d} > 0$ be fixed integers
and $\Aa$ be a fixed threshold automaton.
There exists a polynomial-time algorithm that takes on input a graph $G$, a set $D \subseteq V(G)$,
and a family $\mathcal{F}$ that is a defect-$\widetilde{d}$ extension-$\widehat{d}$ treedepth-$d$
container family 
and outputs, 
for every partial solution $(\T,X,\Sol)$ in $G$ with $|V(\T)| \leq \widehat{d}$
and for every multistate assignment $\xi$ of $(\T,X,\Sol)$,
the $\preceq$-minimal extension $(\T',X',\Sol')$ of $(\T,X,\Sol)$ 
among extensions that evaluate to $\xi$ and $V(\T') \setminus V(\T) \subseteq D$, 
provided that this is a neat extension.
\end{lemma}

The family of Theorem~\ref{thm:bip} is stronger, and we capture it with 
the following definition.
\begin{definition}
Let $d,\widetilde{d} > 0$ be integers.
For a graph $G$, we say
that a family $\mathcal{F}$ of subsets of $V(G)$
is a \emph{defect-$\widetilde{d}$ treedepth-$d$
  container family} 
if for every treedepth-$d$ structure $\T$ in $G$
there exists a tree decomposition $(T,\beta)$ of $G$
such that for every $t \in V(T)$ there exists $A \in \mathcal{F}$
with $\beta(t) \subseteq A$ and $|A \cap \T'| \leq \widetilde{d}$. 

Furthermore, the family $\mathcal{F}$ is \emph{exact} if
additionally $\beta(t) = A$ in the above. 
\end{definition}

In both definitions of a container family we often omit ``treedepth-$d$''
if the constant $d$ is clear from the context.

Neatness is crucial for the following step.
\begin{lemma}\label{lem:use-neatness}
Let $d > 0$ be fixed. Let $G$ be a graph. 
Let $F \subseteq D \subseteq V(G)$. 
Let $(\T,X,\Sol)$ be a partial solution in $G$ and
let $(\T',X',\Sol')$ be a treedepth-$d$ structure in $G$ that is a neat extension of $\T$
with $V(\T') \setminus V(\T) \subseteq D$.

Let $\T_F$ consist of $\T$ and ancestors
 of all nodes of $\T'$ in $D \cap N(F)$.
 Let $\T_F'$ constist of $\T$ and ancestors
 of all nodes of $\T'$ in $F \cup (D \cap N(F))$. 
 Then, $(\T_F', X' \cap V(\T_F'), \Sol' \cap V(\T_F'))$ is a neat extension
 of $(\T_F, X' \cap V(\T_F), \Sol' \cap V(\T_F))$, 
 and $V(\T_F') \setminus V(\T_F) \subseteq F$.

 Furthermore, if $\Aa$ is a threshold automaton,
 $\xi$ is the multistate assignment that $(\T',X',\Sol')$ evaluates
 to in $(\T,X,\Sol)$ and $\xi_F$ is the multistate assignment 
that $(\T_F', X' \cap V(\T_F'), \Sol' \cap V(\T_F'))$
evaluates to in $(\T_F, X' \cap V(\T_F), \Sol' \cap V(\T_F))$,
 then if $(\T',X',\Sol')$ is the $\preceq$-minimum extension
 of $(\T,X,\Sol)$ among extensions that evaluate to $\xi$ and satisfy
 $V(\T') \setminus V(\T) \subseteq D$, then 
$(\T_F', X' \cap V(\T_F'), \Sol' \cap V(\T_F'))$
is the $\preceq$-minimum extension of $(\T_F, X' \cap V(\T_F), \Sol' \cap V(\T_F))$
among extensions 
that evaluate to $\xi_F$ and satisfy $V(\T_F') \setminus V(\T_F) \subseteq F$. 
\end{lemma}
 \begin{proof}
 Let $v \in V(\T')$ and $u \in V(\T')$ be such that $u,v \notin V(\T)$,
 $u,v \notin N(F) \cap D$, and $u$ is the parent of $v$ in $\T'$. 
 We claim if exactly one of $u,v$ is in $F$, then $u,v \in V(\T_F)$. 

 Since $(\T',X',\Sol')$ neatly extends $(\T,X,\Sol)$,
 there exists a path $P$ from $v$ to $u$ with internal vertices
 being the descendants of $v$ in $\T'$ (so, in particular, $P$ is disjoint
     with $\T$). 
 If exactly one of $u,v$ is in $F$, then $P$ needs to contain an edge
 $xy$ with $x \in F$ and $y \notin F$. As $P$ is disjoint with $\T$,
 $V(P) \subseteq D$. Hence, $y \in N(F) \cap D$. 
 Since $y$ is a descendant of $v$ in $\T'$, we have $u,v \in V(\T_F)$, as desired.
 
 We infer that if $v \in V(\T_F') \setminus V(\T_F)$, then all descendants
 of $v$ in $\T'$ are also in $V(\T_F') \setminus V(\T_F)$. 
 Also, if additionally $v \in F$, then all descencendants of $v$ in $\T'$ are in $F$.
 Hence, $(\T_F', X' \cap V(\T_F'), \Sol' \cap V(\T_F'))$ is an extension
 of $(\T_F, X' \cap V(\T_F), \Sol' \cap V(\T_F))$
 and $V(\T_F') \setminus V(\T_F) \subseteq F$.
 Neatness is immediate from the definition.
 
 For the last claim, assume that $(\T_F'', X'', \Sol'')$ is a $\preceq$-smaller
 extension of $(\T_F, X' \cap V(\T_F), \Sol' \cap V(\T_F))$ that evaluates to $\xi_F$
 and with $V(\T_F'') \setminus V(\T_F) \subseteq F$. 
 Then, one easily observes that if we replace $\T_F' \setminus \T_F$ with $\T_F'' \setminus \T_F$
and similarly for $X'$ and $\Sol'$ in $(\T',X',\Sol')$ we obtain an extension
of $(\T,X,\Sol)$ that evaluates to the same $\xi$ but
that is $\preceq$-smaller, a contradiction.
 \end{proof}

We now prove a lemma that handles components $D$ as in the last case ``$D \in \mathcal{D}$''
of Theorem~\ref{thm:to-bip}, using Theorem~\ref{thm:bip} and recursion. 

\begin{lemma}\label{lem:handle-bip}
    Let $d,\widehat{d} \geq 1$ be fixed integers.
    
    Let $G$ be a graph, $D \subseteq V(G)$, and $D = L \uplus R$ be a partition, 
    let $\mathcal{D}$ be the family of connected components of $G[L]$
    and the connected components of $G[R]$. 
    Assume that every $F \in \mathcal{D}$ is a module of $G[D]$
    and let $H$ be the quotient graph of $G[D]/\mathcal{D}$. 

    Furthermore, assume we are given a defect-$d_1$ exact decomposition family $\mathcal{F}_H$ of $H$ 
    and, for every $F \in \mathcal{D}$, we are given a defect-$d_2$ extension-$(d^2+\widehat{d})$ decomposition
    family $\mathcal{F}_{F}$ of $G[D']$. 

    Then, there exists a defect-$((d+1)d_1+d_2)$ extension-$\widehat{d}$
    decomposition family $\mathcal{F}$ of $D$ in $G$ of size
    at most $(1+|\mathcal{F}_H|)^{d_1} + \sum_{F \in \mathcal{D}} |\mathcal{F}_{F}|$. 
    Furthermore, it can be computed in time polynomial in the input and output size,
    given $G$, $D$, $L \uplus R$, $\mathcal{F}_H$, and $(\mathcal{F}_{F})_{F \in \mathcal{D}}$. 

    Moreover, if all families $\mathcal{F}_{F}$ for $F \in \mathcal{D}$ are exact,
    then $\mathcal{F}$ is exact, too.
\end{lemma}

\begin{proof}
Fix $(\T,X,\Sol)$, $\xi$, and $(\T',X',\Sol')$ as in the definition of
a extension-$\widehat{d}$ decomposition family of $D$ in $G$.
As in the definition, assume that $(\T',X',\Sol')$ is a neat extension
of $(\T,X,\Sol)$. 

Observe that $\mathcal{D}$ is a partition of $D$.  
Let 
\[\mathcal{D}_{\T'} = \{F \in \mathcal{D}~|~ F \cap \T' \neq \emptyset\} .\]
As $H[\mathcal{D}_{\T'}]$ is isomorphic to an induced subgraph of $\T'$, it is of treedepth at most $d$. 
Consequently, by the promise of $\mathcal{F}_H$, there exists a tree decomposition $(T_H,\beta_H)$ of $H$
such that for every $t \in V(T_H)$ we have $\beta_H(t) \in \mathcal{F}_H$ and $|\beta_H(t) \cap \mathcal{D}_H| \leq d_1$. 

Let 
\[ \mathcal{D}^\mathrm{big}_{\T'} = \{F \in \mathcal{D}~|~|F \cap \T'| > d\} \subseteq \mathcal{D}_{\T'}. \]
As $\T'$ cannot contain $K_{d+1,d+1}$, we have the following.
\begin{claim}\label{cl:handle-bip:big}
For every $F \in \mathcal{D}^\mathrm{big}_{\T'}$, $\bigcup_{F' \in N_H(F)} F'$ contains at most $d$ elements of $\T$.
In particular, $\mathcal{D}^\mathrm{big}_{\T'}$ is an independent set in $H$.
\end{claim}

Observe that $(T_H,\beta')$ where $\beta'(t) = \bigcup_{F \in \beta_H(t)} F$ is a tree decomposition of $G[D]$.
However, $\beta'(t)$ may contain unbounded number of elements of $\T'$ if $\beta_H(t)$ contains an element of $\mathcal{D}^\mathrm{big}_{\T'}$.
We now modify $(T_H,\beta_H)$ to isolate elements of $\mathcal{D}^\mathrm{big}_{\T'}$, so that we can use the families $\mathcal{F}_{F}$ to handle them. 

As a first step, for every $F \in \mathcal{D}^\mathrm{big}_{\T'}$,
we drag all nodes of $N_H(F)$ to be present in every bag where $F$ is present, too.
More precisely, we define $\beta_H^1 : V(T_H) \to 2^{V(H)}$ as
\[ \beta_H^1(t) := \beta_H(t) \cup \bigcup_{F \in \mathcal{D}^\mathrm{big}_{\T'} \cap \beta_H(t)} N_H(F). \]
To see that $(T_H,\beta_H^1)$ is a tree decomposition of $H$, note that for every $F \in \mathcal{D}^\mathrm{big}_{\T'}$ and
  $F' \in N_H(F)$, there exists a node $t$ with $F,F' \in \beta(t)$, and thus $\{t\in V(T_H)~|~F' \in \beta_H^1(t)\}$ is connected.
Furthermore, as $\bigcup_{F' \in N_H(F)} F'$ contains at most $d$ elements of $\T$ due to Claim~\ref{cl:handle-bip:big}
and every bag $\beta_H(t)$ contains at most $d_1$ elements of $\mathcal{D}_{\T'}$, we have
  \[ \forall_{t \in V(T_H)} \left|\bigcup_{F \in \beta_H^1(t) \setminus \mathcal{D}^\mathrm{big}_{\T'}} F \cap V(\T')\right| \leq  (d+1)d_1. \]

As a second step, for every $F \in \mathcal{D}^\mathrm{big}_{\T'}$, we fix one node $t_F' \in V(T_H)$ with $F \in \beta_H^1(t_F')$ and add a new degree-1 node $t_F$ to $T_H$, adjacent to $t_F'$. 
Let $T_H^2$ be the resulting tree. For $t \in V(T_H^2)$, define
\[ \beta_H^2(t) := \begin{cases} \beta_H^1(t) \setminus \mathcal{D}^\mathrm{big}_{\T'} & \mathrm{if}\ t \in V(T_H)\\ \beta_H^1(t_F') &\mathrm{if}\ t=t_F\ \mathrm{for\ some\ }F \in \mathcal{D}^\mathrm{big}_{\T'}. \end{cases} \]
It is immediate that $(T_H^2,\beta_H^2)$ is a tree decomposition of $H$ and for every $t \in V(T_H^2)$
\[ \beta_H^2(t) \cap \mathcal{D}^\mathrm{big}_{\T'} = \begin{cases} \emptyset & \mathrm{if}\ t \in V(T_H) \\
    \{F\} & \mathrm{if}\ t=t_F\ \mathrm{for\ some\ }F \in \mathcal{D}^\mathrm{big}_{\T'}. \end{cases} \]
 Furthermore, as every bag of $(T_H^2,\beta_H^2)$ is a subset of some bag of $(T_H,\beta_H^1)$, we have
 \begin{equation}\label{eq:handle-bip:2bound}
   \forall_{t \in V(T_H^2)} \left|\bigcup_{F \in \beta_H^2(t) \setminus \mathcal{D}^\mathrm{big}_{\T'}} F \cap V(\T')\right| \leq  (d+1)d_1. 
 \end{equation}

 Fix $F \in \mathcal{D}^\mathrm{big}_{\T'}$.
 Let $\T_F$ consist of $\T$ and ancestors
 of all nodes of $\T'$ in $D \cap N(F)$; note that $|V(\T_F)| \leq d^2 + \widehat{d}$
 as $|D \cap N(F) \cap V(\T')| \leq d$.
 Let $\T_F'$ constist of $\T$ and ancestors
 of all nodes of $\T'$ in $F \cup (D \cap N(F))$. 
 Lemma~\ref{lem:use-neatness} asserts that 
 $(\T_F', X' \cap V(\T_F'), \Sol' \cap V(\T_F'))$ is a neat extension
 of $(\T_F, X' \cap V(\T_F), \Sol' \cap V(\T_F))$, 
 and $V(\T_F') \setminus V(\T_F) \subseteq F$.
 
 Let $\xi_F$ be the multistate assignment such that
 $(\T_F', X' \cap V(\T_F'), \Sol' \cap V(\T_F'))$ as an extension of
 $(\T_F, X' \cap V(\T_F), \Sol' \cap V(\T_F))$ evaluates to. 
 Lemma~\ref{lem:use-neatness} asserts also 
 that $(\T_F', X' \cap V(\T_F'), \Sol' \cap V(\T_F'))$ is the $\preceq$-minimum
 extension of $(\T_F, X' \cap V(\T_F), \Sol' \cap V(\T_F))$
 among extensions that evaluate fo $\xi_F$ and 
 satisfy $V(\T_F') \setminus V(\T_F) \subseteq F$.

Therefore, the promise for $\mathcal{F}_F$ applies.
That is, there exists a tree decomposition $(T_F,\beta_F)$ for $(\T_F',X' \cap V(\T_F'), \Sol' \cap V(\T_F'))$:
 for every $t \in V(T_D)$ there exists $A \in \mathcal{F}_F$ such that
  (1) $\beta_F(t) \subseteq A$, and (2) $|A \cap \T_F'| \leq d_2$. 
  By the definition of $\T_F'$, the last point also implies
  $|A \cap \T'| \leq d_2$.

We now construct a tree decomposition $(T,\beta)$ of $G[D]$ as follows. 
To construct $T$, start with $T_H^2$ and, for every $F \in \mathcal{D}^\mathrm{big}_\T$, replace $t_F$ with $T_F$, i.e., 
   remove $t_F$ and put $T_F$ instead, with its root adjacent to $t_F'$.
Define $\beta$ as follows:
\[ \beta(t) := \begin{cases} \bigcup_{F \in \beta_H^2(t)} F & \mathrm{if}\ t \in V(T_H) \\
                \left( \bigcup_{F' \in \beta_H^2(t_F')} F' \right) \cup \beta_F(t) & \mathrm{if}\ t \in V(T_F)\ \mathrm{for\ some\ }F \in \mathcal{D}^\mathrm{big}_\T.\end{cases} \]

A direct check shows that $(T,\beta)$ is a tree decomposition of $G[D]$. 

We output the family $\mathcal{F}$ consisting of the following sets:
\begin{itemize}
 \item For every $A \in \mathcal{F}$ and $Z \subseteq A$ that is independent in $H$ and is of size at most $d_1$, the set 
    \[ \left(\bigcup_{F \in A \setminus Z} F\right) \cup  \left(\bigcup_{F \in Z} \bigcup_{F' \in N_H(F)} F' \right). \]
    Note that every bag $\beta(t)$ for $t \in V(T_H)$ is of the above form and there are at most $(1+|\mathcal{F}_H|)^{d_1}$ such sets.
    Furthermore, by~\eqref{eq:handle-bip:2bound}, every such set contains at most
    $(d+1)d_1$ elements of $\T'$. 
 \item For every $F \in \mathcal{D}$ and $B \in \mathcal{F}_F$, the set 
    \[ \left( \bigcup_{F' \in N_H(F)} F' \right) \cup B. \]
    Note that every bag $\beta(t)$ for $t \in V(T_F)$ for some $F \in \mathcal{D}^\mathrm{big}_\T$ is contained in a set of the above form
    (or exactly equal to a set of the above form if the families $\mathcal{F}_F$ are exact),
    there are at most $\sum_{F \in \mathcal{D}} |\mathcal{F}_F|$ such sets,
    and every such set contains at most $d + d_2$ elements of $\T'$.
\end{itemize}
 This finishes the proof of the lemma.
\end{proof}

We are now ready to state the technical statement that we will prove
by induction on the clique number $k$.
\begin{theorem}\label{thm:syf}
For every fixed integers $d,\widehat{d},k > 0$ and a threshold
automaton $\Aa$, there exists a constant $\widetilde{d}$
and an algorithm 
that takes on input a vertex-weighted $P_7$-free graph $G$ and a set $D \subseteq V(G)$
with $\omega(G[D]) \leq k$,
runs in polynomial time and computes an family $\mathcal{F}$ of subsets
of $D$ that is a defect-$\widetilde{d}$ extension-$\widehat{d}$ treedepth-$d$
container family for $D$.
\end{theorem}

\begin{proof}
Let $d,\widehat{d},k > 0$ be fixed. We will proceed by induction on $k$. 
In the base case, $k = 1$, so $G[D]$ is edgeless and
$\mathcal{F} = \{\{v\}~|~v \in D\}$ satisfies the requiments with $\widetilde{d}=1$. 
Assume then $k > 1$. 

Let $d_1$ be the bound on $|K \cap V(\T)|$ of Theorem~\ref{thm:to-bip}
for the fixed $d,k$.
Let $\widetilde{d}'$ be the constant from the inductive assumption for treedepth $d$,
clique number $k-1$, and extension $\widehat{d}+dd_1+d^2$. 

We invoke Theorem~\ref{thm:to-bip} on $G[D]$, obtaining a family $\mathcal{F}'$.
For every $\Psi := (K,\mathcal{D},L) \in \mathcal{F}'$ and 
every $D' \in \mathcal{D}$ proceed as follows.
Let $\mathcal{C}(D')$ be the family of connected components of either
$G[L(D')]$ or of $G[D' \setminus L(D')]$.
For every $F \in \mathcal{C}(D')$, recurse
with the same treedepth $d$, clique number $k-1$, and extension $\widehat{d}+dd_1+d^2$,
obtaining a defect-$\widetilde{d}'$ extension-$(\widehat{d}+dd_1+d^2)$ treedepth-$d$
family $\mathcal{F}_F$. 
Invoke Theorem~\ref{thm:bip} on the quotient graph $G[\bigcup \mathcal{D}]/\bigcup_{D' \in \mathcal{D}} \mathcal{C}(D')$, obtaining
a family $\mathcal{F}^{\mathrm{bip}}$. 
Invoke Lemma~\ref{lem:handle-bip} in $G$, $\bigcup \mathcal{D}$, partition $(\bigcup_{D' \in \mathcal{D}} L(D')) \uplus (\bigcup_{D' \in \mathcal{D}} D' \setminus L(D'))$,
the family $\mathcal{F}^{\mathrm{bip}}$ and families $(\mathcal{F}_F)_{D' \in \mathcal{D},F \in \mathcal{C}(D')}$,
obtaining a family $\mathcal{F}''$.
Note that $\mathcal{F}''$ is a defect-$\widetilde{d}''$ extension-$(\widehat{d}+dd_1)$
treedepth-$d$ container family for $\bigcup \mathcal{D}$ in $G$ for some constant $\widetilde{d}''$.

Construct a family $\mathcal{F}^\mathrm{sep}$ as follows. 
Iterate over all partial solutions $(\T_\Psi,X_\Psi,\Sol_\Psi)$ in $G$ with $|V(\T_\Psi)| \leq \widehat{d}+dd_1$ and their multistate assignments $\xi_\Psi$. 
Set $\mathcal{F}_{\Psi,\T} = \{ V(\T_\Psi) \cup A~|~A \in \mathcal{F}''$. 
For every choice of $(\T_\Psi,X_\Psi,\Sol_\Psi)$ and $\xi_\Psi$, use $\mathcal{F}_{\Psi,\T}$ 
and the algorithm of~\cite[Section~3]{DBLP:conf/soda/ChudnovskyMPPR24} (Lemma~\ref{lem:dp-algo}) to compute in polynomial time
the $\preceq$-minimum extension $(\T'_\Psi,X'_\Psi,\Sol'_\Psi)$ of $(\T_\Psi,X_\Psi,\Sol_\Psi)$ 
among extensions that evaluate to $\xi$ and satisfy $V(\T'_\Psi) \setminus V(\T_\Psi) \subseteq \bigcup \mathcal{D}$. 
Insert $K \cup (\bigcup \mathcal{D} \setminus V(\T'_\Psi))$ into $\mathcal{F}^\mathrm{sep}$.

Let $(\T,X,\Sol)$ be a partial solution with $|V(\T)| \leq \widehat{d}$,
let $\xi$ be a multistage assignment of $(\T,X,\Sol)$, and let 
$(\T',X',\Sol')$ be the $\preceq$-minimum extension of $(\T,X,\Sol)$ 
amongs ones that 
evaluate to $\xi$ and satisfy $V(\T') \setminus V(\T) \subseteq D$. 
Assume $(\T',X',\Sol')$ is a neat extension of $(\T,X,\Sol)$. 

Let $\T''$ be any maximal treedepth-$d$ structure that contains
$V(\T') \cap D$. Let $F''$ be any $\T''$-aligned minimal chordal
completion of $G[D]$. Let $S$ be a $\T''$-avoiding
minimal separator of $G[D]+F''$ with full components $A,B$. 
Let $\Psi = (K,\mathcal{D},L) \in \mathcal{F}'$ be the tuple
for $S,A,B$ promised by Theorem~\ref{thm:to-bip}. 
Set $F := \bigcup\mathcal{D}$ and use Lemma~\ref{lem:use-neatness}
to $G$, $D$, $F$, $(\T,X,\Sol)$ and $(\T',X',\Sol')$. 
We infer that $\mathcal{F}^\mathrm{sep}$ contains 
$K \cup (F \setminus V(\T'))$. As $K$ is guaranteed to contain $V(\T') \cap S$,
  we have that $S \subseteq K \cup (F \setminus V(\T'))$. 
That is, $\mathcal{F}^\mathrm{sep}$ contains a set that on one hand
contains $S$, and on the other hand contains at most $d_1 + \widetilde{d}''$
elements of $\T''$. 

Let $c = c(k)$ be the constant from Corollary~\ref{cor:pmc-cover} for $t=7$ and $k$. 
We set $\mathcal{F}$ to be the union of $\{N[v]~|~v \in V(G)\}$
and the set of all possible unions at most $c$ elements of $\mathcal{F}^\mathrm{sep}$.
Corollary~\ref{cor:pmc-cover} implies that for every
maximal clique $\Omega$ of $G[D]+F''$ there exists $A \in \mathcal{F}$
with $\Omega \subseteq A$ and $|A \cap V(\T')| \leq c(d_1+\widetilde{d}'')$. 
Hence, a clique tree $(T,\beta)$ of $G[D]+F''$ satisfies the requirements
of the theorem for $\widetilde{d} = c(d_1+\widetilde{d}'')$. 
This finishes the proof.
\end{proof}

Theorem~\ref{thm:mainalgo} follows now from the combination
of Theorem~\ref{thm:syf} for $\widehat{d} = 0$ and $D = V(G)$
and the algorithm of~\cite[Section~3]{DBLP:conf/soda/ChudnovskyMPPR24} (Lemma~\ref{lem:dp-algo}).
Note that for the case $\widehat{d}$, 
any $\preceq$-minimum extension $(\T',X',\Sol')$ of $(\emptyset,\emptyset,\emptyset)$
among extensions that evaluate to some fixed multistate assignment $\xi$ satisfies that $\T'$ is maximal and
neat, as proven in~\cite{DBLP:conf/soda/ChudnovskyMPPR24}.

\section{Conclusion}
Our work suggests some directions for future research.
Of course the most ambitious goal is to show a~polynomial-time algorithm for 
$(\mathrm{tw} \leq d,\psi)$-MWIS in $P_t$-free graphs, for all fixed $t$,
with the first open case for $t=7$ (with no further assumptions on the graph).

However, there are several intermediate goals that also seem interesting.
For example, we think that the idea of considering graphs of bounded clique number is quite promising.
Not only it allows us to use some strong structural tools like $\chi$-boundedness or VC-dimension, but also to measure the progress of an algorithm by decreasing clique number (see also~\cite{DBLP:journals/siamdm/ChudnovskyKPRS21,DBLP:journals/ipl/BrettellHMP22}).

\begin{problem}[{\cite[Problem 9.2]{DBLP:journals/siamdm/ChudnovskyKPRS21}}]
Show that for every fixed $t$ and $k$, MWIS is polynomial-time solvable in $P_t$-free graphs with clique number at most $k$.
\end{problem}

As illustrated in Section~\ref{sec:bipartite}, assuming additionally that a graph is bipartite might lead to an easier, but still interesting problem.
Of course in this setting asking for an algorithm for MWIS makes little sense, but already the next case is far from trivial.

\begin{problem}
Show that for every fixed $t$, given a vertex-weighted bipartite $P_t$-free graph, in polynomial-time we can find an induced forest of maximum possible weight.
\end{problem}

Recall that the polynomial-time algorithm for $t=7$ follows not only from our work, but also (via a very different approach) from the work of Lozin and Zamaraev~\cite{DBLP:journals/ejc/LozinZ17} -- their structural theorem about $P_7$-free bipartite graphs have bounded mim-width. However, this is no longer the case for $P-8$-free bipartite graphs, as shown by Brettell et al.~\cite{DBLP:journals/jgt/BrettellHMPP22}.

A different subclass of $P_t$-free graphs are $(P_t,K_{1,k})$-free graphs
(i.e., excluding additionally a star with $k$ leaves as an induced subgraph),
studied by Lozin and Rautenbach~\cite{DBLP:journals/ipl/LozinR03},
who show that \textsc{Maximum Weight Independent Set}
is polynomial-time solvable in these graph classes.
Does this result generalize to $(\mathrm{tw} \leq d,\psi)$-MWIS?

Finally, let us point out that Abrishami, Chudnovsky, Pilipczuk, Rzążewski, and Seymour~\cite{Tara} did not only provide an algorithm for $(\mathrm{tw} \leq d,\psi)$-MWIS in $P_5$-free graphs, but they actually considered a richer class of graphs. In particular, their algorithm works for $C_{>4}$-free graphs, i.e., graph with no induced cycles of length more than 4 (analogously we define $C_{>t}$-free graphs for any $t$).
Similarly, the quasipolynomial-time algorithm for $(\mathrm{tw} \leq d,\psi)$-MWIS works for $C_{>t}$-free graphs, for any $t$.
Note that $C_{>t}$-free graphs are a proper superclass of $P_t$-free graphs.

We believe that all polynomial-time results for $P_t$-free graphs, discussed in this paper, can be actually lifted to $C_{>t}$-free graphs.

%\paragraph{Acknowledgements.} We are grateful to Michał Pilipczuk for inspiring discussions on the topic.
%Part of the work was performed during STRUG: Stuctural Graph Theory Bootcamp, funded by the ``Excellence initiative -- research university (2020-2026)'' of University of Warsaw.

\bibliographystyle{plain}
\bibliography{bibliography}

\end{document}